\documentclass[11pt]{article}
\pdfoutput=1
\usepackage{graphicx} %
\usepackage{booktabs} %
\usepackage[LGR, T1]{fontenc}
\usepackage[utf8]{inputenc}
\usepackage{url}            %
\usepackage{booktabs}       %
\usepackage{amsfonts}       %
\usepackage{nicefrac}       %
\usepackage{microtype}      %
\usepackage{float}
\usepackage{multicol}
\usepackage{multirow}
\usepackage{siunitx}
\usepackage{tikz}
\usepackage{pgf}
\usepackage{mleftright}
\usepackage{etoolbox}
\usepackage{contour}
\usepackage{pifont}

\usepackage{fullpage}

\usepackage{subcaption}
\usepackage{tcolorbox}

\usepackage{natbib}

\usepackage{amsmath}
\usepackage{amsthm}
\usepackage{amssymb}
\usepackage{mdframed}
\usepackage{enumitem}
\setlist[itemize]{noitemsep, topsep=0pt}

\usepackage[makeroom]{cancel}

\usepackage[linesnumbered,ruled,lined,noend]{algorithm2e}

\usepackage{bm}

\usepackage{thm-restate}
\usepackage{lipsum}

\usepackage{mathtools}

\usepackage{diagbox}

\usetikzlibrary{calc,shapes,arrows,positioning,arrows.meta,matrix}

\usepackage{xcolor}

\usepackage{hyperref}
\definecolor{royalBlue}{HTML}{057DCD}
\definecolor{darkGreen}{HTML}{2E8B57}
\definecolor{mgreen}{RGB}{160, 200, 140}
\hypersetup{
	colorlinks = true,
	linkcolor = royalBlue,
	citecolor = darkGreen,
	linktocpage = true,
	urlcolor = mgreen
}

\theoremstyle{plain}
\newtheorem{theorem}{Theorem}[section]
\newtheorem{lemma}[theorem]{Lemma}
\newtheorem{corollary}[theorem]{Corollary}
\newtheorem{proposition}[theorem]{Proposition}

\newtheorem{assumption}[theorem]{Assumption}

\theoremstyle{definition}
\newtheorem{definition}[theorem]{Definition}

\newtheorem{remark}[theorem]{Remark}
\newtheorem{example}[theorem]{Example}

\newcommand*{\bbQ}{{\mathbb{Q}}}

\newcommand*{\bbR}{{\mathbb{R}}}

\newcommand*{\cA}{{\mathcal{A}}}
\newcommand*{\cE}{{\mathcal{E}}}
\newcommand*{\cM}{{\mathcal{M}}}
\newcommand*{\cJ}{{\mathcal{J}}}
\newcommand*{\cK}{{\mathcal{K}}}
\newcommand*{\eps}{{\epsilon}}
\newcommand*{\cI}{{\mathcal{I}}}
\newcommand*{\cB}{{\mathcal{B}}}

\newcommand*{\cF}{{\mathcal{F}}}

\newcommand{\cX}{\mathcal{X}}
\newcommand{\cY}{\mathcal{Y}}
\newcommand{\cD}{{\mathcal{D}}}

\let\poly\relax

\let\co\relax

\DeclareMathOperator{\poly}{poly}
\DeclareMathOperator{\reg}{Reg}
\DeclareMathOperator{\linswap}{LinSwapReg}
\DeclareMathOperator{\co}{co}

\renewcommand{\vec}[1]{\bm{#1}}
\newcommand*{\mat}[1]{\mathrm{#1}}

\newcommand*{\range}[1]{[\kern-1mm[#1]\kern-.9mm]}

\makeatletter
\tikzset{
  fitting node/.style={
      inner sep=0pt,
      fill=none,
      draw=none,
      reset transform,
      fit={(\pgf@pathminx,\pgf@pathminy) (\pgf@pathmaxx,\pgf@pathmaxy)}
    },
  reset transform/.code={\pgftransformreset}
}
\tikzset{cross/.style={path picture={
          \draw[black]
          (path picture bounding box.south east) -- (path picture bounding box.north west) (path picture bounding box.south west) -- (path picture bounding box.north east);
        }}}
\tikzstyle{ox}=[semithick,draw=black,circle,cross,inner sep=1.2mm]
\makeatother

\newcommand{\declarecolor}[2]{\definecolor{#1}{RGB}{#2}\expandafter\newcommand\csname #1\endcsname[1]{\textcolor{#1}{##1}}}
\declarecolor{White}{255, 255, 255}
\declarecolor{Black}{0, 0, 0}
\declarecolor{Maroon}{128, 0, 0}
\declarecolor{Coral}{255, 127, 80}
\declarecolor{Red}{182, 21, 21}
\declarecolor{LimeGreen}{50, 205, 50}
\declarecolor{DarkGreen}{0, 100, 0}
\declarecolor{Navy}{0, 0, 128}

\NewDocumentCommand{\numberthis}{om}{%
  \IfNoValueTF{#1}{%
    \refstepcounter{equation}\tag{\theequation}%
  }{%
    \tag{#1}%
  }%
  \label{#2}%
}

\newcommand{\timehat}[1]{^{(#1)}}

\newtoks\mymathaccents
\mymathaccents={%
\let\^\timehat
}

\everymath=\expandafter\expandafter\expandafter{%
  \expandafter\the\expandafter\everymath\the\mymathaccents}

\everydisplay=\expandafter\expandafter\expandafter{%
  \expandafter\the\expandafter\everydisplay\the\mymathaccents}

\def\[#1\]{%
\begin{flalign*}#1\end{flalign*}%
}

\usepackage{dsfont}

\newcommand{\vc}{\vec{c}}
\newcommand{\vs}{\vec{s}}
\newcommand{\vx}{\vec{x}}
\newcommand{\vy}{\vec{y}}
\newcommand{\vu}{\vec{u}}

\newcommand{\vlam}{\vec{\lambda}}

\newcommand{\vzero}{\vec{0}}

\newcommand{\ie}{\emph{i.e.},~}
\newcommand{\eg}{\emph{e.g.},~}

\DeclareMathOperator*{\E}{\mathbb{E}}

\usetikzlibrary{fit,shapes.misc,arrows.meta,calc,positioning}
\makeatletter
\tikzset{
  fitting node/.style={
      inner sep=0pt,
      fill=none,
      draw=none,
      reset transform,
      fit={(\pgf@pathminx,\pgf@pathminy) (\pgf@pathmaxx,\pgf@pathmaxy)}
    },
  reset transform/.code={\pgftransformreset}
}
\tikzset{cross/.style={path picture={
          \draw[black]
          (path picture bounding box.south east) -- (path picture bounding box.north west) (path picture bounding box.south west) -- (path picture bounding box.north east);
        }}}
\tikzstyle{ox}=[semithick,draw=black,circle,cross,inner sep=1.2mm]
\pgfdeclarelayer{background}
\pgfdeclarelayer{middle}
\pgfsetlayers{background,middle,main}

\makeatother

\tikzset{cross/.style={path picture={
          \draw[black]
          (path picture bounding box.south east) -- (path picture bounding box.north west) (path picture bounding box.south west) -- (path picture bounding box.north east);
        }}}
\tikzstyle{chanode}   = [fill=white,draw=black,circle,cross,inner sep=.8mm]
\tikzstyle{pl1node}   = [fill=black,draw=black,circle,inner sep=.55mm]
\tikzstyle{pl2node}   = [fill=white,draw=black,circle,inner sep=.55mm]
\tikzstyle{termina}   = [fill=white,draw=black,inner sep=.6mm]
\tikzstyle{decpt}     = [fill=black,draw=black,inner sep=.8mm]
\tikzstyle{obspt}     = [fill=white,draw=black,cross,inner sep=0.95mm]
\tikzstyle{highlight} = [line width=1.99]
\tikzstyle{infoset} = [black!50!white]

\newcommand{\cQ}{\ensuremath{\mathcal{Q}}}
\newcommand{\cP}{\ensuremath{\mathcal{P}}}
\newcommand{\cC}{\ensuremath{\mathcal{C}}}

\newcommand{\cH}{\ensuremath{\mathcal{H}}}

\newcommand{\mA}{{\mat A}}
\newcommand{\mB}{{\mat B}}

\newcommand{\mI}{{\mat I}}
\newcommand{\mM}{{\mat M}}

\newcommand{\mU}{{\mat U}}

\newcommand{\lintrans}[2]{\mathrm{Lin}(#1, #2)}

\newcommand{\vb}{\vec{b}}

\newcommand{\va}{\vec{a}}
\newcommand{\vz}{\vec{z}}

\let\olabel\label
\NewDocumentCommand \constraint {o} {%
  {\refstepcounter{equation}\mathrm{(\theequation)}\IfValueT{#1}{\olabel{#1}}}
}

\DeclareMathOperator*{\argmin}{arg\,min}

\DeclareMathOperator{\conv}{conv}

\newcommand*\circled[1]{\refstepcounter{equation}\mathrm{(\theequation)}}

\NewDocumentCommand  \Pure           {o} {{\Pi\IfNoValueF{#1}{_{#1}}}}
\NewDocumentCommand  \PureSub       {om} {\Pi_{\IfNoValueF{#1}{#1,\,}\succeq\,#2}}
\NewDocumentCommand  \Seqf           {o} {{\cQ\IfNoValueF{#1}{_{#1}}}}
\NewDocumentCommand  \SeqfSub       {om} {\cQ_{\IfNoValueF{#1}{#1,\,}\succeq\,#2}}
\NewDocumentCommand  \Seqs          {so} {{\Sigma\IfBooleanT{#1}{^*}\IfNoValueF{#2}{_{#2}}}}
\NewDocumentCommand  \SeqsSub       {om} {\Sigma_{\IfNoValueF{#1}{#1,\,}\succeq\,#2}}
\NewDocumentCommand  \Infos          {o} {{\cI\IfNoValueF{#1}{_{#1}}}}
\NewDocumentCommand  \DecNodes       {o} {{\cJ\IfNoValueF{#1}{_{#1}}}}
\NewDocumentCommand  \ObsNodes       {o} {{\cK\IfNoValueF{#1}{_{#1}}}}

\NewDocumentCommand  \Hist           {o} {\cH\IfNoValueF{#1}{_{#1}}}

\NewDocumentCommand  \emptyseq        {} {{\varnothing}}

\newcommand{\p}[1]{\left(#1\right)}

\newcommand{\set}[1]{\left\{#1\right\}}

\newcommand{\norm}[1]{\left\|#1\right\|}
\newcommand{\inp}[1]{\left \langle #1 \right \rangle}

\newcommand{\vp}{\vec{p}}
\newcommand{\vell}{\vec{\ell}}

\newcommand{\lt}{\vell_t}
\newcommand{\pt}{\vp_t}

\newcommand{\vol}[1]{\operatorname{vol}\left( #1 \right)}

\newcommand*{\shellelpsd}{\textsc{ShellEllipsoid}}
\newcommand*{\shellproj}{\textsc{ShellProj}}

\newcommand{\csphere}{\overline{\cB}_d}
\newcommand{\capPhi}{\overline{\Phi}}
\newcommand{\PhiFP}{\Phi_\text{FP}}
\newcommand{\PhiShell}{\Tilde{\Phi}}

\SetKwInput{KwInput}{Input}
\SetKwInput{KwOutput}{Output}

\usepackage{cleveref}
\usepackage{wrapfig}

\title{Efficient Learning and Computation of\\Linear Correlated Equilibrium in General Convex Games}
\author{
    Constantinos Daskalakis\\
    MIT\\
    \texttt{costis@csail.mit.edu}\\
    \and
    Gabriele Farina\\
    MIT\\
    \texttt{gfarina@mit.edu}\\
    \and
    Maxwell Fishelson\\
    MIT\\
    \texttt{maxfish@mit.edu}\\
    \and
    Charilaos Pipis\\
    MIT\\
    \texttt{chpipis@mit.edu}\\
    \and
    Jon Schneider\\
    Google Research\\
    \texttt{jschnei@google.com}
}
\date{}

\begin{document}
\pagenumbering{gobble} 

\maketitle

\begin{abstract}
    We propose efficient no-regret learning dynamics and ellipsoid-based methods for computing {\em linear correlated equilibria}---a relaxation of correlated equilibria and a strenghtening of coarse correlated equilibria---in general convex games. These are games where the number of pure strategies is potentially exponential in the natural representation of the game, such as extensive-form games.  Our work identifies linear correlated equilibria as the tightest known notion of equilibrium that is computable in polynomial time and is efficiently learnable for general convex games. Our results are enabled by a generalization of the seminal framework of \cite{Gordon08:No} for $\Phi$-regret minimization, providing extensions to this framework that can be used even when the set of deviations $\Phi$ is intractable to separate/optimize over. Our polynomial-time algorithms are similarly enabled by extending the Ellipsoid-Against-Hope approach of \cite{Papadimitriou2008:Computing} and its generalization to games of non-polynomial type proposed by~\cite{Farina2024:eah}. We provide an extension to these approaches when we do not have access to the separation oracles required by these works for the dual player.
\end{abstract}

\tableofcontents
\newpage
\pagenumbering{arabic} 
\setlength{\parindent}{1em} %
\setlength{\parskip}{.5em}  %

\section{Introduction}

Equilibrium computation is essential for studying and predicting the outcome of strategic behavior in multi-agent interactions, making it a fundamental problem in Game Theory, Economics, Artificial Intelligence, and a number of other fields studying multi-agent systems. As such, a natural question to consider is this: What notions of equilibrium can be computed in polynomial time or efficiently learned via learning dynamics, and in what classes of games?

The answer to this question crucially depends on the structure of the game at hand. For instance, in normal-form games, a standard model in which each player acts once and simultaneously, it is well-known that a Nash equilibrium can be efficiently computed in two-player zero-sum games, but efficient computation beyond this setting is intractable~\citep{Daskalakis2009:The,chen2009settling,Rubinstein2016:Settling}. In contrast, relaxations of the Nash equilibrium, such as correlated equilibrium and coarse correlated equilibrium---including welfare-maximizing ones---can be computed in polynomial time in the representation of the normal-form game using linear programming, no matter the number of players or strategies available to them. In succinct games of polynomial type with the polynomial expectation property, a class of games defined by~\cite{Papadimitriou2008:Computing} which includes graphical games~\citep{KearnsLS01},  polymatrix games~\citep{howson1972equilibria}, congestion games~\citep{rosenthal1973class} and others, a correlated  equilibrium can be computed efficiently using the Ellipsoid-Against-Hope algorithm~\citep{Papadimitriou2008:Computing,jiang2015polynomial}, even though welfare-maximizing correlated and coarse correlated equilibria are intractable in this setting~\citep{Papadimitriou2008:Computing,barman2015finding}. Finally, in succinct games of non-polynomial type, such as extensive-form (\emph{i.e.}, tree-form) games, a coarse correlated equilibrium can be computed efficiently, but it remains a major open question in the field whether a correlated equilibrium can be computed efficiently (see e.g.~\citep{Papadimitriou2008:Computing,vonStengel2008,huang2008computing,Farina2024:eah} and their references), despite some surprising recent progress on this front, for approximate correlated equilibria of constant approximation accuracy~\citep{dagan2024external,peng2024fast}.

All the classes of games mentioned so far are examples of \emph{convex games}, a rich class in which the strategy set of each player is an arbitrary
convex and compact set $\cP_i \subset \mathbb{R}^{d_i}$ and the utility of each player is a multi-linear function for all players' strategies; see Section~\ref{sec:linear swap, linear correlated, convex games} for a formal definition. As mentioned, convex games encompass normal-form, extensive-form
games and all the other classes of games mentioned thus far, but they go well beyond these to include many other fundamental settings such as routing
games, resource allocation problems, and  competition between firms.

Similarly, all relaxations of the Nash equilibrium mentioned so far are instances of the general class of \emph{$\Phi$-equilibria}, a notion of equilibrium parameterized by a per-player set $\Phi_i$ of \emph{strategy transformations} $\phi : \cP_i \to \cP_i$. Specifically, a $\Phi$-equilibrium is a  (joint) distribution over the set of all possible strategy profiles $\cP_1 \times\dots\times \cP_n$  such that no player $i$ benefits in expectation from applying any transformation $\phi \in \Phi_i$ to the strategy sampled for them by the distribution.
As the size of each set $\Phi_i$ increases to encompass more and more transformations, the set of $\Phi$-equilibria becomes a tighter and tighter relaxation of the set of Nash equilibria. For example, the concept of coarse correlated equilibrium arises when each $\Phi_i$ contains all possible {\em constant} functions. The concept of correlated equilibrium arises when each $\Phi_i$ contains \emph{all} functions $\cP_i \to \cP_i$. And there are various other notions of $\Phi$-equilibrium that have been considered in the literature; see e.g.~\cite{Gordon08:No} for general games, \cite{vonStengel2008,MorrillDSLWGB21,Farina2024:eah,zhang2024efficient} for extensive-form games, and~\cite{cai2024tractable} for games with non-concave utilities.

Within this context, a central open question in equilibrium computation is understanding what are the strongest notions of $\Phi$-equilibria that can be computed efficiently in convex games. In particular, it is known that coarse correlated equilibria can be computed efficiently given oracle access (e.g., a separation oracle) to the strategy set of each player in a convex game. On the other hand, it is a major open question, as we have already discussed, whether correlated equilibria can be computed efficiently even in the special case of extensive-form games, with some preliminary evidence pointing towards intractability~\citep{daskalakis2024lower}, despite some recent success for approximate correlated equilibria. This raises the following natural (informal) question, which serves as a main motivation for this work:
\begin{quote}\itshape
    {\bf Motivating Question 1:} Are there notions of $\Phi$-equilibrium that are stronger than coarse correlated equilibrium and can be efficiently computed in general convex games, given oracle access to the convex strategy sets of the player? 
\end{quote}
A related challenge is understanding whether these notions of equilibrium can be efficiently learned via no-regret learning procedures, and more generally whether some notion of regret that is stronger than external regret can be minimized against an adversary. This amounts to  the following question: 
\begin{quote}\itshape
    {\bf Motivating Question 2:} Are there notions of $\Phi$-regret that are stronger than external regret and can be efficiently minimized against an advarsarial sequence of losses?
\end{quote}
We postpone the definition of $\Phi$-regret to Section~\ref{sec:linear swap, linear correlated, convex games}, but mention here that it is a notion of regret that strengthens external regret and relaxes swap regret. Moreover, in the same way that the dynamics arising when the players of a game run no-external-regret learning procedures to update their strategies converge to coarse correlated equilibrium, and no-swap-regret learning dynamics converge to correlated equilibrium, no-$\Phi$-regret learning dynamics converge to $\Phi$-equilibrium. 

In this paper, we provide positive answers to the above questions. We consider the natural setting where each $\Phi_i$ contains the set of {\em all (affine) linear endomorphisms} from $\cP_i$ to itself, which induces a notion of $\Phi$-equilibrium known as \emph{linear correlated equilibrium} and a notion of $\Phi$-regret known as \emph{linear swap regret}. Our main result for general convex games and this choice of $\Phi$ is stated informally below. It significantly expands the frontier of tractable equilibria and game structures.

\begin{theorem}[Informal; formal version given as Theorems~\ref{thm:linswap-regret-main-algo-strong} and~\ref{thm:eah-equil-computation-strong}] Linear correlated equilibria can be efficiently computed in general convex games, using  polynomially many oracle calls to the players' strategy sets. Moreover, there exist efficient no-linear-swap-regret learning procedures, minimizing the linear swap regret of a learner playing in a convex set to which oracle access is provided against an adversarial sequence of linear losses.
\end{theorem}

En route to proving the previous result, we develop new fundamental techniques for dealing with the set of all  linear endomorphisms of a generic convex set $\cP$ given via oracle access, which might be of independent interest. A key technical challenge, whose relevance is explained in the next subsection, is that even though $\cP$ might admit an efficient membership or separation oracle, constructing a membership or separation oracle for the set of all linear endomorphisms on $\cP$ is in general intractable, as shown in Theorem~\ref{thm:hardness-opt}. This renders standard techniques for $\Phi$-regret minimization and $\Phi$-equilibrium computation non-applicable. Indeed, much of our work develops new algorithmic techniques that operate in a relaxed setting, where separation is replaced with a new concept we call \emph{semi-separation}. 

\subsection{Technical Innovation and Approach}

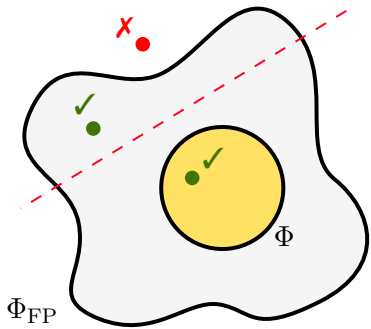
\begin{figure}[htp]
    \centering
    \tikzset{every picture/.style={line width=0.75pt}} %

\begin{tikzpicture}[x=1.1pt,y=1.1pt,yscale=-1,xscale=1]

\draw  [fill={rgb, 255:red, 0; green, 0; blue, 0 }  ,fill opacity=0.04 ][line width=1.5]  (225.39,71.45) .. controls (236.87,49.69) and (254.73,68.63) .. (270.25,60.21) .. controls (285.76,51.8) and (297.16,27.39) .. (313.81,43.75) .. controls (330.46,60.12) and (315.97,86.49) .. (329,97) .. controls (342.03,107.51) and (347.46,127.1) .. (324,139) .. controls (300.54,150.9) and (300.49,141.89) .. (289.32,140.05) .. controls (278.16,138.2) and (266.84,152.41) .. (248,145) .. controls (229.16,137.59) and (244.26,128.97) .. (240.92,111.99) .. controls (237.58,95.01) and (213.92,93.2) .. (225.39,71.45) -- cycle ;
\draw  [fill={rgb, 255:red, 255; green, 226; blue, 99 }  ,fill opacity=1 ][line width=1.5] (273.27,112.2) .. controls (266.52,102.78) and (268.68,89.66) .. (278.11,82.9) .. controls (287.54,76.15) and (300.65,78.31) .. (307.41,87.74) .. controls (314.17,97.16) and (312,110.28) .. (302.58,117.04) .. controls (293.15,123.8) and (280.03,121.63) .. (273.27,112.2) -- cycle ;
\draw  [draw opacity=0][fill={rgb, 255:red, 65; green, 117; blue, 5 }  ,fill opacity=1 ] (246,79.5) circle (2.5);
\draw  [draw opacity=0][fill={rgb, 255:red, 255; green, 0; blue, 0 }  ,fill opacity=1 ] (263,50.5) circle (2.5) ;
\draw [color={rgb, 255:red, 255; green, 0; blue, 31 }  ,draw opacity=1 ] [dash pattern={on 4.5pt off 4.5pt}]  (333,39) -- (221,107) ;
\draw  [draw opacity=0][fill={rgb, 255:red, 65; green, 117; blue, 5 }  ,fill opacity=1 ] (280,96.5) circle (2.5) ;

\draw (238.39,66.45) node [anchor=north west][inner sep=0.75pt]  [color={rgb, 255:red, 65; green, 117; blue, 5 }  ,opacity=1 ] [align=left] {\ding{51}};
\draw (253,40) node [anchor=north west][inner sep=0.75pt]  [color={rgb, 255:red, 255; green, 0; blue, 0 }  ,opacity=1 ] [align=left] {\ding{55}};
\draw (282.5,85) node [anchor=north west][inner sep=0.75pt]  [color={rgb, 255:red, 65; green, 117; blue, 5 }  ,opacity=1 ] [align=left] {\ding{51}};
\draw (307,112) node [anchor=north west][inner sep=0.75pt]   [align=left] {$\displaystyle \Phi $};
\draw (215,137) node [anchor=north west][inner sep=0.75pt]   [align=left] {$\displaystyle \Phi _{\text{FP}}$};

\end{tikzpicture}
    \caption{Illustration of our \emph{semi-separation oracle} for the set of linear endomorphisms $\Phi$ of the feasible set $\cP$. Given a candidate linear transformation $\phi$, the oracle returns a fixed point of $\phi$ in the set $\cP$, if one exists (this is the case of the two points marked with \ding{51}), or a hyperplane separating $\phi$ from $\Phi$ (this is the case of the point marked \ding{55}). In the figure, we denoted with $\Phi_\text{FP}$ the set of linear transformations that admit a fixed point in $\cP$. In general, $\Phi_\text{FP}$ is a strict superset of $\Phi$, and is not a convex set. Building a separation oracle for $\Phi$ is generally computationally intractable. }
    \label{fig:semi-sep}
\end{figure}

Recent results on efficient $\Phi$-regret minimization \citep{celli2020no,Farina22:Simple,MorrillDSLWGB21,farina2024polynomial,zhang2024efficient} have all relied on constructing efficient algorithms for accessing or optimizing over sets of transformations $\Phi$; mostly as part of the seminal framework of Gordon, Greenwald and Marks \citep{Gordon08:No} that requires access to $\Phi$.
Moreover, the generalization of this framework to exact computation of $\Phi$-equilibria in polyhedral games \citep{farina2024polynomial} requires efficient oracle access to $\Phi$.

Our work tackles the more general challenge of computing $\Phi$-equilibria and minimizing $\Phi$-regret with respect to the set of \emph{all} (affine) linear endomorphisms $\Phi$, given only oracle access to the action set $\cP$.  
As we will show, without any structural information about $\cP$, %
we cannot optimize over $\Phi$ and, thus, we cannot directly apply the frameworks of \cite{Gordon08:No} and \citet{farina2024polynomial}.
In particular, we show in Theorem~\ref{thm:hardness-opt} that, given oracle access to $\cP$, we cannot always efficiently distinguish whether a linear transformation is an endomorphism of $\cP$ or not.  In view of this intractability, we put forth an algorithmic approach that generalizes the aforementioned frameworks and achieve efficient $\Phi$-regret minimization and computation of $\Phi$ equilibria in general convex games with only oracle access to the strategy sets $\cP$.

Our main idea is to make use of what we refer to as a ``semi-separation'' oracle.  A standard separation oracle for a convex set $\Phi$ would provide, for any queried point $\phi$, either an assertion that $\phi \in \Phi$, or a hyperplane separating $\phi$ from $\Phi$.
As we said, we cannot efficiently check whether $\phi \in \Phi$, when $\Phi$ is the set of all linear endomorphisms of $\cP$.  However, given that linear functions are simple functions and we have oracle access to $\cP$, we can efficiently check if $\phi$ has a fixed point that lies in $\cP$.  That is, we can efficiently check if $\phi \in \Phi_{\text{FP}}$, where $\Phi_{\text{FP}}$ is the set of linear transformations that admit a fixed point in $\cP$.  Clearly, the set $\Phi_{\text{FP}}$ is a superset of $\Phi$ (as all continuous endomorphisms of $\cP$ must have a fixed point in $\cP$ by Brouwer's fixed point theorem). Thus, it would definitely suffice to minimize regret with respect to all transformations in $\Phi_\text{FP}$.  But, $\Phi_\text{FP}$ is not convex.  We cannot establish a separation for $\Phi_\text{FP}$, and if we were to take its convex hull, we would introduce transformations that do not have fixed points in $\cP$.  An adversary can easily force linear regret against any such transformation.

We circumvent these issues with the \emph{semi-separation} oracle (Section~\ref{sec:semi-sep}), which either asserts that a candidate $\phi$ is in $\Phi_{\text{FP}}$ or identifies a hyperplane separating $\phi$ from the the set $\Phi$ of linear endomorphic transformations.  If $\phi \not \in \Phi_{\text{FP}}$, we can find a hyperplane in $\bbR^d$ that separates its fixed point from $\cP$, which we then use to establish a separating hyperplane in $\bbR^{d\times(d+1)}$, cutting away some affine transformations that have fixed points violating the first hyperplane.  This \emph{semi-separation} suffices as a building block for linear swap regret minimization and a generalized ellipsoid against hope algorithm for linear correlated equilibrium.

\subsection{Related Work}

Prior work has considered the minimization of linear swap regret and the computation of linear correlated equilibria. \citet{mansour2022strategizing} study the setting in which a no-regret learner competes against a rational utility maximizer in a two-player
Bayesian game, finding that linear swap regret minimization is necessary to avoid exploitation (extending a similar result for swap-regret minimization in normal form games by \citet{deng2019strategizing}). \citet{fujii2023bayes} defines the notion of untruthful swap regret for
Bayesian games and proves that, for Bayesian games, it is equivalent to the linear-swap regret which
is of interest in our paper. Recently, \citet{cai2024tractable} demonstrate that the task of computing a ``local-$\Phi$-equilibrium'' in a non-concave game can be reduced to $\Phi$-regret minimization over linear loss functions. This establishes that certain natural variants of local correlated equilibria in non-concave games can be reduced to the problem of linear swap regret minimization, which we study.

From a computational perspective, \citet{dann2023pseudonorm} demonstrate how to efficiently minimize linear swap regret in Bayesian games via reduction to specific forms of Blackwell approachability. \citet*{farina2024polynomial} show that linear swap regret can be minimized in polynomial-time per iteration in the game tree in general extensive-form games. The result was later improved by \citet*{zhang2024mediator}, who show that linear correlated equilibria admit a correlation-device interpretation, and use the connection to develop faster algorithms. Extensions of these ideas to low-degree polynomial swap deviations in extensive-form games were developed by \citet*{zhang2024efficient}. Finally, \citet*{Farina2024:eah} propose a generalization of the ellipsoid-against-hope framework \citep{Papadimitriou2008:Computing} which, among other applications, shows that linear correlated equilibria can be computed in polynomial time in general extensive-form games. As mentioned, %
the results of this paper subsume and greatly generalize most of these prior results.

The algorithm of \citet{Gordon08:No} for $\Phi$-regret minimization requires two subroutine components: fixed-point computations for the transforms $\phi \in \Phi$ and a separation oracle for $\Phi$. An important open question arising from this result is whether both components are necessary. \citet{hazan2007computational} proved the necessity of fixed-point computations. They presented an algorithm that computes fixed points for any $\phi \in \Phi$ given access to a $\Phi$-regret minimization subroutine.\footnote{From this algorithm, the paper concludes that $\Phi$-regret minimization and $\Phi$-fixed point computation are equivalent. However, this conclusion assumes a setting where separation over the transformation set $\Phi$ can be easily established, which is not generally the case.} Consequently, there cannot exist a $\Phi$ such that $\Phi$-regret minimization is feasible while fixed points over $\Phi$ cannot be efficiently computed. The remaining open question was whether there exists a $\Phi$ such that $\Phi$-regret minimization is feasible while efficient separation over $\Phi$ is not. Here, we answer this question affirmatively, demonstrating that there is indeed flexibility in the reliance on the components identified by \citet{Gordon08:No}.

One of the key ideas our efficient linear swap regret minimization algorithm relies on is the idea that we can relax the set $\Phi$ of linear endomorphisms to a convex subset of the set $\Phi_{\text{FP}}$ of linear transformations with fixed points. The class of $\Phi$-regret minimization problems where the linear transformations $\phi \in \Phi$ admit fixed points but are not necessarily endomorphisms was introduced and studied by \cite{dann2024rate} under the name ``improper $\phi$-regret minimization''. \cite{dann2024rate} showed that although improper $\phi$-regret minimization is in some sense ``nonequivalent'' to standard $\phi$-regret minimization, improper $\phi$-regret can still be minimized via the algorithm of \cite{Gordon08:No}; we use this latter fact when constructing our semi-separation oracle.

\section{Preliminaries}

\subsection{Notation}

We write $\cB_d(\vc, r) := \set{\vx \in \bbR^d \mid \norm{\vx - \vc}_2 \leq r}$ to denote the $d$-dimensional Euclidean radius-$r$ ball centered at $\vc \in \bbR^d$. Sometimes, in an abuse of notation, we will use $\cB_d(r)$ to denote \textit{any} ball of radius $r$, when the center is irrelevant or clear from context. We denote the volume of a compact convex set $\cK$ with $\vol{\cK}$ and the volume of the $d$-dimensional ball of radius $r$ with $V_d(r)$.
We use $\norm{\mM}_F$ to denote the Frobenius norm and $\norm{\mM}_2$ to denote the spectral norm of a matrix $\mM \in \bbR^{n \times m}$.

For any two compact convex sets $\cA, \cB \subset \bbR^d$, we will let $\lintrans{\cA}{\cB}$ denote the set of all (affine)\footnote{For brevity, throughout this paper, we will generally omit the word affine and simply use the word \emph{linear} to refer to affine linear transformations of the form $Mx + b$.} linear transformations from $\cA$ to $\cB$. Note that any such transformation $\phi \in \lintrans{\cA}{\cB}$ can be expressed in the form $\phi(x) = Mx + b$ for some matrix $M \in \bbR^{d \times d}$ and some vector $b \in \bbR^{d}$. In this way, we can identify elements of $\lintrans{\cA}{\cB}$ with elements of $\bbR^{d \times (d+1)}$ and define the norm $\norm{\phi}_F$ as the Frobenius norm of the corresponding matrix (likewise, we can speak of $\phi$ belonging to ball $\cB_{d \times (d+1)}(R)$). We will write $\Phi(\cA)$ as shorthand for the set $\lintrans{\cA}{\cA}$ of linear \emph{endomorphisms} of $\cA$.

\subsection{Linear Swap Regret and Linear Correlated Equilibria} \label{sec:linear swap, linear correlated, convex games}

We begin by defining the adversarial online learning problem of \emph{$\Phi$-regret minimization}, originally introduced by \cite{Gordon08:No}. In this setting, at every round $t$ (over $T$ rounds) a learner must pick an action $\pt$ belonging to a bounded $d$-dimensional convex action set $\cP \subseteq \bbR^d$. Simultaneously, an adversary selects a bounded loss vector $\lt \in [-1,1]^d$. The learner then receives loss $\inp{\pt, \lt}$ and learns the loss vector $\lt$ chosen by the adversary (the knowledge of which they can then use to choose their action in future rounds).

The learner aims to minimize their \emph{$\Phi$-regret}. Formally, given any set $\Phi$ containing continuous transformations $\phi: \cP \rightarrow \cP$, the $\Phi$-regret of the learner is defined to equal

\begin{equation}\label{eq:phi-reg-def}
    \reg_{\Phi}(\vp, \vell) = \sum_{t=1}^T \inp{\pt, \lt} - \min_{\phi \in \Phi} \sum_{t=1}^T \inp{\phi(\pt), \lt}.
\end{equation}

Informally, this should be thought of as the gap between the realized utility of the learner, and the maximum utility they could have received had they applied the best transformation $\phi \in \Phi$ to all of their actions.

In this paper, we focus on an important instance of $\Phi$-regret minimization: \emph{linear swap regret minimization}. Linear swap regret minimization is the problem of $\Phi$-regret minimization for the class $\Phi$ of all linear transformations from $\cP$ to itself; in particular, it corresponds to the case of $\Phi$-regret minimization with $\Phi = \Phi(\cP) = \lintrans{\cP}{\cP}$. We write $\linswap(\vp, \vell)$ to denote the linear swap regret of a learner who has played the actions $\vp$ in response to the losses $\vell$. Note that in the case where $\cP$ is the $d$-dimensional simplex $\Delta_d$, the set $\Phi(\cP)$ contains all row-stochastic linear transformations, and linear swap regret collapses to the standard definition of swap regret. 

One of the primary reasons to study these notions of regret minimization is to understand the set of equilibria they induce. Just as external regret minimization and swap regret minimization lead to coarse correlated equilibria and correlated equilibria (respectively) in normal-form games, linear swap regret minimization leads to the set of linear correlated equilibria in convex games. 

\begin{definition}[convex game]
    An $n$-player \emph{convex game} is defined by $n$ compact convex strategy sets $\cP_1, \dots, \cP_n$ (where $\cP_i \subset \bbR^{d_i}$) and $n$ multi-linear payoff functions $u_i:\cP_1 \times \dots \times \cP_n \rightarrow \bbR$.
\end{definition}

This definition of convex games includes
Bayesian games and extensive-form games (where each players' strategy sets are their \emph{sequence form polytopes}, see e.g. \cite{Koller96:Efficient}), games where strategy spaces are normed balls in some norm, and games where the strategies of each player have some combinatorial structure (e.g. routing games where each player picks a path through a graph).

\begin{definition}[linear correlated equilibrium]
    A \emph{linear correlated equilibrium} in a convex game $G$ with strategy sets $\cP_1, \dots, \cP_n$ is a joint distribution $\mu \in \Delta(\cP_1 \times \dots \times \cP_n)$, such that for every player $i \in [n]$ and linear deviation $\phi \in \Phi(\cP_i)$,
    \begin{equation*}
    \E_{\vec{s} \sim \mu}[u_i(\vec s)] \geq \E_{\vec{s} \sim \mu} [u_i(\phi(\vec{s}_i), \vec{s}_{-i})].  
    \end{equation*}
\end{definition}

We say such a joint distribution $\mu$ is an $\epsilon$-approximate linear correlated equilibrium if the above inequalities hold with $\epsilon$ slack, i.e., 

\begin{equation*}
\E_{\vec{s} \sim \mu}[u_i(\vec s)] \geq \E_{\vec{s} \sim \mu} [u_i(\phi(\vec{s}_i), \vec{s}_{-i})] - \epsilon.
\end{equation*}

\subsection{Convex Optimization and Oracle Access}

We are interested in designing efficient swap regret minimization algorithms for generic convex strategy sets $\cP$. As is typical in such settings, we will assume nothing specific about the structure of $\cP$ and instead assume that we have \emph{oracle access} to this set $\cP$. In particular, we will assume we have access to oracles that can answer the following queries in time $\poly(d)$:

\begin{itemize}
	\item {\bf Membership: } Given $\vy \in \bbR^n$, decide whether $\vy \in \cP$.
	\item {\bf Separation: } Given $\vy \in \bbR^n$, decide whether $\vy \in \cP$, and if not, find a hyperplane that separates $\vy$ from $\cP$
	\item {\bf Optimization: } Given a vector $\vc \in \bbR^n$, find a $\vy \in \cP$ that maximizes the inner product $\inp{\vc, \vy}$. 
\end{itemize}

We will also make the standard assumption in convex optimization that the action set $\cP$ we work with is geometrically ``well-behaved'', in that: i) $\cP$ is contained within the ball $\cB_d(\vzero, R)$ of radius $R$ centered at the origin ($\cP$ is bounded), and ii) $\cP$ contains a small ball $\cB_{d}(\vec a, r)$ centered at some point $\vec a$ ($\cP$ is full dimensional and not too flat in any dimension). It is known that under this assumption, all three of the above oracles are equivalent (see e.g. \cite{Grotschel1993:Geometric}).

We will also want our set $\Phi(\cP)$ of linear endomorphisms of $\cP$ to be similarly well-behaved. We can show that this follows from the above assumptions.

\begin{restatable}{lemma}{lintransnice}
\label{lem:lintrans-nice}
    Let $\cB_d(\va, r) \subseteq \cP \subseteq \cB_d(\vzero, R)$. Then there exists a $\phi_{\va} \in \Phi(\cP)$ such that: 
    \[
        \cB_{d \times (d+1)}\left(\phi_{\va}, \frac{r}{2R}\right)\subseteq \Phi(\cP) \subseteq \cB_{d \times (d+1)}\left( \vzero, \frac{3R}{r} \sqrt{R^2 + d} \right).
    \]
\end{restatable}
\begin{proof}
See Appendix \ref{app:lintrans}.
\end{proof}

However, note that we \textbf{do not} assume we have oracle access to our set of linear endomorphisms $\Phi(\cP)$. Indeed, we will later show that oracle access to $\Phi(\cP)$ does not necessarily follow from oracle access to $\cP$, which in turn makes the problem of linear swap regret minimization more challenging.

Next, we define the isotropic position of a compact convex set.
\begin{definition}[Isotropic Position]
\label{defn:isotropic}
    A convex compact set $\cK$ is said to be in isotropic position if for a uniformly sampled $y \sim \cK$ it holds $\E[y] = 0$ and $\E[y y^\top] = \mI$.
\end{definition}
This brings general convex sets to a useful position with many nice properties. In this paper, we will mainly need the well-known property from convex geometry \citep{Kannan1997:RandomWA} that if $\cK$ is in isotropic position, then the outer and inner radii can be bounded as $\cB_d\left( \vzero, 1 \right) \subseteq \cK \subseteq \cB_d\left( \vzero, n + 1 \right)$.

For all our regret minimization results in Section~\ref{sec:opt_linear_end} and Section~\ref{sec:linswap-regret} we will assume that the set of strategies is in isotropic position. This comes without loss of generality because: (1) there exist polynomial time algorithms to compute an affine transformation that brings any convex set into isotropic position \citep{Lovasz2003:simulated, BenTal2001:Lectures} and, (2) we can construct an efficient linear-swap regret minimizer for any convex set by constructing a linear-swap regret minimizer for its transformation into isotropic position (Lemma~\ref{lem:isotropic_regret}).

Finally, we note that in practice, we can only work with finite precision, and must work with \emph{weak variants} of the above oracles that only return answers that are guaranteed to be correct up to some $\epsilon$ slack. For convenience of presentation, we present all results in the main body of this paper assuming we have access to these strong oracles and ignoring issues of precision; that said, all of our results hold (when suitably relaxed) in the presence of weak oracles. We defer these details to Appendices \ref{app:weak} and \ref{app:eah-weak}.

\subsection{From Linear Swap Regret Minimization to External Regret Minimization}\label{sec:linear-to-external}

Essentially all\footnote{Recent results \citep{dagan2024external}, \citep{peng2024fast} establish a new approach for swap regret minimization, but incur exponential dependence on $1/\eps$.}%
standard approaches to linear swap regret minimization (and more generally, $\phi$-regret minimization) work by reducing the original swap regret minimization problem to a related \emph{external regret} minimization problem, albeit one where the action set is the set $\Phi(\cP)$ of endomorphisms of $\cP$, instead of the original action set $\cP$. Here we outline this reduction (which follows the reduction presented by \cite{Gordon08:No}, along with reductions that arise by casting this as an instance of Blackwell approachability \citep{abernethy2011blackwell, blackwell1956analog, dann2023pseudonorm}).

Recall that the goal of linear swap regret minimization is to produce a sequence of actions $\pt$ that guarantees that the quantity

\begin{equation}\label{eq:reduction-linswap}
\linswap(\vp, \vell) = \sum_{t=1}^T \inp{\pt, \lt} - \min_{\phi^* \in \Phi(\cP)} \sum_{t=1}^T \inp{\phi^*(\pt), \lt}
\end{equation}

\noindent
grows sublinearly with $T$. The main idea behind this reduction is to simultaneously solve a ``dual'' online learning problem where the goal is to produce a sequence of linear transformations $\phi_t \in \Phi(\cP)$ with the guarantee that the ``dual regret''

\begin{equation}\label{eq:reduction-dualreg}
\reg^{\text{dual}}(\mathbf{\phi}, \vp, \vell) = \sum_{t=1}^T \inp{\phi_t(\pt), \lt} - \min_{\phi^* \in \Phi(\cP)} \sum_{t=1}^T \inp{\phi^*(\pt), \lt}
\end{equation}

\noindent
grows sublinearly with $T$.

In each round $t$, after producing $\phi_t$, the learner then plays any action $\pt$ that is a fixed point of $\phi_t$; i.e., with the property that $\phi_t(\pt) = \pt$. Note that such a fixed point is guaranteed to exist by Brouwer's fixed point theorem since $\phi_t$ continuously maps $\cP$ into itself; moreover, for $\phi_t \in \Phi(\cP)$, we can find this fixed point easily via solving a linear system. But if $\phi_t(\pt) = \pt$, then from equations \eqref{eq:reduction-linswap} and \eqref{eq:reduction-dualreg} we immediately have that $\linswap(\vp, \vell) = \reg^{\text{dual}}(\mathbf{\phi}, \vp, \vell) = o(T)$.

Now, the dual regret $\reg^{\text{dual}}(\mathbf{\phi}, \vp, \vell)$ is a form of external regret, where the learner plays actions $\phi_t \in \Phi(\cP)$ and faces losses in $\bbR^{d \times (d+1)}$ determined by $\pt$ and $\lt$; in particular, the learner is trying to play a sequence of transformations $\phi_t$ that competes with the single best transformation $\phi^*$ in hindsight.

\section{Optimizing over Linear Endomorphisms}
\label{sec:opt_linear_end}

Motivated by the reduction in Section \ref{sec:linear-to-external}, we begin by exploring the question of when it is possible to construct an efficient optimization oracle for $\Phi(\cP)$ (or equivalently, a membership / separation oracle). We will show that if we have an explicit representation of $\cP$ as a polytope with a small (polynomial) number of faces or vertices, then this is possible. On the other hand, we will show that if we only have oracle access to $\cP$, it is \emph{impossible} to construct such oracles for $\Phi(\cP)$ that run in sub-exponential time. 

\subsection{Endomorphisms of Simple Polytopes}

In this section, we restrict our attention to endomorphisms of polytopes that are ``simple'': polytopes that are either intersections of a small number of half-spaces, or the convex hull of a small number of vertices (or both). We say that a polytope $\cP$ has an \emph{$H$-representation} of size $n$ if we have an explicit representation of $\cP$ as the set of points satisfying $n$ constraints ($\cP = \set{ \vx \in \bbR^d \mid \vec{a}_i^\top \vx \leq \vb_i, i \in [n]}$). Similarly, we will say that a polytope $\cP$ has a \emph{$V$-representation} of size $n$ if we have an explicit representation of $\cP$ as the convex hull of $n$ vertices ($\cP = \conv(\{v_1, v_2, \dots, v_n\})$). 

We will show that in either case, we can construct an efficient membership oracle (that runs in $\poly(n, d)$ time) for the set $\Phi(\cP)$ of endomorphisms of $\cP$. In fact, we prove a slightly stronger statement: we can construct efficient oracles for $\lintrans{\cA}{\cB}$ as long as either $\cA$ has a small $V$-representation or $\cB$ has a small $H$-representation (and we have oracle access to the other set).

\begin{lemma}\label{lem:vrep}
Let $\cA$ and $\cB$ be convex sets where $\cA$ has a $V$-representation of size $n$ and where we are given a membership oracle for $\cB$. Then we can construct an efficient membership oracle for $\lintrans{\cA}{\cB}$ that runs in time $\poly(d, n)$.
\end{lemma}
\begin{proof}
Write $\cA = \conv(\{v_1, v_2, \dots, v_n\})$. To check whether a transformation $\phi$ belongs to $\lintrans{\cA}{\cB}$, it suffices to check whether $\phi(v_i) \in \cB$ for each $v_i$. This can be done with $n$ calls to the membership oracle for $\cB$. 
\end{proof}

\begin{lemma}\label{lem:hrep}
Let $\cA$ and $\cB$ be convex sets where $\cB$ has a $H$-representation of size $n$ and where we are given an optimization oracle for $\cA$. Then we can construct an efficient membership oracle for $\lintrans{\cA}{\cB}$ that runs in time $\poly(d, n)$.
\end{lemma}
\begin{proof}
Write $\cB = \set{ \vx \in \bbR^d \mid \vec{a}_i^\top \vx \leq \vb_i, i \in [n]}$. To check whether a transformation $\phi$ belongs to $\lintrans{\cA}{\cB}$, it suffices to check whether $\vec{a}_i^{T}\phi(\vx) \leq \vec{b}_i$ for all $i \in [n]$ and $x \in \cA$. Since $\vec{a}_i^{T}\phi(\vx)$ is an affine function in $\vx$, for each $i$ this can be checked with one call to the optimization oracle for $\cA$, and hence overall can be done in $\poly(d, n)$ time. 
\end{proof}

As a consequence of Lemmas \ref{lem:vrep} and \ref{lem:hrep}, we have the following corollary.

\begin{corollary}\label{cor:efficient-rep}
If $\cP$ has either an $H$-representation of size $n$ or a $V$-representation of size $n$, it is possible to construct an efficient membership oracle for $\cP$.
\end{corollary}

We remark here that many practically occurring strategy sets have the property that they have a small $H$- or $V$- representation with $n$ polynomial in the dimension $d$. This is the case, for example, for the strategy sets that arise in Bayesian games and extensive-form games, along with for many combinatorial sets (e.g. matchings or flows). Therefore, assuming that the strategy set is in Isotropic Position (Definition~\ref{defn:isotropic}), it follows from Corollary \ref{cor:efficient-rep} that there exist efficient linear-swap regret minimization algorithms over these action sets that incur linear-swap regret at most $O(\poly(d)\sqrt{T})$.

However, other naturally occurring combinatorial polytopes (such as the spanning tree polytope) \emph{do not} have a succinct representation in terms of vertices or half-spaces. Moreover, a generic convex set will, in general, not even have a finite such representation. In the next section we will show that this is a real obstacle to optimizing over the space of endomorphisms.

\subsection{Hardness of Optimization over Linear Endomorphisms}

We now show that it is, in general, computationally intractable to build an efficient membership (or separation, or optimization) oracle for $\Phi(\cP)$ given only an efficient membership oracle for $\cP$. In fact, this holds even for deciding very weak forms of membership.

\begin{theorem}\label{thm:hardness-opt}
Fix a $d > 0$. Let $A$ be any algorithm that takes as input a convex set $\cP \subset \bbR^d$ (through a membership oracle to $\cP$) and an affine transformation $\phi \in \bbR^{d \times (d+1)}$ and has the guarantee that:
\begin{enumerate}
    \item If $\cB_{d \times (d+1)}(\phi, 1/d) \subseteq \Phi(\cP)$, outputs YES.
    \item If $\phi \not\in \Phi(\cP)$, outputs NO.
\end{enumerate}
Then $A$ must make at least $\exp(\Omega(d))$ queries to the membership oracle for $\cP$. 
\end{theorem}

The details of the proof of Theorem~\ref{thm:hardness-opt} are deferred to Appendix \ref{app:lb}. The main idea is that it is hard to distinguish (given only oracle access) between the full $d$-dimensional unit ball and the $d$-dimensional unit ball with a random ``cap'' removed (i.e., the intersection of the unit ball with a halfspace of the form $\inp{\vx, \vec{u}} \leq \kappa$ for a random direction $u$ and some constant $\kappa$). For example, the only membership oracle queries one can make that distinguish between these two sets are sets that belong to the cap itself. On the other hand, these two sets have very different sets of endomorphisms (e.g. negation is always an endomorphism of the full ball, and never an endomorphism of the capped ball). Since there exist $\exp(\Omega(d))$ non-intersecting caps of constant width on the unit ball, any algorithm satisfying the conditions of Theorem~\ref{thm:hardness-opt} requires at least as many membership queries.

One interesting implication of Theorem~\ref{thm:hardness-opt} is that, in general, it is intractable to exactly compute the linear swap regret for a given sequence of losses $\vell_1, \dots, \vell_T$ and strategies $\vp_1, \dots, \vp_T$. In particular, note that computing linear swap regret is exactly the problem of computing the maximum of the linear functional $\sum_{t =1}^{T} \langle \vp_t - \phi(\vp_t), \vell_t\rangle$ over all $\phi \in \Phi$. Moreover, since we can decompose any $n$-by-$n$ matrix $\mM$ into a sum of $n$ rank-1 matrices, we can construct a pair of strategy / loss sequences that correspond to any linear functional. Therefore any efficient procedure that generically computes linear swap regret can be used as an optimization oracle for the set $\Phi$ (which in turn implies the existence of an efficient membership oracle, contradicting Theorem~\ref{thm:hardness-opt}). Nevertheless, as we will see in the next section, despite the intractability of computing the regret for given sequences of losses and strategies, there exist efficient polynomial-time algorithms that minimize linear-swap regret with only oracle access to the strategy sets.

\section{Linear Swap Regret Minimization for General Convex Sets}
\label{sec:linswap-regret}

In this section, we prove the first major result of this paper: that, despite the hardness of optimizing over the set of linear endomorphisms (shown in Appendix \ref{app:lb}), it is still possible to construct an algorithm that efficiently minimizes linear swap regret (obtaining $O(\poly(d)\sqrt{T})$ linear swap regret after $T$ rounds)
with only oracle access to the action set $\cP$.

Our algorithm builds off of the reduction of \cite{Gordon08:No} summarized in Section \ref{sec:linear-to-external}, but necessitates some new algorithmic ideas since we cannot reduce directly to an external regret problem over $\Phi(\cP)$. The first key observation is that although this reduction requires the learner to play transformations $\phi_t \in \Phi(\cP)$ in the dual regret minimization problem, actually the only property of $\phi_t$ we use is that it possesses a fixed point $\vp_t \in \cP$. Therefore, instead of working with the set of linear endomorphisms $\Phi(\cP)$ of $\cP$, it is natural to work with the \emph{set of linear transformations with fixed points in $\cP$}: the set $\PhiFP(\cP)$ containing all $\phi \in \bbR^{d \times (d+1)}$ such that there exists a $\vp \in \cP$ with $\phi(\vp) = \vp$. %

The set of transformations $\PhiFP(\cP)$ is a superset of the set of endomorphisms $\Phi(\cP)$, and has the advantage that it is computationally tractable to check whether a given $\phi$ contains a fixed point in $\cP$; for example, this can be done by minimizing the convex quadratic function $\norm{\phi(p) - p}^2$ over $p \in \cP$, which can be done with just oracle access to $\cP$. However, the downside is that this set of fixed-point transformations $\PhiFP(\cP)$ is highly \emph{non-convex}. As a result, it can easily be shown that it is impossible to obtain sublinear $\phi$-regret with respect to the full set of transformations $\PhiFP(\cP)$. 

The second key observation is that instead of working with either the computationally intractable set of transformations $\Phi(\cP)$, or the computationally tractable but non-convex set of transformations $\PhiFP(\cP)$, we can effectively interpolate between them, starting with a very loose approximation to $\Phi(\cP)$ and refining it over time. More precisely, in each round we will compute a ``shell'' approximation $\PhiShell(\cP)$ of $\Phi(\cP)$ that is guaranteed to contain $\Phi(\cP)$, but also defined in such a way that when we project to it (using a variant of projected gradient descent to solve the dual online learning problem), we will find a linear transformation belonging to $\PhiFP(\cP)$. In addition, $\PhiShell(\cP)$ will always be defined via a polynomial number of explicit half-space constraints, thus allowing us to optimize over it and substitute it in for $\Phi(\cP)$ in the dual external regret minimization problem.

The last ingredient is how we update $\PhiShell(\cP)$. Note that as long as the transformations $\phi_t \in \PhiShell(\cP)$ we construct contain a fixed point in $\cP$ (i.e., $\phi_t$ also belongs to $\PhiFP(\cP)$), we can play the corresponding fixed points and the logic of the original reduction goes through. On the other hand, when $\phi_t$ does not belong to $\PhiFP(\cP)$, we would like to update $\PhiShell(\cP)$ so to also exclude $\phi_t$. To this end, we develop an efficient subroutine that we call a \emph{semi-separation oracle}, that given any linear transformation $\phi$, either certifies that $\phi$ belongs to $\PhiFP(\cP)$, or produces a linear hyperplane separating it from $\Phi(\cP)$ (that we can then use to update the shell set). Note that this is a weaker oracle than a separation oracle for $\Phi(\cP)$ (which by Theorem \ref{thm:hardness-opt} cannot be efficiently implemented).

The remainder of this section is structured as follows. First, we present the construction of this semi-separation oracle in Section \ref{sec:semi-sep}. We then describe how to strengthen this semi-separation oracle to an oracle that can semi-separate convex sets of points from $\Phi(\cP)$ by running a variant of the ellipsoid algorithm we call $\shellelpsd$ (Section \ref{sec:shell-ellipsoid}). In Section \ref{sec:shell_gd_and_proj}, we use these oracles to develop two useful tools: Shell Gradient Descent (a variant of projected gradient descent that can handle the fact that we change our action space to a different shell set every round), and Shell Projection (a tool that uses $\shellelpsd$ to guarantee that projecting onto our shell set results in a transformation $\phi \in \PhiFP(\cP)$). Finally, we put all these pieces together and present the final learning algorithm in Section \ref{sec:mainalgo}.

\subsection{A Semi-Separation Oracle for the Set of  Endomorphisms}\label{sec:semi-sep}

We begin by constructing a semi-separation oracle for the set $\Phi(\cP)$ of endomorphisms of $\cP$ (see Figure \ref{fig:semi-sep}). Recall that this is an algorithm which, when given a linear transformation $\phi$, either certifies that it contains a fixed point in $\cP$ or produces a hyperplane separating it from $\Phi(\cP)$. Intuitively, this follows from the fact that any linear transformation $\phi$ without a fixed point must strictly increase the inner product of $p$ in some (efficiently computable) direction $u$, whereas this cannot be the case for any endomorphism of $\cP$.

\begin{lemma}[Semi-separation oracle]\label{lem:sep-non-fixed-point-strong}
    There exists an algorithm that, given any $\phi \in \bbR^{d \times (d+1)}$, runs in $\poly(d)$ time (making $\poly(d)$ oracle queries to $\cP$) and either returns a fixed point of $\phi$ contained within $\cP$ or a hyperplane separating $\phi$ from the set $\Phi(\cP)$ of endomorphisms.
\end{lemma}
(For the weak version of this, see Lemma~\ref{lem:sep-non-fixed-point}.)
\begin{proof}
With oracle access to $\cP$, we can find the point $\vp^{*} \in \cP$ that minimizes the convex function $f(\vp) = \norm{\phi(\vp) - \vp}^2$. If $f(\vp^*) = 0$, then $\phi(\vp^{*}) = \vp^{*}$ and we can return $\vp^{*}$ as a fixed point of $\phi$.

Otherwise, it must be the case that $\inp{\phi(\vp) - \vp, \phi(\vp^*) - \vp^*} > 0$ for all $\vp \in \cP$. Let $\vu = \phi(\vp^*) - \vp^*$, and use our oracle access to $\cP$ to compute the point $\vp_{\vu} = \arg\max_{\vp \in \cP} \inp{\vp, \vu}$ in $\cP$ that is furthest in the direction $\vu$. Note that by this maximality, any actual endomorphism $\phi' \in \Phi(\cP)$ must satisfy the linear constraint that $\inp{\phi'(\vp_{\vu}), \vu} \leq \inp{\vp_{\vu}, \vu}$. However for the provided $\phi$, we have that $\inp{\phi(\vp) - \vx, \vu} > 0$ for all $\vp \in \cP$, and in particular $\inp{\phi(\vp_{\vu}), \vu} > \inp{\vp_{\vu}, \vu}$. Therefore the linear constraint $\inp{\phi'(\vp_{\vu}), \vu} \leq \inp{\vp_{\vu}, \vu}$ separates $\phi$ from all $\phi' \in \Phi(\cP)$, and we can return it (in particular, this is a linear constraint on $\phi'$, and can be rewritten in the form $\inp{\phi', C} \leq b$ for an appropriate $C \in \bbR^{d \times (d+1)}$ and $b \in \bbR$ that we can efficiently compute given $\vu$ and $\vp_{\vu}$).
\end{proof}

\subsection{The Shell Ellipsoid Algorithm}\label{sec:shell-ellipsoid}

\begin{figure}
    \centering
    \tikzset{every picture/.style={line width=0.75pt}} %

\begin{tikzpicture}[x=1.1pt,y=1.08pt,yscale=-1,xscale=1]

\draw  [fill={rgb, 255:red, 0; green, 0; blue, 0 }  ,fill opacity=0.04 ][line width=1.5]  (126.39,83.45) .. controls (137.87,61.69) and (155.73,80.63) .. (171.25,72.21) .. controls (186.76,63.8) and (198.16,39.39) .. (214.81,55.75) .. controls (231.46,72.12) and (216.97,98.49) .. (230,109) .. controls (243.03,119.51) and (248.46,139.1) .. (225,151) .. controls (201.54,162.9) and (201.49,153.89) .. (190.32,152.05) .. controls (179.16,150.2) and (167.84,164.41) .. (149,157) .. controls (130.16,149.59) and (145.26,140.97) .. (141.92,123.99) .. controls (138.58,107.01) and (114.92,105.2) .. (126.39,83.45) -- cycle ;
\draw  [fill={rgb, 255:red, 255; green, 226; blue, 99 }  ,fill opacity=1 ][line width=1.5]  (172.38,122.19) .. controls (165.79,113) and (167.9,100.2) .. (177.09,93.62) .. controls (186.29,87.03) and (199.08,89.14) .. (205.67,98.33) .. controls (212.26,107.52) and (210.14,120.32) .. (200.95,126.9) .. controls (191.76,133.49) and (178.97,131.38) .. (172.38,122.19) -- cycle ;
\draw  [fill={rgb, 255:red, 74; green, 144; blue, 226 }  ,fill opacity=0.36 ] (217.78,93.46) -- (237.51,81.14) -- (239.54,101.68) -- cycle ;
\draw  [fill={rgb, 255:red, 208; green, 2; blue, 27 }  ,fill opacity=0.36 ] (123.81,114.11) -- (123.74,137.37) -- (105.27,128.15) -- cycle ;
\draw  [fill={rgb, 255:red, 184; green, 233; blue, 134 }  ,fill opacity=0.36 ] (193.52,115.35) -- (265.58,162.83) -- (252.24,178.58) -- cycle ;
\draw  [fill={rgb, 255:red, 184; green, 233; blue, 134 }  ,fill opacity=0.36 ] (183.16,119.21) -- (194.32,142.03) -- (181.21,144.53) -- cycle ;
\draw  [fill={rgb, 255:red, 74; green, 144; blue, 226 }  ,fill opacity=0.36 ] (201.35,79.59) -- (204.85,62) -- (215.59,68.68) -- cycle ;
\draw  [fill={rgb, 255:red, 184; green, 233; blue, 134 }  ,fill opacity=0.36 ] (187.02,113.26) -- (183.58,95.65) -- (196.06,97.76) -- cycle ;

\node[rotate=0] at (172, 140) {$\cF_1$};
\node[rotate=0] at (248, 160) {$\cF_2$};
\node[rotate=20] at (177, 110) {$\cF_3$};
\node[rotate=0] at (250, 90) {$\cF_4$};
\node[rotate=0] at (192, 70) {$\cF_5$};
\node[rotate=0] at (105, 115) {$\cF_6$};

\node at (170, 90) {\large$\Phi$};
\node at (160, 170) {\large$\Phi_\text{FP}$};
\node[anchor=west,text width=8.0cm] at (290, 60) {$\blacktriangleright$ $\cF_1, \cF_2,\cF_3$: \shellelpsd{}$(\cF)$ will return $\phi \in \Phi_\text{FP} \cap \cF$.};

\node[anchor=west,text width=8.0cm] at (290, 110) {$\blacktriangleright$ $\cF_4, \cF_5$: \shellelpsd{}$(\cF)$ might return a collection of hyperplanes separating $\cF$ from $\Phi$, or a transformation $\phi \in \Phi_\text{FP} \cap \cF$.};

\node[anchor=west,text width=8.0cm] at (290, 160) {$\blacktriangleright$ $\cF_6$: \shellelpsd{}$(\cF_6)$ will return a collection of hyperplanes separating $\cF$ from $\Phi$.};

\end{tikzpicture}
    \caption{Illustration of the behavior of the \shellelpsd{} subroutine depending on the input convex set $\cF$. $\Phi$ is the set of linear endomorphism on $\cP$, and $\Phi_\text{FP}$ is the set of linear transformations with a fixed point in $\cP$.}
    \label{fig:shell-ellipsoid}
\end{figure}

We now take the semi-separation oracle from the previous section (for individual points) and use it to construct a stronger semi-separation oracle for convex sets of points. This procedure -- which we call $\shellelpsd(\cF)$ -- takes as input a convex set $\cF \subset \bbR^{d \times (d+1)}$ of transformations (with efficient oracle access) and either returns a transformation $\phi \in \cF$ that possesses a fixed point in $\cP$, or certifies that the intersection of $\cF$ and $\Phi(\cP)$ is (nearly) empty (see Figure \ref{fig:shell-ellipsoid}). This will be an important primitive later, allowing us to efficiently shrink our shell set when we are presented with a fixed-point-free $\phi$ (in essence, allowing us to not just remove $\phi$ from our shell, but some ball of transformations around $\phi$ from the shell).

\begin{algorithm}[ht]
    \caption{$\shellelpsd(\cF)$ either finds a $\phi \in \cF$ with a fixed point in $\cP$, or returns a frontier (nearly) separating $\cF$ from $\Phi(\cP)$}\label{algo:shellelpsd-strong}
    \KwData{Convex strategy set $\cP \subset \bbR^d$}
    \KwInput{Oracle access to a convex set $\cF \subset \bbR^{d \times (d+1)}$ of affine transformations, precision parameter $\epsilon$}
    \KwOutput{Either a transformation $\phi \in \cF$ with a fixed point inside $\cP$, or a convex set $\cQ \supseteq \Phi(\cP)$ (specified as the intersection of polynomially many half-spaces) with $\vol{ \cQ \cap \cF} < \epsilon$}

    \DontPrintSemicolon

    Initialize $\cQ = \bbR^{d \times (d+1)}$

    \Repeat{$\vol{\mathcal{E}} < \epsilon$}{
        Let $\mathcal{E}$ be the minimum volume ellipsoid containing $\cQ \cap \cF$
    
        Let $\phi$ be the center of $\mathcal{E}$ (note $\phi \in \cQ \cap \cF$)

        Run the semi-separation oracle (Lemma~\ref{lem:sep-non-fixed-point}) on $\phi$
            
        \uIf{$\phi \in \PhiFP(\cP)$}{
            \Return $\phi$
        }
        \uElse{
            Get a half-space $H$ that separates $\phi$ from $\Phi(\cP)$ (\ie $\Phi(\cP) \subset H$, $\phi \not\in H$)

            Set $\cQ \leftarrow \cQ \cap H$
        }
    }

    \Return $\cQ$
\end{algorithm}

\begin{lemma}[Shell Ellipsoid]\label{lem:shellelpsd-strong}
    For any convex set $\cF \subseteq \cB_{d \times (d+1)}(\vzero, D)$ with efficient oracle access and $\epsilon > 0$, Algorithm~\ref{algo:shellelpsd-strong} runs in time $\poly(d, \log \epsilon^{-1}, \log D)$ and either
    \begin{itemize}
        \item Returns an affine transformation $\phi \in \cF$ with a fixed point inside $\cP$ (\emph{i.e.},  $\phi \in \cF \cap \PhiFP(\cP)$). 

        \item Returns a polytope $\cQ$ specified as the intersection of at most $\poly(d, \log \epsilon^{-1}, \log D)$ half-spaces with the property that $\Phi(\cP) \subseteq \cQ$ and $\vol{\cQ \cap \cF} < \eps$. 
    \end{itemize}
\end{lemma}
(For the weak version of this, see Lemma~\ref{lem:shellelpsd}.)
\begin{proof}
The termination conditions of Algorithm \ref{algo:shellelpsd-strong} (along with the proof of correctness of our semi-separation oracle) guarantee that any transformation $\phi$ returned must belong to $\cF \cap \PhiFP(\cP)$, and likewise that any polytope $\cQ$ returned must contain $\Phi(\cP)$ and satisfy $\vol{\cQ \cap \cF} \leq \vol{\mathcal{E}} < \epsilon$. It therefore suffices to verify that the loop in Algorithm \ref{algo:shellelpsd-strong} runs for at most $\poly(d, \log \epsilon^{-1}, \log D)$ iterations, as this bounds both the time complexity of the algorithm and the number of hyperplanes defining $\cQ$.

Note that in each iteration where the algorithm does not terminate, the ellipsoid $\mathcal{E}$ shrinks from the minimum volume ellipsoid containing $\cQ \cap \cF$ to the minimum volume ellipsoid $\mathcal{E'}$ containing $\cQ \cap \cF \cap H$, where $H$ is a halfspace passing through the centroid of $\mathcal{E}$. In particular, the volume of $\mathcal{E'}$ is no bigger than the volume of the minimum volume ellipsoid containing the half ellipsoid $\mathcal{E} \cap H$, which by the standard analysis of the ellipsoid algorithm has volume at most $\exp\left(-\frac{1}{2(\dim(\mathcal{E})+1)}\right)\vol{\mathcal{E}}$. In particular, the volume of $\mathcal{E}$ starts at most at $\vol{\cB_{d\times(d+1)}(D)} = O(D^{d(d+1)})$, shrinks by a factor of at least $\exp\left(-\frac{1}{2(d^2 + d +1)}\right)$ per round, and therefore after $O(d^2 \log d \log \eps^{-1} \log D)$ rounds will be at most $\eps$.

Finally, we remark that it is possible to efficiently compute the minimum volume ellipsoid $\mathcal{E}$ of $\cQ \cap \cF$ using only oracle queries to $\cQ \cap \cF$, which we can efficiently implement. (We can also start with any ellipsoid $\mathcal{E}$ whose centroid belongs to $Q \cap \cF$, which we can find efficiently by repeatedly doing ellipsoid updates -- this is the approach we take in the proof of the version with weak oracles, in Appendix \ref{app:weak}). 
\end{proof}

\subsection{Shell Gradient Descent and Shell Projection}\label{sec:shell_gd_and_proj}

We now turn our attention to applying these semi-separation oracles to implement a modified variant of the reduction in Section \ref{sec:linear-to-external}, where instead of running an external regret algorithm over the set of endomorphisms $\Phi(\cP)$, we run it over a set of changing shell sets of $\Phi(\cP)$. In this section we introduce two tools to help with this reduction: \emph{shell gradient descent} and \emph{shell projection}.

The first tool -- shell gradient descent (Algorithm \ref{algo:improper_ogd-strong}) -- addresses the fact that, unlike in standard online learning settings, we are faced with an online learning problem with an action set that changes over time. In particular, we are faced with the constraint that in each round $t$, our action must belong to some specified shell set $\PhiShell_t$ that we observe at the beginning of round $t$. We show that a simple variant of projected gradient descent (where we do a gradient step and then project to $\PhiShell_t$) ensures low regret with respect to any action contained in the intersection of all the shell sets (and in particular, to any endomorphism in $\Phi(\cP)$).

\begin{algorithm}[ht]
    \caption{Shell Gradient Descent}\label{algo:improper_ogd-strong}
    \KwData{Compact convex set $\cM_{\cX} \supset \cX$, %
    step sizes $\eta_t$}

    \For{$t = 1, 2, \dots T$}{
        Receive (efficient oracle access to) the convex set $\cX_{t}$.

        Set $\vx_{t} = \Pi_{\cX_{t}}(\vx_{t-1} - \eta_{t-1} \vell_{t-1})$ (if $t = 1$, choose an arbitrary $\vx_1 \in \cX_1$).
        
        Output $\vx_t$ and receive feedback $\lt \in [-1, 1]^d$
    }
\end{algorithm}

\begin{theorem}\label{thm:improper_ogd-strong}
    Let $\cX_1, \dots, \cX_T$ be an arbitrary sequence of ``shell sets'' satisfying $\cX \subseteq \cX_t \subseteq \cB_{d}(\vzero, D)$. Then, for any sequence of losses $\vell_1, \dots, \vell_T \in [-1, 1]^d$, Algorithm~\ref{algo:improper_ogd-strong} has regret
    \[
        \max_{\vx^* \in \cX} \sum_{t=1}^T \inp{\lt, \vx_t - \vx^*} \leq \frac{D^2}{2 \eta_T} + \sum_{t=1}^T \frac{\eta_t}{2} \norm{\lt}_2^2.
    \]
\end{theorem}
\begin{proof}
Note that
\begin{align*}
    \norm{\vx_{t+1}-\vx^*}^2 \leq \norm{\Pi_{\cX_{t+1}}(\vx_t-\eta_t \lt)-\vx^*}^2&\leq \norm{\vx_t-\eta_t \lt-\vx^*}^2\\
    &=\norm{\vx_t-\vx^*}^2+\eta_t^2\norm{\lt}^2-2\eta_t  \inp{\lt,\vx_t-\vx^*}
\end{align*}
where the second inequality follows since $\vx^* \in \cX$ and therefore $\vx^* \in \cX_{t}$ for all $t$. 
Thus,
\begin{align*}
    \inp{\lt, \vx_t-\vx^*} &\leq \frac{1}{2\eta_t}\p{\norm{\vx_t-\vx^*}^2-\norm{\vx_{t+1}-\vx^*}^2}+\frac{\eta_t }{2} \norm{\lt}^2
\end{align*}
Summing over $t$
\begin{align*}
    \sum_{t=1}^T (\lt(\vx_t)-\lt(\vx^*))
    &\leq \sum_{t=1}^T \norm{\vx_t-\vx^*}^2 \p{\frac{1}{2\eta_t}-\frac{1}{2\eta_{t-1}}}+ \sum_{t=1}^T \frac{\eta_t}{2}\norm{\lt}^2\\
    &\leq \frac{D^2}{2\eta_T}+ \sum_{t=1}^T \frac{\eta_t}{2}\norm{\lt}^2
\end{align*}
under the convention $1/\eta_0 = 0$, since $\norm{\vx_t-\vx^*} \leq D$.
\end{proof}

To successfully apply our Shell Gradient Descent algorithm, we would like the transformations we play to have fixed points in $\cP$. In particular, it would be ideal if we chose our $t$-th shell set $\PhiShell_t$ in such a way that this projected action belongs to $\PhiFP(\cP)$.

This is exactly the goal of the Shell Projection procedure (Algorithm \ref{algo:noisyproj-strong}). This procedure takes as input a linear transformation $\phi$ and constructs a shell set $\PhiShell$ containing $\Phi(\cP)$ with the property that the projection of $\phi$ onto $\PhiShell$ has a fixed point in $\cP$. 

\begin{figure}
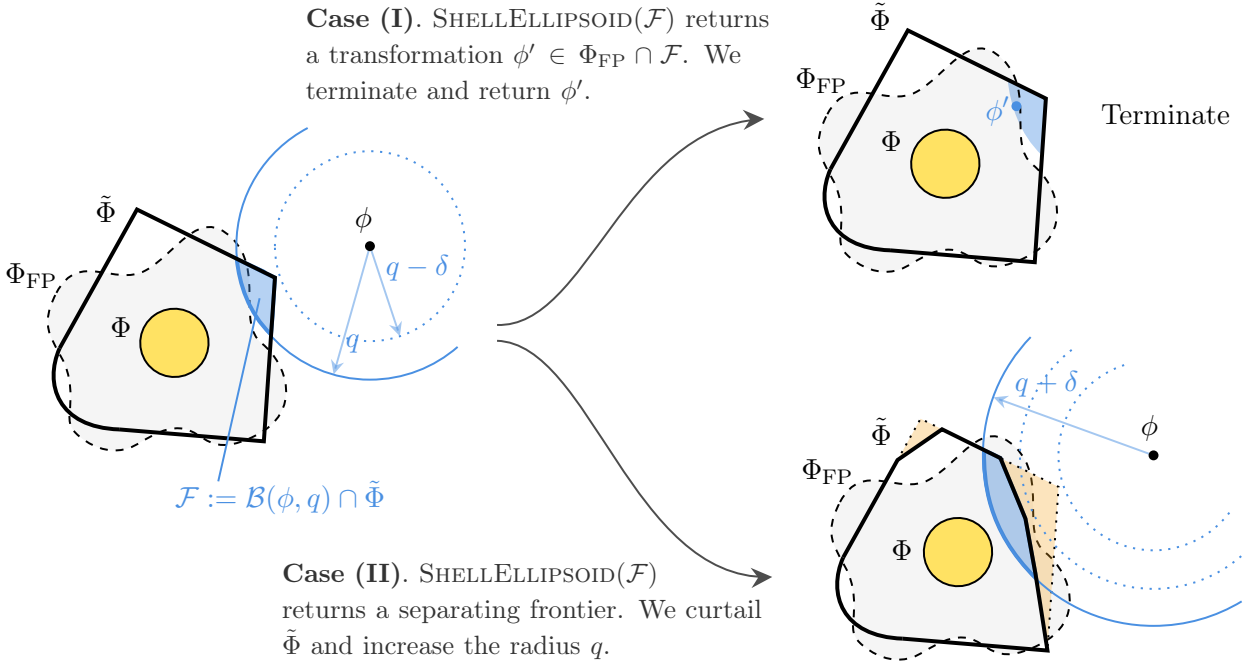

    \centering
    \tikzset{every picture/.style={line width=0.75pt}} %

\begin{tikzpicture}[x=0.75pt,y=0.75pt,yscale=-1,xscale=1]

\begin{scope}[yshift=-20]
    \input{figs/algo_step1v2}
\end{scope}

\begin{scope}[yshift=10]
    \input{figs/algo_step2av2}
\end{scope}

\begin{scope}[yshift=-20]
    \input{figs/algo_step2bv2}
\end{scope}

\draw [black!70]   (257,215.8) .. controls (302.54,217.38) and (321.03,117.04) .. (389.7,111.93) ;
\draw [shift={(391.8,111.8)}, rotate = 177.4] [fill=black!70  ][line width=0.08]  [draw opacity=0] (10.72,-5.15) -- (0,0) -- (10.72,5.15) -- (7.12,0) -- cycle    ;
\draw [black!70]   (257,223.8) .. controls (304.14,225.38) and (324.98,343.8) .. (393.7,343.06) ;
\draw [shift={(395.8,343)}, rotate = 177.4] [fill=black!70  ][line width=0.08]  [draw opacity=0] (10.72,-5.15) -- (0,0) -- (10.72,5.15) -- (7.12,0) -- cycle    ;

\node[black!80,text width=6.3cm] at (280, 80) {\small \textbf{Case (I)}. \shellelpsd{}$(\cF)$ returns a transformation $\phi' \in \Phi_\text{FP} \cap \cF$. We terminate and return $\phi'$.};

\node[black!80,text width=6.4cm] at (270, 360) {\small \textbf{Case (II)}. \shellelpsd{}$(\cF)$\\ returns a separating frontier. We curtail $\tilde\Phi$ and increase the radius $q$.};

\end{tikzpicture}
    \caption{Visual depiction of a generic step of the \shellproj{}$(\phi)$, \Cref{algo:noisyproj-strong}. %
    }
    \label{fig:shell-proj}
\end{figure}

To accomplish this, this procedure makes careful use of the Shell Ellipsoid procedure from the previous section. 
Specifically, we use $\shellelpsd$ to attempt to semi-separate a ball of radius $q$ around $\phi$ from $\Phi(\cP)$ that we slowly expand over time. Specifically, whenever the separation succeeds, we update the shell set with the new separating hyperplanes returned by $\shellelpsd$, and expand the radius $q$ slightly. We terminate when Shell Ellipsoid returns a transformation $\phi'$ in the radius $q$ ball with a fixed point. Since $q$ is incremented gradually and the shell set is cut for each previous step we failed to find such a $\phi'$, we end up with a shell set $\PhiShell$ for which $\phi'$ is an approximate projection of $\phi$. (see Figure \ref{fig:shell-proj})

\begin{algorithm}[ht]
    \caption{$\shellproj_\Phi(\phi)$ projects $\phi$ to a shell of $\Phi$}\label{algo:noisyproj-strong}
    \KwData{Convex strategy set $\cP \subset \bbR^d$ with $\cB_d(r) \subseteq \cP \subseteq \cB_d(\vzero, R)$.}
    \KwInput{Convex set $\cM$ such that $\Phi(\cP) \subseteq \cM \subseteq \cB_{d \times (d+1)}(\vzero, D)$, affine transformation $\phi \in \cB_{d \times (d+1)}(\vzero, D)$, precision parameter $\epsilon \in (0, 1)$}
    \KwOutput{Convex set $\PhiShell$ satisfying $\Phi(\cP) \subseteq \PhiShell \subseteq \cM$ and an affine transformation $\phi' \in \cM$ satisfying $\norm{\phi' - \Pi_{\PhiShell}(\phi)}_F \leq \epsilon$ and $\phi' \in \PhiFP(\cP)$}
    
    \DontPrintSemicolon

    Set $\epsilon' := \frac{\eps r}{32RD^2}$ %

    Initialize $\Tilde{\Phi} := \cM$
    
    \For{$q \gets 0$, increment by $\delta = \eps/(4D)$}{

        Run $\shellelpsd(\Tilde{\Phi} \cap \cB_{d \times (d+1)}(\phi, q))$ with precision $V_{d\times(d+1)}(\epsilon')$
        
        \eIf{it finds a $\phi'$ with a fixed point inside $\cP$}{
            
            \Return $\Tilde{\Phi}$ and $\phi'$
        }{
            We receive a polytope $\cQ$ with the property that $\Phi(\cP) \subseteq \cQ$, $\vol{\cQ \cap \Tilde{\Phi} \cap \cB_{d \times (d+1)}(\phi, q)} < V_{d\times(d+1)}(\epsilon')$, and $\cQ$ has $\tilde{O}(\poly(d))$ defining hyperplanes.

            Set $\PhiShell \leftarrow \PhiShell \cap \cQ$

        }
    }
\end{algorithm}

\begin{theorem}[Shell Projection]\label{thm:shellyproj-strong}
    Let $\cP$ be a convex set with $\cB_d(r) \subseteq \cP \subseteq \cB_d(\vzero, R)$ and $\cM$ be a convex set such that $\Phi(\cP) \subseteq \cM \subseteq \cB_{d \times (d+1)}(\vzero, D)$. For any $\phi \in \cB_{d \times (d+1)}(\vzero, D)$ and $\epsilon > 0$, Algorithm~\ref{algo:noisyproj-strong} runs in time $\poly(d, \epsilon^{-1}, R/r, D)$ and returns both: %
    
    \begin{itemize}
        \item \sloppy{A ``shell set'' $\PhiShell$ satisfying $\Phi(\cP) \subseteq \PhiShell$ constructed by intersecting $\cM$ with at most $\poly(d, \epsilon^{-1}, R/r, D)$ half-spaces.}

        \item A transformation $\phi' \in \PhiShell$ such that $\norm{\phi' - \Pi_{\PhiShell}(\phi)}_F \leq \epsilon$, and $\phi'$ has a fixed point in $\cP$.
    \end{itemize}
\end{theorem}
(For the weak version of this, see Proposition~\ref{prop:shellyproj}.)
\begin{proof}
Consider the state of $q$ in Algorithm \ref{algo:noisyproj-strong} at the termination of the algorithm (the algorithm is guaranteed to terminate since for large enough $q$, $\cB_{d \times (d+1)}(\phi, q)$ must intersect $\Phi(\cP)$). At this point we know that:

\begin{itemize}
    \item The shell set $\PhiShell$ has the property that it barely intersects with the ball of radius $q-\eps/(4D)$ around $\phi$: that is, $\vol{\PhiShell \cap \cB_{d\times(d+1)}(\phi, q-\eps/(4D)))} < V_{d\times(d+1)}(\epsilon').$
    \item The last found $\phi'$ has the property that $\phi' \in \PhiFP(\cP)$, $\phi' \in \PhiShell$ %
\end{itemize}

We will now use these facts to argue that $\phi'$ is close to the projection of $\phi$ onto $\PhiShell$. Intuitively, this should be the case, since we know that $\phi'$ is a point in $\PhiShell$ within distance $q$ of $\phi$, and almost no points in $\PhiShell$ are within distance $q - \eps/(4D)$ of $\phi$. 

In fact, we can use the fact that the set of linear endomorphisms $\Phi(\cP)$ contains a ball of radius $r/2R$ (Lemma \ref{lem:lintrans-nice}) to show that this first bullet point implies that no point in $\PhiShell$ is within distance $q - \frac{\eps}{4D}$ of $\phi$ (this follows as a consequence of Lemma \ref{lem:min-ball-intersection}).

In particular, this means that $\norm{\phi - \Pi_{\PhiShell}(\phi)} \geq q - \frac{\epsilon}{4D}$. But in addition, by standard properties of projections (see e.g. \cite{orabona2022:modern}), since $\phi' \in \PhiShell$, we know that

$$\norm{\phi - \phi'}^2 \geq \norm{\phi - \Pi_{\PhiShell}(\phi)}^2 + \norm{\Pi_{\PhiShell}(\phi) - \phi'}^2.$$

Rearranging, this implies that

$$\norm{\Pi_{\PhiShell}(\phi) - \phi'}^2 \leq \norm{\phi - \phi'}^2 - \norm{\phi - \Pi_{\PhiShell}(\phi)}^2 \leq q^2 - \left(q-\frac{\epsilon}{4D}\right)^2 \leq \epsilon,$$

\noindent
as $q$ will never exceed $2D$ in the algorithm, as $\phi$ is at most $2D$ away from $\Phi(\cP)$. The other guarantees of the algorithm follow naturally from the guarantees of $\shellelpsd$ (Lemma \ref{lem:shellelpsd-strong}) and the fact that $q$ increments at most $8D^2/\eps$ times.
\end{proof}

\subsection{Minimizing Linear Swap Regret}\label{sec:mainalgo}

We finally assemble the results from the previous sections, and present a learning algorithm that obtains $O(\poly(d)\sqrt{T})$ swap regret while running in polynomial time per iteration. Essentially, this algorithm (Algorithm \ref{algo:main-strong}) runs the Shell Gradient Descent algorithm (Algorithm \ref{algo:improper_ogd-strong}) to solve the dual regret-minimization problem constructed in the reduction of Section \ref{sec:linear-to-external}, but using the Shell Projection procedure (Algorithm \ref{algo:noisyproj-strong}) to construct a sequence of shell sets $\PhiShell_t$ of $\cP$ with the property that the desired projection onto $\PhiShell_t$ has a fixed point. 

\begin{algorithm}[ht]
    \caption{Linear-swap regret minimizer for convex strategy sets}\label{algo:main-strong}
    \KwData{Convex strategy set $\cP \subset \bbR^d$ in Isotropic Position (Definition~\ref{defn:isotropic}).}

    Let $\cM := \cB_{d \times (d+1)}(\vzero, 4 d^2)$ \qquad \tcp{this guarantees that $\Phi(\cP) \subset \cM$}

    Set step size $\beta := \frac{4}{\sqrt{T}}$ and parameter $\epsilon := \frac{1}{16 d^4 T^2}$

    Let $\phi_1 = \mathrm{Id}$ and $\vp_1$ be any point in $\cP$

    \For{$t = 1, 2, \dots, T$}{
        Output $\pt \in \cP$ and receive feedback $\lt \in [-1, 1]^d$

        Set $\mat{L}_t \in \bbR^{d \times (d+1)}$ so that for any affine transformation $\phi \in \bbR^{d \times (d+1)}$, $\inp{\mat{L_t}, \phi} = \inp{\phi(\vp_t), \vell_t}$.\hspace{-1cm}

        Run $\shellproj_\Phi(\phi_t - \beta \mat{L}_t)$ with input $\cM$ and precision $\epsilon$, receiving a shell set $\PhiShell_{t+1}$ and transformation $\phi_{t+1} \in \PhiShell_{t+1} \cap \PhiFP(\cP)$.

        \tcp{Note: the shell set $\PhiShell_{t+1}$ is only used in the regret analysis}

        Compute a fixed point $\vp_{t+1} \in \cP$ of $\phi_{t+1}$.
    }
\end{algorithm}

\begin{theorem}\label{thm:linswap-regret-main-algo-strong}
    Let $\cP \subset \bbR^d$ be a compact convex set in isotropic position (Definition~\ref{defn:isotropic}) and assume $\lt \in [-1, 1]^d$.
    Algorithm~\ref{algo:main-strong} has a per-iteration time complexity of $\poly(d, T)$ and guarantees that
    \[
        \linswap(\vp, \vell) = \sum_{t=1}^T \inp{\pt, \lt} - \min_{\phi \in \Phi(\cP)} \sum_{t=1}^T \inp{\phi(\pt), \lt} = O\left( d^4 \sqrt{T} \right).
    \]
    
\end{theorem}
(For the weak version of this, see Theorem~\ref{thm:linswap-regret-main-algo}.)
\begin{proof}
First note that for $\cP$ in isotropic position, it holds $\cB_d(\vzero, 1) \subseteq \cP \subseteq \cB_d(\vzero, n+1)$. Then, by Lemma~\ref{lem:lintrans-nice}, we get that $\Phi(\cP)$ is well-bounded between balls of radius $R_\phi = 4 d^2$ and $r_\phi = 1 / 4 (d + 1)$.

Algorithm \ref{algo:improper_ogd-strong} produces the sequence $\phi_t$ by essentially running the Shell Gradient Descent algorithm (Algorithm \ref{algo:improper_ogd-strong}) for the shell sets $\cX_t = \PhiShell_t$ and constant rate $\eta_t = \beta$, with the small caveat that instead of playing the exact projected action $\vx_t = \Pi_{\PhiShell_t}(\phi_t - \beta \mat{L}_t)$, it plays an action $\phi_t$ with the guarantee (by Theorem~\ref{thm:shellyproj-strong}) that $\norm{\phi_t - \vx_t} \leq \epsilon$. By the analysis of Shell Gradient Descent (and since $\Phi(\cP)$ belongs to each shell set $\PhiShell_t$), we therefore have that
\[
\reg^{\text{dual}}(\mathbf{\phi}, \vp, \vell) &= \sum_{t=1}^T \inp{\phi_t(\pt), \lt} - \min_{\phi^* \in \Phi(\cP)} \sum_{t=1}^T \inp{\phi^*(\pt), \lt}\\
    &= \max_{\phi^* \in \Phi(\cP)} \sum_{t=1}^T \inp{\mat{L}_t, \phi_t - \phi^*} \\
    &\leq \frac{R_{\phi}^2}{2 \beta} + \sum_{t=1}^T \frac{\beta}{2} \norm{\mat{L}_t}_F^2 + \epsilon \sum_{t=1}^{T}\norm{\mat{L}_t}_F = O\left( d^4 \sqrt{T} \right).
\]
(In particular, $\norm{\mat{L}_t}_F \leq \sqrt{d \norm{\lt}_\infty^2 \norm{\pt}_2^2 + \norm{\lt}_2^2} \leq \sqrt{d} (d + 1 + \epsilon) \leq 2 d^{3/2}$.)

In addition, since we choose $\vp_t$ to be a fixed point of $\phi_t$, we again have that $\linswap(\vp, \vell) = \reg^{\text{dual}}(\mathbf{\phi}, \vp, \vell)$, and thus $\linswap(\vp, \vell) = O( d^4 \sqrt{T})$. The per-iteration time complexity follows from the per-iteration time complexity of $\shellproj$ (Theorem \ref{thm:shellyproj-strong}) with the chosen $\epsilon$ and the fact that $\cP$ is in isotropic position.
\end{proof}

\begin{remark}
    We remark that the previous algorithm can also be executed when the time horizon $T$ is unknown, by employing the standard doubling trick (see e.g. \cite{orabona2022:modern}).
\end{remark}

\section{Polynomial-Time Computation of Linear Correlated Equilibria}
\label{sec:eah}

In this section, we present an algorithm that computes a high-precision linear correlated equilibrium in general convex games in oracle-polynomial time. That is, given oracle access to the strategy sets of an $n$-player convex game, the algorithm computes an $\epsilon$-approximate linear correlated equilibrium
in time (and number of oracle queries) that is polynomial in the size of the game and $\log(\frac{1}{\epsilon})$. The algorithm follows the idea of \citet{Papadimitriou2008:Computing} of casting the equilibrium computation task as a two-player zero-sum bilinear game between a ``Correlator'', who controls distributions over joint strategy profiles, and a ``Deviator'', who controls deviation profiles for all players. Computing any min-max equilibrium in this meta-game will then give us a correlation plan for the Correlator that corresponds to a valid linear correlated equilibrium.

To compute a min-max equilibrium in the bilinear zero-sum meta-game, we follow the framework of \citet{Farina2024:eah} that is based on the Ellipsoid Against Hope idea of \citet{Papadimitriou2008:Computing}. However, there are two main technical challenges in applying this framework to our case. First, the framework of \citet{Farina2024:eah} only works for polyhedral strategy sets that are well-behaved; it was left as an open problem to generalize their results to arbitrary convex strategy sets, as is the case here for the strategies of the Correlator and the Deviator. Second, we do not have direct oracle access to the strategy set of the Deviator (and instead, we only have some sort of ``semi-separation'' oracle access), which is required by the existing framework.

In this section, we begin by first providing a generalized version of the Ellipsoid Against Hope framework of \citet{Farina2024:eah} that works for general convex sets and also relaxes the requirement for oracle access to the strategy sets (Theorem~\ref{thm:ellipsoid-against-hope-strong}). Then, utilizing this new tool, we describe an algorithm (Theorem~\ref{thm:eah-equil-computation-strong}) that computes a linear correlated equilibrium in general convex games, by applying the same semi-separation ideas used for our regret minimization algorithm in Section~\ref{sec:linswap-regret}.

\subsection{Generalizing the Ellipsoid Against Hope Framework to General Convex Sets}

Let $\cX \subset \bbR^M, \cY \subset \bbR^N$ be compact convex sets with $M \gg N$ and consider the following convex program
\[
    \numberthis{eq:eah-desc-primal}\text{find}\ \vec{x} \in \mathcal{X} \qquad
    \text{s.t.}\ \min_{\vec y \in \mathcal{Y}} \vec{x}^\top \mat{A} \vec{y} \geq 0.
\]
Our goal is to compute a feasible point $\vx^* \in \cX$ in time that is polynomial only in $N$, since $M$ may be exponentially large in general (as will be the case when computing a linear correlated equilibrium in Section~\ref{sec:comp-lin-equil}, where $M$ is the number of all joint strategy profiles).
When the sets $\cX, \cY$ are polyhedral, the Ellipsoid Against Hope framework \citep{Farina2024:eah, Papadimitriou2008:Computing} computes a feasible solution $\vx^* \in \cX$ under the condition that we are given access to a ``Good-Enough-Response'' (GER) oracle: an oracle that for any $\vy \in \cY$ is guaranteed to return a succinctly-representable $\vx \in \cX$ (requiring far fewer than $M$ bits to represent) such that $\vx^\top \mA \vy \geq 0$. The idea is that, under GER oracle access, the convex program \eqref{eq:eah-desc-dual}, which resembles the dual of \eqref{eq:eah-desc-primal}, is guaranteed to be infeasible
\[
    \numberthis{eq:eah-desc-dual} \text{find}\ \vec y \in \bbR_+ \cY \qquad
    \text{s.t.}\ \max_{\vec x \in \cX} \vec{x}^\top \mat{A} \vec{y} \leq -1,
\]
where $\bbR_+ \cY := \set{c \vy \mid \vy \in \cY, c > 0}$ is the conic hull of $\cY$.
Despite the infeasibility of \eqref{eq:eah-desc-dual}, we proceed to execute the ellipsoid method on it, by using the GER oracle to repeatedly cut out subspaces where $\vx^{\top}\mat{A}\vy \geq 0$ for some specific $\vx$. Eventually, the ellipsoid method will deem that \eqref{eq:eah-desc-dual} is infeasible and we will have generated a sequence of $L$ good-enough-responses $\vx_1, \vx_2, \dots, \vx_L$. These good-enough-responses allow us to ``cover'' the original set $\cY$, in the sense that for any $\vy \in \cY$, it must be the case that $\vx_i \mat{A} \vy \geq 0$ for at least one GER $\vx_i$. It follows from the minimax theorem that there must be a convex combination $\vx^*$ of the $\vx_i$ that simultaneously guarantees $(\vx^*)^{\top}\mat{A}\vy \geq 0$ for \emph{every} $\vy \in \cY$. Moreover, we can find this $\vx^*$ (specifically, its representation as a mixture of the $\vx_i$) efficiently by solving a compressed primal program where we only search over such mixtures.

For our applications, we need to handle the two generalizations of the setting described above.

\begin{itemize}
    \item First, we need to handle generic convex sets $\cX$ and $\cY$, not just explicit polytopes. This means that we will not ever be able to entirely cover our set $\cY$ with good-enough-responses, and will have to settle for some $\eps$ slack in our approximation to equilibrium.
    \item More interestingly, since our convex set $\cY$ will be constructed from the sets of linear endomorphisms of each player's action set, we will not have direct oracle access to $\cY$. Instead, we will need to extend this framework to the setting where we only have ``semi-separation'' oracle access to $\cY$.
\end{itemize}

In this subsection, we give a generalization of the framework that addresses both of these points. In particular, we replace the requirement for a good-enough-response oracle for all $\vy \in \cY$ with the following: we have an oracle that at every point $\vy \in \bbR^N$, must either produce a good-enough-response $\vx \in \cX$, or a hyperplane separating $\vy$ from $\cY$. Then, we can replace $\cY$ in the dual program \eqref{eq:eah-desc-dual} with a convex shell set containing $\cY$ and then execute the ellipsoid method (constraining the shell set further whenever our oracle returns a separating hyperplane). This way, we end up with $L$ GER oracle responses as well as with a new, smaller, shell set $\widetilde{\cY}$ of $\cY$ and, eventually, we compute a solution $\vx^* \in \cX$ that satisfies the guarantee that $\min_{\vy \in \widetilde{\cY}} \vx^\top \mA \vy \geq -\eps$; this procedure is summarized in Algorithm \ref{alg:eah-strong}. Interestingly, when $\widetilde{\cY}$ is a strict superset of $\cY$, this guarantee is in some ways slightly stronger than the guarantee of the classic GER framework (as we provide a guarantee against all $\vy \in \widetilde{\cY}$).

\begin{algorithm}[ht]
    \SetAlgoNoLine
    \caption{Ellipsoid Against Hope for bilinear zero-sum games with general convex strategies}\label{alg:eah-strong}
    \KwData{Parameters $r_y, R_y$ such that $\cB_N(r_y) \subseteq \cY \subseteq \cB_N(\vzero, R_y)$, precision parameter $\epsilon > 0$, constant $B \geq 1$ such that $\norm{\vx^\top \mA}_2 \leq B$ for all $\vx \in \cX$}
	\KwIn{An oracle that for every $\vy \in \cB_N(\vzero, R_y)$, either produces a good-enough-response (\texttt{GER}) or a separating hyperplane (\texttt{SEP}) separating $\vy$ from $\cY$.}
	\KwOut{A sparse solution $\vec{x}^*$ of \eqref{eq:primal_cp-strong} represented as a mixture of \texttt{GER} oracle responses.}

    \tcp{Run the ellipsoid algorithm on the infeasible dual \eqref{eq:eah-desc-dual}}

    Initialize ellipsoid $\mathcal{E} := \cB_N(\vzero, R_y)$

    \Repeat{$\vol{\mathcal{E}} < V_N\left(\frac{\epsilon}{B}\right)$}{
    
        Query the oracle on the center $y$ of $\mathcal{E}$ 
            
        \uIf{the oracle returns a weak good-enough-response $(\vx, \vx^\top\mat{A})$}{
            Update $\mathcal{E}$ to the minimum volume ellipsoid containing $\mathcal{E} \cap \{\vy \in \bbR^{N} \mid (\vx^{\top}\mat{A})\vy \leq 0\}$
        }
        \uElse{
            Receive a half-space $H$ that separates $y$ from $\cY$ (\ie $\cY \subset H$, $y \not\in H$)
            
            Update $\mathcal{E}$ to the minimum volume ellipsoid containing $\mathcal{E} \cap H$
        }
    }

    Let $H_1, H_2, \dots, H_K$ be the separating half-spaces and let $\vx_1, \dots, \vx_L$ be the \texttt{GER} oracle response vectors returned by the oracle during the above process
    
    Let $\widetilde{\cY} := \cY \cap \bigcap_{i=1}^{K} H_i$

    Define $\mat{X} = [\vx_1 \mid \dots \mid \vx_L]$ and compute $\mat{X}^{\top}\mat{A}$ (this is an $L$-by-$N$ matrix whose rows are $\vx_i^\top\mat{A}$)
    
    Compute a solution $\vec{\lambda}^*$ to the following convex program \eqref{eq:primal_compressed_cp-strong}
    \[
        \tag{$P'$}\label{eq:primal_compressed_cp-strong}
        \text{find}\ &\vec{\lambda} \in \Delta^L\\
        \text{s.t.}\ &\min_{\vec y \in \widetilde{\cY}} \vec{\lambda}^\top (\mat{X}^\top \mat{A}) \vec{y} \geq -\eps
    \]

    Return the final solution $\vec{x}^* = \mat{X} \vlam^* = \sum_{i=1}^{L} \lambda^{*}_i \vx_i$. 
\end{algorithm}

\begin{restatable}{theorem}{thmeah-strong}\label{thm:ellipsoid-against-hope-strong}
    Let $\cX \subset \bbR^M, \cY \subset \bbR^N$ be compact convex sets (to which we do not necessarily have direct oracle access), and let $\epsilon > 0$ be a precision parameter. If the following conditions hold:
    \begin{enumerate}
        \item $\mA \in \bbR^{M \times N}$ is a matrix such that $\forall \vx \in \cX,\ \norm{\vx^\top \mA}_2 \leq B$ for some $B \geq 1$,

        \item the set $\cY$ is well-bounded as $\cB_N(r_y) \subseteq \cY \subseteq \cB_N(\vzero, R_y)$,

        \item\label{assmpt:eah_sep_oracles-strong} there exists an oracle that for every point $\vy \in \cB_N(\vzero, R_y)$, runs in $\poly(N)$ time, and either returns a separating hyperplane (\texttt{SEP}) separating $\vy$
        from $\cY$ or a good-enough-response (\texttt{GER}) $\vx \in \cX$,

        \item the encoding lengths of the \texttt{GER} responses produced by the previous oracle are polynomially bounded (in $N$ and $\log(1/\eps)$),
    \end{enumerate}
    then, Algorithm~\ref{alg:eah-strong} runs in oracle-polynomial time and computes a solution $\vx^*$ to
    \[
        \tag{$P$}\label{eq:primal_cp-strong}
        \text{find}\ &\vec{x} \in \mathcal{X}\\
        \text{s.t.}\ &\min_{\vec y \in \mathcal{Y}} \vec{x}^\top \mat{A} \vec{y} \geq -\epsilon.
    \]
    Furthermore, $\vx^*$ is a specified as an explicit mixture of $\poly(N, \log(B/\eps))$ \texttt{GER} responses.
\end{restatable}
(For the version of this under weak oracles, see Theorem~\ref{thm:ellipsoid-against-hope}.)
\begin{proof}
If Algorithm \ref{alg:eah-strong} successfully completes, it is guaranteed to return a valid solution $\vx^*$ to \eqref{eq:primal_cp-strong} that is the convex combination of $L$ \texttt{GER} responses (since $\widetilde{\cY}$ is guaranteed to contain $\cY$, $X\vlam$ must be a solution to \eqref{eq:primal_cp-strong} for any solution $\vlam$ to \eqref{eq:primal_compressed_cp-strong}). Moreover, each individual step of the algorithm runs in time polynomial in $N$, $L$, and $\log(B/\eps)$. It therefore suffices to verify that: (1) the ellipsoid algorithm runs for at most $\poly(N, \log(B/\eps))$ iterations (and therefore also $L \leq \poly(N, \log(B/\eps))$), and (2) the final convex program \eqref{eq:primal_compressed_cp-strong} has a solution.

For the first point, note that in each iteration of our ellipsoid algorithm, we intersect $\cM$ with some half-space passing through (or excluding) the center of the current ellipsoid $\cE$. By the standard analysis of the ellipsoid algorithm, after one iteration the new volume of the minimum volume ellipsoid is at most $\exp\left(-\frac{1}{2(N+1)}\right)\vol{\mathcal{E}}$. Since the initial ellipsoid $\cE$ has volume $V_{N}(R_y)$ and the process terminates when the volume drops below $V_{N}(\eps/B)$, this subroutine runs for at most $2(N+1)\log \left((B/\eps)^{N}\right) = \poly(N, \log(B/\eps))$ iterations.

To prove the second point, let $\cE_{\text{fin}}$ be the final ellipsoid $\cE$ at the conclusion of the ellipsoid algorithm subroutine. Note that any $\vy \in \widetilde{\cY} \setminus \cE_{\text{fin}}$ must have the property that $(\vx^\top\mat{A})\vy \geq 0$ for one of the \texttt{GER} responses $\vx$, as it must have been separated from $\cE_{\text{fin}}$ in one such step. On the other hand, since $\vol{\cE_{\text{fin}}} \leq V_{N}(\eps/B)$, every $\vy' \in \widetilde{\cY} \cap \cE_{\text{fin}}$ must be within distance $\eps/B$ from some $\vy \in \widetilde{\cY} \setminus \cE_{\text{fin}}$. If $(\vx^\top\mat{A})\vy \geq 0$ for this $\vy$, then we also have that

$$(\vx^\top\mat{A})\vy' = (\vx^\top\mat{A})\vy + (\vx^\top\mat{A})(\vy' - \vy) \geq -\norm{\vx^\top\mat{A}}_2\norm{\vy' - \vy}_2 \geq -\eps.$$

We have therefore shown that given any $\vy \in \widetilde{\cY}$, it must be $\vx^\top\mat{A}\vy \geq -\eps$ for at least one \texttt{GER} response $\vx$ we received. By the minimax theorem, there exists a convex combination of such responses which guarantees this property simultaneously for all $\vy \in \widetilde{\cY}$, and therefore the convex program \eqref{eq:primal_compressed_cp-strong} is feasible. To make a separation oracle for \eqref{eq:primal_compressed_cp-strong}, we can run the ellipsoid algorithm in time $\poly(\log(R_y / r_y), L)$. Thus, the the time complexity will be oracle-polynomial in $\poly(N, \log(R_y / r_y), \log(B / \epsilon))$.
\end{proof}

\subsection{Computing a Linear Correlated Equilibrium in General Convex Games}\label{sec:comp-lin-equil}

We are now ready to apply our new Ellipsoid Against Hope framework to construct an algorithm that computes an $\epsilon$-approximate linear correlated equilibrium in general convex $n$-player games with efficient oracle access to the strategy sets $\cP_1, \dots, \cP_n$. The algorithm will run in time polynomial in the size of the game and $\log(\frac{1}{\epsilon})$.

We begin with some definitions and notation. First, for the sake of handling affine linear transformations, we augment the players' strategy sets as $\cX_i = \set{1} \times \cP_i \subset \bbR^{d_i + 1}$ for every $i \in [n]$. We thus assume, without loss of generality, that utilities are multi-linear in this new space, with $u_i(\vs_{-i})[d_i + 1] = 0$ for all $i \in [n]$.
Define $M = \prod_{i=1}^n (d_i + 1)$ and $N = 1 + \sum_{i=1}^n d_i (d_i + 1)$.
Let $\cX := \conv(\cX_1 \otimes \dots \otimes \cX_n) \subset \bbR^M$ be the set of all correlated strategy profiles and $\cY := \set{1} \times \Phi(\cP_1) \times \dots \Phi(\cP_n) \subset \bbR^N$ be the set of all players' deviation profiles (note that while $M$ is exponentially large in the size of the game, $N$ is merely polynomially large). Additionally, by the multi-linearity of utility functions, we can represent each $u_i$ as an $M$-dimensional tensor $\mU_i$ such that $u_i(\vx) = \mU_i \cdot \vx$. For $k \in [d_1+1] \times \dots \times [d_n+1]$, we will refer to the $k$-th entry of the $\mU_i$ tensor as $\mU_i[k]$. Finally, in usual game-theoretic convention, for sets $S_1, \dots, S_n$ and some $\vs \in S_1 \times \dots \times S_n$, we will denote $\vs_{-i} = (s_1, \dots, s_{i-1}, s_{i+1}, \dots, s_n)$ and $\vs = (s_i, \vs_{-i})$.

The only structural assumption that we impose on the games we consider is the \emph{polynomial utility gradient property}, that lets us compute the counterfactual payoff function faced by any individual player.

\begin{assumption}[Polynomial utility gradient property {\citep[Assumption 4.2]{Farina2024:eah}}]\label{asmpt:polynomial-utility-gradient}
    Given a product distribution $\vx \in \cX_1 \otimes \dots \otimes \cX_n$, it is possible to compute the value of
    \[
        \vec{g}_i(\vx_{-i}) = \E_{\vs_{-i} \sim \vx_{-i}}[\nabla u_i(\vs_{-i})]
    \]
    for all players $i \in [n]$ in polynomial time in the encoding length of $\vx$ and the size of the game.
\end{assumption}

This property has been assumed in all previous results regarding computation of equilibria in games (including in \citet{Papadimitriou2008:Computing} as the ``Polynomial Expectation Property'') and is also implicitly assumed in every approximate equilibrium computation result that uses no-regret learning. Note that any algorithm for computing an equilibrium in polynomial time in the number of players has to impose a similar structural assumption on the utilities, because otherwise it is prohibitive to even represent the exponentially large $M$-dimensional utility tensor.

Following \citet{Papadimitriou2008:Computing, Farina2024:eah}, we now define a meta-game: a two-player zero-sum bilinear game between a Correlator (with strategy set $\cX$) and a Deviator (with strategy set $\cY$). The utility matrix  $\mA \in \bbR^{M \times N}$ of the Correlator is such that for every $k \in [d_1+1] \times \dots \times [d_n+1]$ and $j \in \set{\emptyseq} \cup \set{(i, a, b) \mid i \in [n], a \in [d_i+1], b \in [d_i]}$,
\[
    \mA_{k j} := \left\{ \begin{array}{ll}
        \sum_i \mU_i[k] & j = \emptyseq \\
        -\mU_i[(b, k_{-i})] & a = k_i\\
        0 & \text{otherwise}
    \end{array} \right. .
\]
This is a huge utility matrix (exponentially big in the number of players) and will never be explicitly stored by our algorithm. We will only ever access it via Assumption~\ref{asmpt:polynomial-utility-gradient}. The following lemma shows that the utility of the Correlator player corresponds to the sum of expected $\Phi$-regrets of the players for the specific deviation functions $\phi_1, \phi_2, \dots, \phi_n$ comprising $\vy$. In particular, this means that to compute a linear correlated equilibrium, it suffices to compute a min-max optimal solution $\vx^*$ for the Correlator.

\begin{restatable}{lemma}{lemeahutilitymat-strong}\label{lem:eahutilitymat-strong}
    Let $\cP_1, \dots, \cP_n$ be the compact convex strategy sets of an $n$-player convex game and define $\cX_i := \set{1} \times \cP_i$. For any $\vx \in \cX := \conv(\cX_1 \otimes \dots \otimes \cX_n)$ and $\vy = (1, \phi_1, \dots, \phi_n) \in \cY := \set{1} \times \Phi(\cP_1) \times \dots \times \Phi(\cP_n)$, it holds that\footnote{We slightly abuse notation here by writing $\E_{\vs \sim \vx}[\cdot]$ to refer to sampling $\vs \in \cX_1 \times \dots \times \cX_n$ from some convex decomposition of $\vx$ into pure strategy profiles. This decomposition is not unique (and so this is not always well-defined), but since we only take expectations over quantities which are linear in $\vx$, it is consistent in our case.}
    \[
        \vx^\top \mA \vy = \sum_{i=1}^n \E_{\vs \sim \vx}[u_i(\vs) - u_i(\phi_i(\vs_i), \vs_{-i})]
    \]
    Moreover, if $\vx$ satisfies $\vx^{\top}\mat{A}\vy \geq -\eps$ for every $\vy \in \cY$, then the corresponding distribution $\mu \in \Delta(\cP_1 \times \dots \times \cP_n)$ (obtained by removing the constant coefficient from each component of $\cX$ and decomposing into a distribution over product elements) is an $\eps$-approximate linear correlated equilibrium.
\end{restatable}
\begin{proof}
    Let $\mB_i \in \bbR^{d_i \times (d_i+1)}$ be the matrix corresponding to the affine transformation $\phi_i$. Then,
    \[
        \vx^\top \mA \vy &= \sum_{k=1}^M \vx[k] \sum_{i=1}^n \left( \mU_i[k] - \sum_{b \in [d_i]} \mB_i[b, k_i] \mU_i[(b, k_{-i})] \right)\\
            &= \sum_{i=1}^n \left( \E_{\vs \sim \vx}[u_i(\vs)] - \sum_{k_{-i}} \sum_{k_i \in [d_i+1]} \vx[(k_i, k_{-i})] \sum_{b \in [d_i]} \mB_i[b, k_i] \mU_i[(b, k_{-i})] \right)\\
            &= \sum_{i=1}^n \E_{\vs \sim \vx}[u_i(\vs) - u_i(\phi_i(\vs_i), \vs_{-i})].
    \]

    Now, assume $\vx$ satisfies $\vx^{\top}\mat{A}\vy \geq -\eps$ for every $\vy \in \cY$. Note that if for any player $i \in [n]$ and endomorphism $\phi \in \Phi(\cP_i)$, we choose $\vy$ by setting $\phi_i = \phi$ and $\phi_j = \mat{Id}$ (the identity map on $\cP_j$) for all $j \neq i$, this implies that
    \[
        \E_{\vs \sim \vx}[u_i(\vs)] \geq \E_{\vs \sim \vx}[u_i(\phi_i(\vs_i), \vs_{-i})]  -\eps,
    \]
    and in particular, the correlated strategy profile $\mu$ corresponding to the element $\vx$ is an $\eps$-approximate linear correlated equilibrium.
\end{proof}

Finally, we present the main result of this section. The main technical idea needed for this result is to develop an oracle of the form that is required by Theorem~\ref{thm:ellipsoid-against-hope-strong}. The way we do this, at a high level, is as follows: for a given deviation point $\vy \in \bbR^N$, we check if each of the component transformations $\phi_i$ has a fixed point in $\cX_i$. If any of the $\phi_i$ do not, we can query the semi-separation oracle (Section~\ref{sec:semi-sep}) to separate $\phi_i$ from $\cY$, and we use this separating hyperplane as our response. Otherwise, if every $\phi_i$ possesses a fixed point, we can construct a good-enough-response by taking the product distribution formed from these fixed points.

\begin{restatable}{theorem}{thmequilcomputation}\label{thm:eah-equil-computation-strong}
    Let $G$ be an $n$-player convex game with convex strategy sets $\cP_i \subset \bbR^{d_i}$ for $i \in [n]$. Assume we have efficient oracle access to each $\cP_i$ and each $\cP_i$ is well-bounded via $\cB_{d_i}(r_i) \subseteq \cP_i \subseteq \cB_{d_i}(\vzero, R_i)$. Furthermore, assume that $G$ satisfies the polynomial utility gradient property (Assumption~\ref{asmpt:polynomial-utility-gradient}), and assume that $u_i(\vs) \in [-1, 1]$ for every strategy profile $\vs \in \cX_1 \times \dots \times \cX_n$. Then, there exists an algorithm that computes an $\epsilon$-approximate linear correlated equilibrium in time $\poly(\sum_i d_i, \log(1/\eps), \sum_i \log(R_i / r_i))$. Furthermore, the computed equilibrium is represented as a mixture of polynomially many product distributions over strategy profiles.
\end{restatable}
(For the version of this under weak oracles, see Theorem~\ref{thm:eah-equil-computation-weak}.)
\begin{proof}
By Lemma~\ref{lem:eahutilitymat-strong}, it suffices to construct the oracle required by Algorithm \ref{alg:eah-strong} for the sets $\cX, \cY$ and matrix $\mat{A}$ that arise from the Correlator-Deviator game. This oracle needs to take an arbitrary deviation $\vy \in \bbR^{N}$ and return either a hyperplane separating $\vy$ from the true Deviator action set $\cY$ or a good-enough-response $\vx$.

Write $\vy = (1, \phi_1, \phi_2, \dots, \phi_n)$ (if the first coordinate of $\vy$ is not $1$, we can trivially separate it from $\cY$). For each $\phi_i$, run the semi-separation oracle of Lemma \ref{lem:sep-non-fixed-point-strong} on this $\phi_i$. This oracle runs in time $\poly(d_i)$ and either returns a fixed point $\vx_i$ contained in $\cX_i$, or a halfspace $H_i$ separating $\phi_i$ from $\cY_i$.

If we received a halfspace $H_i$ note that this same halfspace separates $\vy$ from $\cY$ (after appropriately extending it to the higher dimensional space), and so we can simply return this halfspace. Therefore, assume that we received a fixed point $\vx_i \in \cX_i$ of $\phi_i$ for all $i \in [n]$, and consider the point $\vx = \vx_1 \otimes \vx_2 \otimes \dots \otimes \vx_n$. Note that by Lemma \ref{lem:sep-non-fixed-point-strong}, this $\vx$ has the property that

\[
\vx^\top \mA \vy = \sum_{i=1}^n u_i(\vx) - u_i(\phi_i(\vx_i), \vx_{-i}) = 0\]

\noindent
(since $\phi_i(\vx_i) = \vx_i$) and therefore $\vx$ is a valid good-enough-response. Moreover, the representation of $\vx$ as the product of $n$ mixed strategies only requires $O(\sum_{i}d_i)$ bits to represent.

Finally, note that, since by assumption $|u_i(\vs)| \leq 1$, this implies that $\norm{\vx^\top \mA} \leq \sqrt{N}$, and so we call Algorithm \ref{alg:eah} with $B = \sqrt{N}$. It follows that Algorithm 1 runs in time $\poly(N, \log(B/\eps), \sum_i \log(R_i / r_i)) = \poly(\sum_{i} d_i, \log(1/\eps), \sum_i \log(R_i / r_i))$, as desired.
\end{proof}

\newpage
\bibliographystyle{unsrtnat}
\bibliography{references}

\begin{thebibliography}{39}
\providecommand{\natexlab}[1]{#1}
\providecommand{\url}[1]{\texttt{#1}}
\expandafter\ifx\csname urlstyle\endcsname\relax
  \providecommand{\doi}[1]{doi: #1}\else
  \providecommand{\doi}{doi: \begingroup \urlstyle{rm}\Url}\fi

\bibitem[Gordon et~al.(2008)Gordon, Greenwald, and Marks]{Gordon08:No}
Geoffrey~J Gordon, Amy Greenwald, and Casey Marks.
\newblock No-regret learning in convex games.
\newblock In \emph{International Conference on Machine learning}, pages
  360--367, 2008.

\bibitem[Papadimitriou and Roughgarden(2008)]{Papadimitriou2008:Computing}
Christos~H. Papadimitriou and Tim Roughgarden.
\newblock Computing {C}orrelated {E}quilibria in {M}ulti-{P}layer {G}ames.
\newblock \emph{Journal of the ACM}, 55\penalty0 (3), 2008.

\bibitem[Farina and Pipis(2024{\natexlab{a}})]{Farina2024:eah}
Gabriele Farina and Charilaos Pipis.
\newblock {Polynomial-Time Computation of Exact $\Phi$-Equilibria in Polyhedral
  Games}, 2024{\natexlab{a}}.
\newblock URL \url{https://arxiv.org/abs/2402.16316}.

\bibitem[Daskalakis et~al.(2009)Daskalakis, Goldberg, and
  Papadimitriou]{Daskalakis2009:The}
Constantinos Daskalakis, Paul~W. Goldberg, and Christos~H. Papadimitriou.
\newblock The complexity of computing a nash equilibrium.
\newblock \emph{Commun. ACM}, 52\penalty0 (2):\penalty0 89–97, feb 2009.
\newblock ISSN 0001-0782.

\bibitem[Chen et~al.(2009)Chen, Deng, and Teng]{chen2009settling}
Xi~Chen, Xiaotie Deng, and Shang-Hua Teng.
\newblock Settling the complexity of computing two-player nash equilibria.
\newblock \emph{Journal of the ACM (JACM)}, 56\penalty0 (3):\penalty0 1--57,
  2009.

\bibitem[Rubinstein(2016)]{Rubinstein2016:Settling}
Aviad Rubinstein.
\newblock Settling the complexity of computing approximate two-player nash
  equilibria.
\newblock In \emph{2016 IEEE 57th Annual Symposium on Foundations of Computer
  Science (FOCS)}, pages 258--265, 2016.

\bibitem[Kearns et~al.(2001)Kearns, Littman, and Singh]{KearnsLS01}
Michael~J. Kearns, Michael~L. Littman, and Satinder Singh.
\newblock {Graphical Models for Game Theory}.
\newblock In \emph{Proceedings of the 17th Conference in Uncertainty in
  Artificial Intelligence (UAI)}, 2001.

\bibitem[Howson~Jr(1972)]{howson1972equilibria}
Joseph~T Howson~Jr.
\newblock Equilibria of polymatrix games.
\newblock \emph{Management Science}, 18\penalty0 (5-part-1):\penalty0 312--318,
  1972.

\bibitem[Rosenthal(1973)]{rosenthal1973class}
Robert~W Rosenthal.
\newblock A class of games possessing pure-strategy nash equilibria.
\newblock \emph{International Journal of Game Theory}, 2:\penalty0 65--67,
  1973.

\bibitem[Jiang and Leyton-Brown(2015)]{jiang2015polynomial}
Albert~Xin Jiang and Kevin Leyton-Brown.
\newblock Polynomial-time computation of exact correlated equilibrium in
  compact games.
\newblock \emph{Games and Economic Behavior}, 91:\penalty0 347--359, 2015.

\bibitem[Barman and Ligett(2015)]{barman2015finding}
Siddharth Barman and Katrina Ligett.
\newblock Finding any nontrivial coarse correlated equilibrium is hard.
\newblock In \emph{Proceedings of the Sixteenth ACM Conference on Economics and
  Computation}, pages 815--816, 2015.

\bibitem[von Stengel and Forges(2008)]{vonStengel2008}
B.~von Stengel and F.~Forges.
\newblock Extensive-form correlated equilibrium: Definition and computational
  complexity.
\newblock \emph{Mathematics of Operations Research}, 33\penalty0 (4):\penalty0
  1002--1022, 2008.

\bibitem[Huang and von Stengel(2008)]{huang2008computing}
Wan Huang and Bernhard von Stengel.
\newblock Computing an extensive-form correlated equilibrium in polynomial
  time.
\newblock In \emph{International Workshop on Internet and Network Economics},
  pages 506--513. Springer, 2008.

\bibitem[Dagan et~al.(2024)Dagan, Daskalakis, Fishelson, and
  Golowich]{dagan2024external}
Yuval Dagan, Constantinos Daskalakis, Maxwell Fishelson, and Noah Golowich.
\newblock From external to swap regret 2.0: An efficient reduction for large
  action spaces.
\newblock In \emph{Proceedings of the 56th Annual ACM Symposium on Theory of
  Computing}, pages 1216--1222, 2024.

\bibitem[Peng and Rubinstein(2024)]{peng2024fast}
Binghui Peng and Aviad Rubinstein.
\newblock Fast swap regret minimization and applications to approximate
  correlated equilibria.
\newblock In \emph{Proceedings of the 56th Annual ACM Symposium on Theory of
  Computing}, pages 1223--1234, 2024.

\bibitem[Morrill et~al.(2021)Morrill, D'Orazio, Sarfati, Lanctot, Wright,
  Greenwald, and Bowling]{MorrillDSLWGB21}
Dustin Morrill, Ryan D'Orazio, Reca Sarfati, Marc Lanctot, James~R. Wright,
  Amy~R. Greenwald, and Michael Bowling.
\newblock Hindsight and sequential rationality of correlated play.
\newblock In \emph{Proceedings of the 35th Conference on Artificial
  Intelligence (AAAI)}, 2021.

\bibitem[Zhang et~al.(2024{\natexlab{a}})Zhang, Anagnostides, Farina, and
  Sandholm]{zhang2024efficient}
Brian~Hu Zhang, Ioannis Anagnostides, Gabriele Farina, and Tuomas Sandholm.
\newblock {Efficient $\Phi$-Regret Minimization with Low-Degree Swap Deviations
  in Extensive-Form Games}.
\newblock In \emph{Neural Information Processing Systems}, 2024{\natexlab{a}}.

\bibitem[Cai et~al.(2024)Cai, Daskalakis, Luo, Wei, and
  Zheng]{cai2024tractable}
Yang Cai, Constantinos Daskalakis, Haipeng Luo, Chen-Yu Wei, and Weiqiang
  Zheng.
\newblock {On Tractable $\Phi$-Equilibria in Non-Concave Games}.
\newblock In \emph{NeurIPS}, 2024.

\bibitem[Daskalakis et~al.(2024)Daskalakis, Farina, Golowich, Sandholm, and
  Zhang]{daskalakis2024lower}
Constantinos Daskalakis, Gabriele Farina, Noah Golowich, Tuomas Sandholm, and
  Brian~Hu Zhang.
\newblock A lower bound on swap regret in extensive-form games.
\newblock \emph{arXiv preprint arXiv:2406.13116}, 2024.

\bibitem[Celli et~al.(2020)Celli, Marchesi, Farina, and Gatti]{celli2020no}
Andrea Celli, Alberto Marchesi, Gabriele Farina, and Nicola Gatti.
\newblock No-regret learning dynamics for extensive-form correlated
  equilibrium.
\newblock \emph{Advances in Neural Information Processing Systems},
  33:\penalty0 7722--7732, 2020.

\bibitem[Farina et~al.(2022)Farina, Celli, Marchesi, and
  Gatti]{Farina22:Simple}
Gabriele Farina, Andrea Celli, Alberto Marchesi, and Nicola Gatti.
\newblock Simple uncoupled no-regret learning dynamics for extensive-form
  correlated equilibrium.
\newblock \emph{Journal of the ACM}, 69\penalty0 (6), 2022.

\bibitem[Farina and Pipis(2024{\natexlab{b}})]{farina2024polynomial}
Gabriele Farina and Charilaos Pipis.
\newblock Polynomial-time linear-swap regret minimization in
  imperfect-information sequential games.
\newblock \emph{Advances in Neural Information Processing Systems}, 36,
  2024{\natexlab{b}}.

\bibitem[Mansour et~al.(2022)Mansour, Mohri, Schneider, and
  Sivan]{mansour2022strategizing}
Yishay Mansour, Mehryar Mohri, Jon Schneider, and Balasubramanian Sivan.
\newblock Strategizing against learners in bayesian games.
\newblock In \emph{Conference on Learning Theory}, pages 5221--5252. PMLR,
  2022.

\bibitem[Deng et~al.(2019)Deng, Schneider, and Sivan]{deng2019strategizing}
Yuan Deng, Jon Schneider, and Balasubramanian Sivan.
\newblock Strategizing against no-regret learners.
\newblock \emph{Advances in neural information processing systems}, 32, 2019.

\bibitem[Fujii(2023)]{fujii2023bayes}
Kaito Fujii.
\newblock Bayes correlated equilibria and no-regret dynamics.
\newblock \emph{arXiv preprint arXiv:2304.05005}, 2023.

\bibitem[Dann et~al.(2023)Dann, Mansour, Mohri, Schneider, and
  Sivan]{dann2023pseudonorm}
Christoph Dann, Yishay Mansour, Mehryar Mohri, Jon Schneider, and Balubramanian
  Sivan.
\newblock Pseudonorm approachability and applications to regret minimization.
\newblock In \emph{International Conference on Algorithmic Learning Theory},
  pages 471--509. PMLR, 2023.

\bibitem[Zhang et~al.(2024{\natexlab{b}})Zhang, Farina, and
  Sandholm]{zhang2024mediator}
Brian~Hu Zhang, Gabriele Farina, and Tuomas Sandholm.
\newblock Mediator interpretation and faster learning algorithms for linear
  correlated equilibria in general extensive-form games.
\newblock 2024{\natexlab{b}}.

\bibitem[Hazan and Kale(2007)]{hazan2007computational}
Elad Hazan and Satyen Kale.
\newblock Computational equivalence of fixed points and no regret algorithms,
  and convergence to equilibria.
\newblock \emph{Advances in Neural Information Processing Systems}, 20, 2007.

\bibitem[Dann et~al.(2024)Dann, Mansour, Mohri, Schneider, and
  Sivan]{dann2024rate}
Christoph Dann, Yishay Mansour, Mehryar Mohri, Jon Schneider, and
  Balasubramanian Sivan.
\newblock Rate-preserving reductions for blackwell approachability.
\newblock \emph{arXiv preprint arXiv:2406.07585}, 2024.

\bibitem[Koller et~al.(1996)Koller, Megiddo, and {von
  Stengel}]{Koller96:Efficient}
Daphne Koller, Nimrod Megiddo, and Bernhard {von Stengel}.
\newblock Efficient computation of equilibria for extensive two-person games.
\newblock \emph{Games and Economic Behavior}, 14\penalty0 (2):\penalty0
  247--259, 1996.

\bibitem[Gr{\"o}tschel et~al.(1993)Gr{\"o}tschel, Lov{\'a}sz, and
  Schrijver]{Grotschel1993:Geometric}
Martin Gr{\"o}tschel, L{\'a}szl{\'o} Lov{\'a}sz, and Alexander Schrijver.
\newblock \emph{Geometric Algorithms and Combinatorial Optimization}.
\newblock Springer Berlin, Heidelberg, 1993.

\bibitem[Kannan et~al.(1997)Kannan, Lov{\'a}sz, and
  Simonovits]{Kannan1997:RandomWA}
Ravi Kannan, L{\'a}szl{\'o}~Mikl{\'o}s Lov{\'a}sz, and Mikl{\'o}s Simonovits.
\newblock Random walks and an o*(n5) volume algorithm for convex bodies.
\newblock \emph{Random Struct. Algorithms}, 11:\penalty0 1--50, 1997.

\bibitem[Lovasz and Vempala(2003)]{Lovasz2003:simulated}
Laszlo Lovasz and Santosh Vempala.
\newblock Simulated annealing in convex bodies and an o*(n4) volume algorithm.
\newblock Technical report, April 2003.

\bibitem[Ben-Tal and Nemirovski(2001)]{BenTal2001:Lectures}
Aharon Ben-Tal and Arkadi Nemirovski.
\newblock \emph{Lectures on Modern Convex Optimization: Analysis, Algorithms,
  and Engineering Applications}.
\newblock 01 2001.
\newblock ISBN 9780898714913.
\newblock \doi{10.1137/1.9780898718829}.

\bibitem[Abernethy et~al.(2011)Abernethy, Bartlett, and
  Hazan]{abernethy2011blackwell}
Jacob Abernethy, Peter~L Bartlett, and Elad Hazan.
\newblock Blackwell approachability and no-regret learning are equivalent.
\newblock In \emph{Proceedings of the 24th Annual Conference on Learning
  Theory}, pages 27--46. JMLR Workshop and Conference Proceedings, 2011.

\bibitem[Blackwell(1956)]{blackwell1956analog}
David Blackwell.
\newblock An analog of the minimax theorem for vector payoffs.
\newblock 1956.

\bibitem[Orabona(2022)]{orabona2022:modern}
Francesco Orabona.
\newblock {A Modern Introduction to Online Learning}, 2022.

\bibitem[Papadimitriou et~al.(2023)Papadimitriou, Vlatakis-Gkaragkounis, and
  Zampetakis]{Papadimitriou23:Kakutani}
Christos Papadimitriou, Emmanouil-Vasileios Vlatakis-Gkaragkounis, and Manolis
  Zampetakis.
\newblock The computational complexity of multi-player concave games and
  kakutani fixed points.
\newblock In \emph{Proceedings of the 24th ACM Conference on Economics and
  Computation}, EC '23, page 1045, New York, NY, USA, 2023. Association for
  Computing Machinery.
\newblock ISBN 9798400701047.

\bibitem[Kozlov et~al.(1980)Kozlov, Tarasov, and Khachiyan]{Kozlov80:quadratic}
M.K. Kozlov, S.P. Tarasov, and L.G. Khachiyan.
\newblock The polynomial solvability of convex quadratic programming.
\newblock \emph{USSR Computational Mathematics and Mathematical Physics},
  20\penalty0 (5):\penalty0 223--228, 1980.
\newblock ISSN 0041-5553.

\end{thebibliography}

\clearpage
\appendix
\section{Sufficiency of Regret Minimization in Isotropic Position}

\begin{lemma}\label{lem:isotropic_regret}
    Let $\cP \subset \bbR^d$ be a compact convex set with $\cP \subseteq \cB_d(\vzero, R)$.
    Let $\psi(\vx)$ be the (invertible) affine Isotropic transformation for $\cP$, such that $\cP' := \psi(\cP)$ is in the Isotropic Position.
    Assume that we have a linear-swap regret minimizer $\mathcal{R}_{\cP'}$ for $\cP'$ incurring regret $\text{Reg}_{\cP'}$.
    Then, using $\poly(d)$ time per-iteration, $\mathcal{R}_{\cP'}$ can be converted to a linear-swap regret minimizer $\mathcal{R}_{\cP}$ for $\cP$ that incurs a regret of $\text{Reg}_{\cP} = 2 R \cdot \text{Reg}_{\cP'}$.
\end{lemma}
\begin{proof}
    Let $\psi(\vx) := \mA \vx + \vb$ and note that $\mA$ is invertible as $\psi$ is invertible. Additionally, for all $\vz \in \cB_d(\vzero, 1)$,
    \[
        \norm{\mA^{-1} \vz}_2 = \norm{\mA^{-1} (\vz - \vb + \vb)}_2 \leq 2 R
    \]
    and thus, $\norm{\mA}_2 \leq 2R$.
    
    Let $\vell_1, \dots, \vell_T$ be the sequence of losses given to $\mathcal{R}_{\cP}$. We can transform each one to a loss $\lt' = \frac{1}{2R} (\mA^{-1})^\top \lt$ for the $\mathcal{R}_{\cP'}$ regret minimizer, which in turn will return a sequence $\vy_1, \dots, \vy_T \in \cP'$ of strategies.
    
    Define $\vx_t = \psi^{-1}(\vy_t) \in \cP$ to be the inverse of strategy $\vy_t$. Consider any $\phi \in \Phi(\cP)$ and define $\phi'(\vy) = \psi(\phi(\psi^{-1}(\vy))) \in \Phi(\cP')$. Then, we have
    \[
        \sum_{t=1}^T \inp{\vx_t - \phi(\vx_t), \lt} &= 2R \sum_{t=1}^T \frac{1}{2R}\lt^\top \mA^{-1} [\mA (\vx_t - \phi(\vx_t)) + \vb - \vb]\\
            &= 2R \sum_{t=1}^T \frac{1}{2R} \lt^\top \mA^{-1} (\psi(\vx_t) - \psi(\phi(\vx_t)))\\
            &= 2R \sum_{t=1}^T \inp{\vy_t - \phi'(\vy_t), \lt'}\\
            &\leq 2R \cdot \text{Reg}_{\cP'}.
    \]
    This concludes the proof that if $\mathcal{R}_{\cP'}$ is a linear-swap regret minimizer, then $\mathcal{R}_{\cP}$ will be a linear-swap regret minimizer too.
\end{proof}

\section{Properties of Sets of Linear Transformations}\label{app:lintrans}

In this section, we establish some basic properties of sets $\lintrans{\cA}{\cB}$ of linear transformations, ultimately establishing Lemma \ref{lem:lintrans-nice} (that the set $\Phi(\cP)$ is nicely behaved if $\cP$ is).

\begin{lemma}\label{lem:subset-lintrans}
    Let $\cA, \cB, \cal{C}, \cal{D}$ be compact convex sets. If $\cA \supseteq \cal{C}$ and $\cB \subseteq \cal{D}$, then $\lintrans{\cA}{\cB} \subseteq \lintrans{\cal{C}}{\cal{D}}$.
\end{lemma}
\begin{proof}
    For any $(\mM, \vx_0) \in \lintrans{\cA}{\cB}$, it holds that $\forall \vx \in \cA : \mM \vx + \vx_0 \in \cB$. Since $\cA \supseteq \cal{C}$ and $\cB \subseteq \cal{D}$, it follows that $\forall \vx \in \cC : \mM \vx + \vx_0 \in \cB \subseteq \cD$. Thus, $(\mM, \vx_0) \in \lintrans{\cC}{\cD}$ and $\lintrans{\cA}{\cB} \subseteq \lintrans{\cC}{\cD}$.
\end{proof}

\begin{lemma}\label{lem:inscribed-lintrans}
    Let $\cP \subset \bbR^d$ be a compact convex set of strategies that is both inscribed and circumscribed, $\cB_d(\vec a, r) \subseteq \cP \subseteq \cB_d(\vzero, R)$ for some $\vec a \in \bbR^d$ and $R \geq 1$. Then, the set of its endomorphic affine transformations $\Phi(\cP)$ is inscribed with a ball $\cB_{d \times (d+1)}(\phi_{\vec a}, \frac{r}{2R}) \subseteq \Phi(\cP)$, where $\phi_{\vec a} = (\vzero, \vec a)$ is the constant transformation $\vx \mapsto \vec a$.
\end{lemma}
\begin{proof}
    It suffices to prove that $\cB_{d \times (d+1)}(\phi_{\vec a}, \frac{r}{2R}) \subseteq \lintrans{\cP}{\cB_d(\vec a, r)}$ and the result follows from Lemma~\ref{lem:subset-lintrans}.

    Consider any $(\mA, \vx_0) \in \cB_{d \times (d+1)}(\phi_{\vec a}, \frac{r}{2R})$, which implies that $\norm{\mA}_F^2 + \norm{\vx_0 - \vec a}_2^2 \leq \left( \frac{r}{2R} \right)^2$. Then, for any $\vx \in \cP$, we have
    \[
        \norm{\mA \vx + \vx_0 - \vec a}_2 \leq \norm{\mA \vx}_2 + \norm{\vx_0 - \vec a}_2 \leq \norm{\mA}_F \norm{\vx}_2 + \frac{r}{2R} \leq \frac{r}{2R} R + \frac{r}{2R} \leq r.
    \]
    We conclude that $\cB_{d \times (d+1)}(\phi_{\vec a}, \frac{r}{R}) \subseteq \lintrans{\cP}{\cB_d(\vec a, r)}$.
\end{proof}

\begin{lemma}\label{lem:lintrans-sphere-spectral}
    Let $\cB_d(\va, r) \subseteq \cB_d(\vzero, R)$. It holds that
    \[
        \lintrans{\cB_d(\va, r)}{\cB_d(\vzero, R)} \subseteq \set{(\mA, \vx_0) \in \bbR^{d \times d} \times \bbR^d : \norm{A}_2 \leq 2 R / r, \norm{\vx_0}_2 \leq R + 2R^2/r},
    \]
    where $\norm{A}_2$ denotes the spectral norm of $\mA$. Additionally, this implies that
    \[
        \lintrans{\cB_d(\va, r)}{\cB_d(\vzero, R)} \subseteq \cB_{d \times (d+1)}\left( \vzero, \frac{3R}{r} \sqrt{R^2 + d} \right).
    \]
\end{lemma}
\begin{proof}
    Consider any $(\mA, \vx_0) \in \lintrans{\cB_d(\va, r)}{\cB_d(\vzero, R)}$ and note that, by definition, it holds $\norm{\mA \vx + \vx_0}_2 \leq R$ for all $\vx \in \cB(\va, r)$. Thus, for any $\vx \in \bbR^d$ with $\norm{\vx - \va}_2 \leq r$, we have $\norm{\mA (\vx - \va)}_2 = \norm{\mA (\vx - \va) + \vx_0 - \vx_0} \leq \norm{\mA \vx + \vx_0}_2 + \norm{\mA \va + \vx_0}_2 \leq 2 R$. By setting $\vy = (\vx - \va) / r$, we can equivalently write the previous statement as
    \[
            &\norm{\mA (r \vy)}_2 \leq 2 R \quad \forall \vy \in \cB(\vzero, 1)\\
        \implies &\norm{\mA \vy}_2 \leq \frac{2 R}{r} \quad \forall \vy \in \cB(\vzero, 1)\\
        \implies &\norm{\mA}_2 \leq \frac{2 R}{r},
    \]
    which concludes the proof that all $\mA \in \lintrans{\cB_d(\va, r)}{\cB_d(\vzero, R)}$ have $\norm{\mA}_2 \leq \frac{2 R}{r}$. 

    To prove the bound for $\vx_0$, note that $\norm{\mA \vx}_2 \leq 2R^2 / r$ for any $\vx \in \cB_d(\va, r)$ and thus, if it was $\norm{\vx_0}_2 > R + 2R^2/r$, we would get $\norm{\mA \vx + \vx_0}_2 > R$.

    Finally, we have that $\norm{\mA}_F \leq \sqrt{d} \norm{\mA}_2 \leq \sqrt{d} 2R/r$ and $\norm{(\mA, \vx_0)}_F^2 = \norm{\mA}_F^2 + \norm{\vx_0}_2^2 \leq d \frac{4R^2}{r^2} + \frac{9R^4}{r^2}$. Thus, $\norm{(\mA, \vx_0)}_F \leq \frac{3R}{r} \sqrt{R^2 + d}$.
\end{proof}

\lintransnice*
\begin{proof}
    The proof is immediate by Lemma~\ref{lem:inscribed-lintrans} and Lemma~\ref{lem:lintrans-sphere-spectral}.
\end{proof}

Finally, we prove some Lemmas that are useful when dealing with weak oracles and are based on the definition of the weak set of linear endomorphisms (Definition~\ref{defn:weak-linear-endomorphisms}).

\begin{lemma}\label{lem:eps-inscribed-lintrans}
    Let $\cP \subset \bbR^d$ be a compact convex set of strategies that is both inscribed and circumscribed, $\cB_d(\vec a, r) \subseteq \cP \subseteq \cB_d(\vzero, R)$ for some $\vec a \in \bbR^d$. Then, $\cB_{d \times (d+1)}(\phi_{\vec a}, \frac{r}{4R}) \subseteq \Phi(\cP, \epsilon)$ for all $\epsilon \in (0, \frac{r}{4R})$, where $\phi_{\vec a} = (\vzero, \vec a)$ is the constant transformation $\vx \mapsto \vec a$.
\end{lemma}
\begin{proof}
    Let $r' = r - \epsilon$ and $R' = R + \epsilon$ such that $\cB_d(\va, r') \subseteq \cP^{-\epsilon}$ and $\cP^{+\epsilon} \subseteq \cB_d(\vzero, R')$. Then, following the same steps as in Lemma~\ref{lem:inscribed-lintrans}, we get $\cB_{d \times (d+1)}(\phi_{\va}, \frac{r'}{2 R'}) \subseteq \lintrans{\cP^{+\epsilon}}{\cB_d(\va, r')} \subseteq \Phi(\cP, \epsilon)$.
    
    Note also that $r' / 2R' > r / 4R \iff \epsilon < R r / (4R + 2r)$, which is true. Thus, $\cB_d(\phi_{\va}, \frac{r}{4R}) \subseteq \Phi(\cP, \epsilon)$.
\end{proof}

\begin{lemma}[Approximate Identity]\label{lem:approx-identity}
    Let $\cP \subset \bbR^d$ be a compact convex set that is well-bounded as $\cB_d(\va, r) \subseteq \cP \subseteq \cB_d(\vzero, R)$. Then, for any $\epsilon \in (0, \frac{r}{4R})$, there exists a transformation $\phi \in \Phi(\cP, \epsilon)$ such that $\norm{\vx - \phi(\vx)}_2 \leq \frac{5 R}{r} \epsilon$ for all $\vx \in \cP^{+\epsilon}$.
\end{lemma}
\begin{proof}
    Define $\phi_1(\vx) := \frac{\epsilon}{r} \va + (1 - \frac{\epsilon}{r}) \vx$. First, we prove that, for any $\vx \in \cP$, it is $\cB_d(\phi_1(\vx), \epsilon) \subseteq \cP$ and hence, $\phi_1(\vx) \in \cP^{-\epsilon}$. To prove this, consider any $\vz \in \cB_d(\phi(\vx), \epsilon)$ and let $\va' = \frac{r}{\epsilon}\left( \vz - (1 - \frac{\epsilon}{r}) \vx \right)$. We have
    \[
        \norm{\va' - \va}_2 &= \frac{r}{\epsilon} \norm{\vz - \left( 1 - \frac{\epsilon}{r} \right) \vx - \frac{\epsilon}{r} \va}_2\\
            &\leq \frac{r}{\epsilon} \norm{\vz - \phi_1(\vx)}_2 + \frac{r}{\epsilon} \norm{\phi_1(\vx) - \left( 1 - \frac{\epsilon}{r} \right) \vx - \frac{\epsilon}{r} \va}_2\\
            &= \frac{r}{\epsilon} \norm{\vz - \phi_1(\vx)}_2\\
            &\leq r
    \]
    or, equivalently, $\va' \in \cB_d(\va, r) \subseteq \cP$. But $\vz$ can be written as a convex combination of $\va', \vx \in \cP$ as $\vz = \frac{\epsilon}{r} \va' + (1 - \frac{\epsilon}{r}) \vx$ and thus, by convexity of $\cP$, $\vz \in \cP$.

    Next, we have that for $\vx \in \cP$,
    \[
        \norm{\phi_1(\vx) - \vx}_2 &= \norm{\frac{\epsilon}{r} \va + \left( 1 - \frac{\epsilon}{r} \right) \vx - \vx}_2\\
        &= \frac{\epsilon}{r} \norm{\va - \vx}_2\\
        &\leq \frac{2R}{r} \epsilon.
    \]
    Thus, $\phi_1$ is an affine transformation that maps any point $\vx' \in \cP$ to a point $\phi_1(\vx') \in \cP^{-\epsilon}$ such that $\norm{\vx' - \phi_1(\vx')}_2 \leq \frac{2R}{r} \epsilon$. By the exact same argumentation, it follows that for any $\vx \in \cP^{+\epsilon}$ it is $\phi_1(\vx) \in \cP$ and $\norm{\vx - \phi_1(\vx)}_2 \leq \frac{2R+\epsilon}{r} \epsilon \leq \frac{3R}{r} \epsilon$.

    Now set $\phi(\vx) := \phi_1(\phi_1(\vx))$. Then, $\phi \in \Phi(\cP, \epsilon)$ because for all $\vx \in \cP^{+\epsilon}$ it is $\phi(\vx) \in \cP^{-\epsilon}$ and additionally,
    \[
        \norm{\vx - \phi(\vx)}_2 &\leq \norm{\vx - \phi_1(\vx)}_2 + \norm{\phi_1(\vx) - \phi(\vx)}_2\\
            &\leq \frac{5R}{r} \epsilon,
    \]
    which concludes the proof.
\end{proof}

\section{Weak Oracles for Convex Optimization}

\subsection{Definitions and Basic Results}

In this paper we work with general convex sets. Since these sets might lead to irrational solutions and other numerical peculiarities, the usual way to deal with them is to assume some form of weak access to the sets and to also weaken the optimization objectives \citep[Chapter 2]{Grotschel1993:Geometric}. In this subsection, we define the concepts we will need.

\begin{mdframed}[nobreak=true]
    \textbf{The Weak Membership Problem \citep[Definition 2.1.14]{Grotschel1993:Geometric}}

    Given a vector $\vy \in \bbQ^n$ and a rational number $\delta > 0$, either
    \vspace{-3mm}
    \begin{itemize}
        \item[(i)] assert that $\vy \in \cK^{+\delta}$, or
        
        \item[(ii)] assert that $\vy \notin \cK^{-\delta}$.
    \end{itemize}
\end{mdframed}

\begin{mdframed}[nobreak=true]
    \textbf{The Weak Separation Problem \citep[Definition 2.1.13]{Grotschel1993:Geometric}}

    Given a vector $\vy \in \bbQ^n$ and a rational number $\delta > 0$, either
    \vspace{-3mm}
    \begin{itemize}
        \item[(i)] assert that $\vy \in \cK^{+\delta}$, or
        
        \item[(ii)] find a vector $\vc \in \bbQ^n$ with $\norm{\vc}_\infty = 1$ such that $\vc^\top \vx \leq \vc^\top \vy + \delta$ for every $\vx \in \cK^{-\delta}$.
    \end{itemize}
\end{mdframed}

\begin{mdframed}[nobreak=true]
    \textbf{The Weak Optimization Problem \citep[Definition 2.1.10]{Grotschel1993:Geometric}}

    Given a vector $\vc \in \bbQ^n$ and a rational number $\epsilon > 0$, either
    \vspace{-3mm}
    \begin{itemize}
        \item[(i)] find a vector $\vy \in \bbQ^n$ such that $\vy \in \cK^{+\epsilon}$ and $\vc^\top \vx \leq \vc^\top \vy + \epsilon$ for all $\vx \in \cK^{-\epsilon}$, or
        
        \item[(ii)] assert that $\cK^{-\epsilon}$ is empty.
    \end{itemize}
\end{mdframed}

Sometimes we refer to a weak separation oracle with parameter $\delta$ as a $\delta$-weak separation oracle, and similarly we refer to a solution to the weak optimization problem with parameter $\epsilon$ as an $\epsilon$-weak solution. An algorithm that runs in polynomial time in the size of the input and, additionally, performs a polynomial number of oracle calls is called an \textit{oracle-polynomial} time algorithm.

It is worth noting that all these three problems are known to be equivalent, under minimal extra assumptions. Most of our results in this paper are stated in terms of access to a weak separation oracle, but the next Theorem gives an additional condition that is sufficient to convert any membership oracle to a separation oracle.

\begin{lemma}
    There exists an oracle-polynomial time algorithm that solves the weak separation problem for every compact convex set $\cK \subset \bbR^d$ given by a weak membership oracle, provided that we know parameters $\va_0, r, R$ such that $\cB_d(\va_0, r) \subseteq \cK \subseteq \cB_d(\vzero, R)$.
\end{lemma}
\begin{proof}
    This Lemma directly follows from applying Theorems 4.3.2, 4.2.5, and 4.4.7 from \citet{Grotschel1993:Geometric}.
\end{proof}

Next, we provide two central theorems that allow us to weakly optimize over linear or convex functions, given access to a weak separation oracle. We remark that Theorem~\ref{thm:ellipsoid-convex-sep-to-opt} is implied by Theorems 4.4.9 and 4.3.13 of \citet{Grotschel1993:Geometric}.

\begin{theorem}[{\citet[Corollary 4.2.7]{Grotschel1993:Geometric}}]\label{thm:ellipsoid-sep-to-opt}
    There exists an oracle-polynomial time algorithm that solves the weak optimization problem for every circumscribed convex set $\cK \subseteq \cB_n(\vzero, R)$ given by a weak separation oracle.
\end{theorem}

\begin{theorem}[{\citet[variation of Theorem 4.3.13]{Grotschel1993:Geometric}}]\label{thm:ellipsoid-convex-sep-to-opt}
    There exists an oracle-polynomial time algorithm that solves the following problem:\\
    \textbf{Input:} A rational number $\epsilon > 0$, a compact convex set $\cK \subset \bbR^n$ given by a weak membership oracle, parameters $R, r$ such that $\cB_d(r) \subseteq \cK \subseteq \cB_d(\vzero, R)$, and a convex function $f : \bbR^n \to \bbR$ given by an oracle that, for every $\vx \in \bbQ^n$ and $\delta > 0$, returns a rational number $t$ such that $|f(\vx) - t| \leq \delta$.\\
    \textbf{Output:} A vector $\vy \in \cK^{+\epsilon}$ such that $f(\vy) \leq f(\vx) + \epsilon$ for all $\vx \in \cK^{-\epsilon}$.
\end{theorem}

Next, we define a weak version of approximate linear correlated equilibria. In particular, for some parameter $\eta > 0$, we relax the strategy sets of every player $i$ to be $\cP_i^{+\eta}$ and we additionally relax the set of transformations we compete with to those that map points from $\cP^{+\eta}$ to $\cP^{-\eta}$. This definition is akin to the definition of the $(\epsilon, \eta)$-approximate equilibrium in \citet{Papadimitriou23:Kakutani} but for $\Phi$-equilibria and is important to avoid precision and numerical issues when we are only given weak oracle access to the strategy ests.

\begin{definition}[$\eta$-weak linear endomorphisms]
\label{defn:weak-linear-endomorphisms}
    For a compact convex set $\cP$, we define the $\eta$-weak set of linear endomorphisms as
    \[
        \Phi(\cP, \eta) := \lintrans{\cP^{+\eta}}{\cP^{-\eta}}.
    \]
\end{definition}

Lemma~\ref{lem:subset-lintrans} implies that $\Phi(\cP_i, \eta) \subseteq \Phi(\cP_i)$, so $\Phi(\cP, \eta)$ is a relaxation of $\Phi(\cP)$.

\begin{definition}[$(\epsilon, \eta)$-approximate linear equilibrium]\label{defn:weak-approx-equil}
    Let $G$ be a convex game of $n$ players with strategy sets $\cP_1, \dots, \cP_n$.
    An $(\epsilon, \eta)$-approximate linear equilibrium for $G$ is a joint distribution $\mu \in \Delta(\cP_1^{+\eta} \times \dots \times \cP_n^{+\eta})$, such that for every player $i \in [n]$,
    \[
        \E_{\vec{s} \sim \mu}[u_i(\vec s)] \geq \E_{\vec{s} \sim \mu} [u_i(\phi(\vec{s}_i), \vec{s}_{-i})] - \epsilon, \qquad \forall \phi \in \Phi(\cP_i, \eta).
    \]
\end{definition}

Of course, this is a relaxation, compared to the strong version of $\epsilon$-approximate equilibria, where $\eta=0$, that are used when we have strong oracle access to the strategy sets. We remark that all results in this paper work even if we compete with all transformations in $\Phi(\cP_i)$, because both the transformations and the utility functions have bounded Lipschitzness. However, we choose to work with this relaxed definition to allow for uniformity with other, more general, transformations.

\subsection{Useful Primitives for Weak Oracles}

We begin by stating some Lemmas that will be useful tools for working with weak oracles. Note that Lemma~\ref{lem:colinear-projections} is stated for a general bounded set, while the proof in \citet{Papadimitriou23:Kakutani} is only for sets in $[-1, 1]^d$, but these statements are equivalent.

\begin{lemma}[{\citet[Lemma B.2]{Papadimitriou23:Kakutani}}]\label{lem:colinear-projections}
    Consider a bounded convex compact set $\cX \subseteq \bbR^d$ and an arbitrary point $\vx \in \bbR^d$. Then, it holds that $\vx, \Pi_{\cX}(\vx), \Pi_{\cX^{+\epsilon}}(\vx)$ are co-linear for every $\epsilon > 0$.
\end{lemma}

\begin{lemma}[{\citet[Lemma B.3]{Papadimitriou23:Kakutani}}]\label{lem:eps-project-closeness}
    For any convex set $\cX \subset \bbR^d$ with $\cB_d(\va, r) \subseteq \cX \subseteq \cB_d(\vzero, R)$ and a parameter $\epsilon \in (0, r)$, it holds that $\norm{\vx - \Pi_{\cX^{-\epsilon}}(\vx)}_2 \leq \frac{R}{r} \epsilon$ for every $\vx \in \cX$.
\end{lemma}

\begin{lemma}\label{lem:polytope-projection}
    Let $\cP \subset \bbR^d$ be a rational polyhedron, described by a strong separation oracle. Then there exists an oracle-polynomial time algorithm that computes an exact rational solution to the projection problem
    \[
        \hat{\vx} = \Pi_{\cP}(\vy) := \argmin_{\vx \in \cP} \norm{\vx - \vy}_2^2,
    \]
    for any $\vy \in \bbQ^d$.
\end{lemma}
\begin{proof}
    The problem can be written as
    \[
        \min\ &\frac{1}{2} \vx^\top  \vx - 2 \vy^\top \vx\\
        \text{s.t.} &\vx \in \cP,
    \]
    which is a convex quadratic programming problem with a positive definite matrix and is known to be exactly solvable via the ellipsoid method \citep{Kozlov80:quadratic}.
\end{proof}

\begin{lemma}\label{lem:inflated-polytope-sep}
    Let $\cP \subset \bbR^d$ be a rational polyhedron, described by a strong separation oracle. For any parameter $\epsilon > 0$ we can create a strong separation oracle for the set $\cP^{+\epsilon}$ that runs in oracle-polynomial time.
\end{lemma}
\begin{proof}
    Consider any input point $\vy \in \bbQ^d$ and compute the projection $\hat{\vx} = \Pi_{\cP}(\vy)$, using Lemma~\ref{lem:polytope-projection}.
    
    If $\norm{\hat{\vx} - \vy}_2^2 \leq \epsilon^2$, then $\vy \in \cP^{+\epsilon}$ and we halt.

    Otherwise, $\vy \notin \cP^{+\epsilon}$, so we need to return a separating hyperplane. Define $\va := \hat{\vx} - \vy$ and note that $\va^\top \vy < \va^\top \vx$ for all $\vx \in \cP$, so it is a valid separating hyperplane of $\vy$ from $\cP$. Since, Lemma~\ref{lem:colinear-projections} implies that $\Pi_{\cP^{+\epsilon}}(\vy)$ and $\Pi_{\cP}(\vy)$ are co-linear, it follows that $\va$ also is a separating hyperplane of $\vy$ from $\cP^{+\epsilon}$.
\end{proof}

\begin{lemma}\label{lem:approx-optimality}
    Let $\cP \subset \bbR^d$ be a compact convex set that is bounded by balls $\cB_d(r) \subset \cP \subset \cB_d(\vzero, R)$. Let $\mM \in \bbR^{d \times d}$, $\vb \in \bbR^d$ be a matrix and vector respectively with $\norm{\mM}_2,\norm{\vb}_2\leq B$ for some $B > 0$. Define $f(\vx) = \frac{1}{2} \norm{\mM \vx +\vb}_2^2$. For some $\epsilon > 0$, let $\vx_* = \argmin_{\vx \in \cP^{-\epsilon}} f(\vx)$ and let $\hat{\vx} \in \cP^{+\epsilon}$ be a weak solution such that $f(\hat{\vx}) \leq f(\vx) + \epsilon$ for all $\vx \in \cP^{-\epsilon}$. Then,
    \[
        \inp{\nabla f(\hat{\vx}), \vx - \hat{\vx}} \geq -\delta \qquad \forall \vx \in \cP^{-\epsilon},
    \]
    where $\delta := 6 B^2 R^2 \sqrt{\epsilon / r}$. Moreover, $\norm{\mM (\vx_* - \hat{\vx})}_2 \leq 5 B R \sqrt{\epsilon / r}$.
\end{lemma}
\begin{proof}
    \[
        \inp{\nabla f(\hat{\vx}), \vx - \hat{\vx}} &= \inp{\mM^\top(\mM\hat{\vx}+\vb), \vx - \hat{\vx}}\\
        &=\inp{\mM^\top(\mM\vx_*+\vb), \vx - \hat{\vx}} - \inp{\mM^\top\mM(\vx_*-\hat{\vx}), \vx - \hat{\vx}}
    \]
    We have
    \[
        \min_{\vx \in \cP^{-\epsilon}}\inp{\mM^\top(\mM\vx_*+\vb), \vx - \hat{\vx}} &= \inp{\mM^\top(\mM\vx_*+\vb), \vx_* - \hat{\vx}} \geq f(\vx_*) - f(\hat{\vx}) \geq -\epsilon
    \]
    and
    \[
        \inp{\mM^\top\mM(\vx_*-\hat{\vx}), \vx - \hat{\vx}} &\leq \norm{\mM(\vx_*-\hat{\vx})}_2\norm{\mM(\vx - \hat{\vx})}_2\\
        &\leq BR\norm{\mM\vx_*-\mM\hat{\vx}}_2
    \]
    Letting $\bar{\vx}$ be the projection of $\hat{\vx}$ to $\cP^{-\epsilon}$, so that $\norm{\bar{\vx} - \hat{\vx}}_2 \leq 2 \frac{R}{r} \epsilon$ (by Lemma~\ref{lem:eps-project-closeness}). We have $\norm{\mM\vx_*-\mM\hat{\vx}}_2 \leq \norm{\mM\vx_*-\mM\bar{\vx}}_2 +\norm{\mM\bar{\vx}-\mM\hat{\vx}}_2$.  Moreover, $\norm{\mM\bar{\vx}-\mM\hat{\vx}}_2 \leq B\norm{\bar{\vx}-\hat{\vx}}_2 \leq 2B \frac{R}{r} \epsilon$.
    And
    \[
        \norm{\mM\vx_*-\mM\bar{\vx}}_2^2&=\norm{(\mM\vx_*+\vb)-(\mM\bar{\vx}+\vb)}_2^2\\
        &=2\inp{\mM^\top(\mM\vx_*+\vb),\vx_*-\bar{\vx}}+\norm{\mM\bar{\vx}+\vb}_2^2-\norm{\mM\vx_*+\vb}_2^2\\
        & \leq 0+\norm{\mM\bar{\vx}+\vb}_2^2-\norm{\mM\hat{\vx}+\vb}_2^2+\norm{\mM\hat{\vx}+\vb}_2^2-\norm{\mM\vx_*+\vb}_2^2\\
        & \leq \norm{\mM(\bar{\vx}-\hat{\vx})}_2\p{\norm{\mM\bar{\vx}+\vb}_2+\norm{\mM\hat{\vx}+\vb}_2}+2\eps\\
        & \leq 2B \frac{R}{r} \epsilon (4BR)+2\epsilon\\
        &\leq 10 B^2 R^2 \epsilon / r
    \]
    Thus,
    \[
        \norm{\mM (\vx_* - \hat{\vx})}_2 \leq 2 B \frac{R}{r} \epsilon + B R \sqrt{10 \frac{\epsilon}{r}} \leq 5 B R \sqrt{\epsilon / r}
    \]
    \[
    \min_{\vx \in \cP^{-\epsilon}}\inp{\nabla f(\hat{\vx}), \vx - \hat{\vx}} \geq -\epsilon - BR (5 B R \sqrt{\epsilon / r}) \geq -6B^2R^2\sqrt{\epsilon / r}
    \]
    as desired.
\end{proof}

\begin{lemma}\label{lem:min-ball-intersection}
    Let $\cM \subset \bbR^d$ be a compact convex set, of diameter $D = \text{diam}(\cM)$, that contains a ball $\cB_d(\vec a, r) \subset \cM$. Consider the intersection $\cM \cap \cB_d(\vb, l)$ of this set with a ball of radius $l$. If there exists a point $\vp \in \cM \cap \cB_d(\vb, l)$ with $\norm{\vb - \vp}_2 = d > 0$, then $\vol{\cM \cap \cB_d(\vb, l)} > V_d(t)$, where $V_d(t)$ is the volume of the radius-$t$ ball in $d$ dimensions and $t = \frac{r}{2D} (l - d)$.
\end{lemma}
\begin{proof}
    If $\norm{\vb - \vp}_2 = d$, then $\cB_d(\vp, l - d) \subseteq \cB_d(\vb, l)$. We will focus on this radius-$(l - d)$ ball and prove that its intersection with $\cM$ must contain a smaller ball of radius $t = \frac{r}{2D} (l - d)$.

    Consider the following set (a subset of the conic hull of $\cB_d(\vec a, r)$)
    \[
        C = \set{\vp + \gamma (\vx - \vp) : \vx \in \cB_d(\vec a, r), 0 \leq \gamma \leq 1}.
    \]
    Since this is a set of convex combinations of points from $\cM$, it follows that $C \subseteq \cM$. Now consider the ball defined by scaling $\cB_d(\vec a, r)$ by $\gamma_* = \frac{l - d}{2 D} < 1$ as follows
    \[
        \vp + \gamma_* \left( \cB_d(\vec a, r) - \vp \right) \subset C \subseteq \cM.
    \]
    The center of this ball is $\vec{a}' = \vp + \gamma_*(\vec a - \vp)$ and the radius is
    \[
        r' &= \max_{\vx \in \cB_d(\vec a, r)} \norm{\vec{a}' - (\vp + \gamma_* (\vx - \vp))}_2\\
        &= \gamma_* \max_{\vx \in \cB_d(\vec a, r)} \norm{\vec a - \vx}_2\\
        &= \gamma_* r = \frac{r}{2 D} (l - d) = t.
    \]
    We have shown that $\cB_d(\vec{a}', t) \subset C \subseteq \cM$. It suffices to prove that $\cB_d(\vec{a}', t) \subset \cB_d(\vp, l - d)$. To this end, consider any $\vx \in \cB_d(\vec{a}', t)$ and observe that
    \[
        \norm{\vx - \vp}_2 &= \norm{\vx - \vec{a}' + \vec{a}' - \vp}_2\\
        &\leq \norm{\vx - \vec{a}'}_2 + \norm{\vec{a}' - \vp}_2\\
        &\leq t + \gamma_* \norm{\vec a - \vp}_2\\
        &\leq \frac{r}{2 D} (l - d) + \frac{l - d}{2 D} D\\
        &\leq \frac{l - d}{2} + \frac{l - d}{2} = l - d.
    \]
    We conclude that $\vol{\cM \cap \cB_d(\vb, l)} > \vol{\cM \cap \cB_d(\vp, l - d)} > \vol{\cB_d(\vec{a}', t)} > V_d(t)$.
\end{proof}

\section{Infeasibility of Membership over Endomorphic Transformations}\label{app:lb}

In this section, we establish that no algorithm can compute membership to the linear endomorphism set $\Phi(\cP) \subset \bbR^{d \times (d+1)}$ of a compact convex set $\cP \subset \bbR^{d}$ given only polynomially-many membership queries to $\cP$.  More specifically, we take the map $\phi(\vx) = -\frac{4d-1}{4d}\vx$ (a negative scaling of $\vx$)
and show that we cannot check if it is an endomorphism of $\cP$.

Let $\cB_{d} \subseteq \bbR^d$ be the $d$-dimensional ball containing all points $\vx$ with Euclidean norm at most $1$. Let $\Phi_d \subseteq \bbR^{d \times (d+1)}$ be the set of linear endomorphisms of $\cB_d$.  Given a $\kappa \in (0, 1)$ and a unit vector $\vu \in \bbR^d$, define the $d$-dimensional \emph{$\kappa$-capped ball} $\csphere(\kappa, \vu)$ to be the intersection of $\cB_{d}$ with the half-space $\langle\vx, \vu\rangle \leq \kappa$ (note that this corresponds to removing a ``cap'' from $\cB_d$). Similarly, we let $\capPhi_d(\kappa,\vu) \subseteq \bbR^{d \times (d+1)}$ be the set of linear endomorphisms of $\csphere(\kappa,\vu)$. %
We prove the following about the sphere, the capped sphere, and the negative scale map $-\frac{4d-1}{4d} I$.

\begin{lemma}

\begin{enumerate}
    \item $-\frac{4d-1}{4d} I \in \Phi_d$.  Moreover, $\mA-\frac{4d-1}{4d} I \in \Phi_d$ for all $\mA \in \bbR^{d \times d}$ with $\norm{\mA}_2 < \frac{1}{4d}$
    \item $-\frac{4d-1}{4d} I \not \in \capPhi_d\p{\frac{2d-1}{2d},\vu}$ for all $\vu$.  Moreover, $\mA-\frac{4d-1}{4d} I \not \in \capPhi_d\p{\frac{2d-1}{2d},\vu}$ for all $\mA \in \bbR^{d \times d}$ with $\norm{\mA}_2 < \frac{1}{4d}$
    \item We require $\Omega(\exp(d))$ queries to distinguish $\cB_d$ from $\csphere\p{\frac{2d-1}{2d},\vu}$ for all $\vu$.
\end{enumerate}

Thus, it is impossible to solve Weak Membership for $\Phi(\cP)$ with poly$(d)$ membership queries to $\cP$.
\end{lemma}
\begin{proof}
    \textit{1.} A linear map $\phi$ is in $\Phi_d$ iff $\norm{\phi}_2 \leq 1$. Indeed, we have $\norm{-\frac{4d-1}{4d} I}_2 = \frac{4d-1}{4d}$ and $\norm{\mA-\frac{4d-1}{4d} I}_2 \leq \norm{\mA}_2 + \norm{\frac{4d-1}{4d} I}_2 \leq 1$ for all $\norm{\mA}_2 \leq \frac{1}{4d}$.
    
    \textit{2.} A linear map $\phi$ in $\capPhi_d\p{\frac{2d-1}{2d},\vu}$ must satisfy $\phi(-\vu) \in \csphere\p{\frac{2d-1}{2d},\vu}$. Indeed, we have $-\frac{4d-1}{4d} (-\vu) \not \in \csphere\p{\frac{2d-1}{2d},\vu}$.  Moreover, $$\norm{\p{\mA-\frac{4d-1}{4d} I}(-\vu)}_2 \geq  \frac{4d-1}{4d} - \norm{\mA \vu}_2 \geq \frac{2d-1}{2d}$$
    for all $\norm{\mA}_2 \leq \frac{1}{4d}$.

    These two arguments and the fact that $\norm{\mA}_2 \leq \norm{\mA}_F$ prove any algorithm for Weak Membership must accept $-\frac{4d-1}{4d} I$ for $\cP = \cB_d$ and reject it for $\cP = \csphere\p{\frac{2d-1}{2d},\vu}$, as long as $\epsilon \leq 1$.  Lastly, we show that the membership oracles for these $\cP$ sets are indistinguishable under poly$(d)$ calls.  Therefore, we cannot hope to build a Weak Membership algorithm for $\Phi$ that is expected to return two different things given two indistinguishable oracles.
    
    \textit{3.} Consider the set $U = \set{\vu \in \bbR^{d} \middle| \vu[i] = \pm1/\sqrt{d} \text{ for all }i \in [d]}$.  We will show that, for all $\vu_1,\vu_2 \in U$, $\vu_1 \ne \vu_2$, the corresponding caps are disjoint.  That is, if $S_1 = \set{\vx \in \cB_d \middle| \inp{\vx,\vu_1}>\frac{2d-1}{2d}}$ and $S_2 = \set{\vx \in \cB_d \middle| \inp{\vx,\vu_2}>\frac{2d-1}{2d}}$, then $S_1 \cap S_2 = \emptyset$.  Let $i$ be a coordinate upon which $\vu_1,\vu_2$ differ.  Assume WLOG $\vu_1[i]=-1$ and $\vu_2[i]=1$.  We will show $\vx_1[i]<0$ for all $\vx_1 \in S_1$.  This is due to the fact that $\vx^* = \arg\max \set{\inp{\vx,\vu_1} \middle| \vx \in \cB_d, \vx[i] \geq 0}$ satisfies $\vx^*[i] = 0$ and $\vx^*[j] = \frac{\sqrt{d}}{\sqrt{d-1}}\vu_1[j]$ for all $j\ne i$.  Thus, 
    $$\inp{\vx^*,\vu_1} = (d-1)\frac{1}{\sqrt{d(d-1)}} = \frac{\sqrt{d^2-d}}{d} \leq \frac{d-1/2}{d}$$
    Similarly, $\vx_2[i]>0$ for all $\vx_2 \in S_2$.  This proves the desired $S_1 \cap S_2 = \emptyset$.

    Thus, the only way to distinguish $\cB_d$ from $\csphere\p{\frac{2d-1}{2d},\vu}$ for all $\vu \in U$ is to query membership of a point in each of the $2^d$ potential caps.  Therefore, we cannot resolve $\mathrm{WeakMembership}\p{-\frac{4d-1}{4d}I}$ with $\text{poly}(d)$ membership queries to $\cP$.
\end{proof}

\section{Minimizing Linear Swap Regret with Weak Oracles}\label{app:weak}

\subsection{A Weak Semi-Separation Oracle}

We begin with a definition of $\epsilon$-approximate fixed points and an efficient algorithm to find them or to (nearly) certify that one does not exist.

\begin{lemma}[$\epsilon$-approximate fixed-points]\label{lem:eps-approx-fp}
    Let $\cP \subset \bbR^d$ be a compact convex set with $\cB_d(r) \subseteq \cP \subseteq \cB_d(\vzero, R)$
    and consider any affine transformation $\phi := (\mB, \vx_0) \in \bbR^{d \times d} \times \bbR^d$. Then, there exists an oracle-polynomial time algorithm that either returns an $\epsilon$-approximate fixed point $\vx \in \cP^{+\epsilon}$ of $\phi$ such that $\norm{\phi(\vx) - \vx}_2^2 \leq \epsilon$, or certifies that $\norm{\phi(\vx) - \vx}_2^2 > \epsilon / 2$ for all $\vx \in \cP^{-\epsilon}$.
\end{lemma}
\begin{proof}
    Using Theorem~\ref{thm:ellipsoid-convex-sep-to-opt} for $f(\vx) = \norm{\phi(\vx) - \vx}_2^2 = \norm{(\mA - \mI) \vx + \vx_0}_2^2$, we can compute $\vx \in \cP^{+\epsilon/2} \subset \cP^{+\epsilon}$ such that $\norm{\phi(\vx) - \vx}_2^2 \leq \norm{\phi(\vx') - \vx}_2^2 + \epsilon/2$ for all $\vx' \in \cP^{-\epsilon/2} \supset \cP^{-\epsilon}$.

    If $\norm{\phi(\vx) - \vx}_2^2 \leq \epsilon$ then we have found an $\epsilon$-approximate fixed-point.

    Otherwise, $\norm{\phi(\vx') - \vx'}_2^2 > \epsilon - \epsilon/2 > \epsilon / 2$ for all $\vx' \in \cP^{-\epsilon}$.
\end{proof}

\noindent The following, is the weak version of Lemma~\ref{lem:sep-non-fixed-point-strong}.

\begin{lemma}[Weak semi-separation oracle]\label{lem:sep-non-fixed-point}
    Assume weak oracle access to a compact convex set $\cP \subset \bbR^d$ such that $\cB_d(r) \subseteq \cP \subseteq \cB_d(\vzero, R)$.
    Consider any $\phi \in \lintrans{\cB_d(r)}{\cB_d(\vzero, R)}$. Then, we can construct an oracle-polynomial time algorithm that either returns an $\epsilon$-approximate fixed-point of $\phi$ inside $\cP$, or an $\epsilon$-weak separating hyperplane from the relaxed set of affine endomorphisms $\Phi(\cP, \epsilon)$.
\end{lemma}
\begin{proof}
    Let $(\mA, \vx_0) \in \bbR^{d \times d} \times \bbR^d$ be the affine transformation corresponding to $\phi$.
    Throughout this proof, we will be using the spectral norm bound $\norm{\mA}_2 \leq 2R / r$ that follows from Lemma~\ref{lem:lintrans-sphere-spectral}. Additionally, $\norm{\mA - \mI}_2 \leq \norm{\mA}_2 + \norm{\mI}_2 \leq 4R / r$ and define $B := 4R / r$.
    
    First, applying Lemma~\ref{lem:eps-approx-fp} on $\phi$ with precision $\epsilon$, we can either compute an $\epsilon$-approximate fixed point of $\phi$, or certify that $\norm{\phi(\vx) - \vx}_2^2 > \epsilon / 2$ for all $\vx \in \cP^{-\epsilon}$. If we found an $\epsilon$-approximate fixed point of $\phi$ we can stop and return it.
    Otherwise, we are guaranteed that
    \[
        \norm{(\mA - \mI) \vx + \vx_0}_2^2 > \epsilon / 2 \quad \forall \vx \in \cP^{-\epsilon}, \numberthis{eq:norm_lb_one}
    \]
    In the rest of the proof, we focus on this case and describe how to construct the desired weak separating hyperplane.
    
    Equation~\eqref{eq:norm_lb_one} implies that there exists a hyperplane that separates $\vzero$ from the set $(\mA - \mI) \cP^{-\epsilon} + \vx_0$,
    which is an affine transformation of a convex set, and hence, it is convex. Our next step is to compute such a hyperplane by finding an approximate projection of $\vzero$ on that set. To this end, we can apply Theorem~\ref{thm:ellipsoid-convex-sep-to-opt} for $f(\vx) = \frac{1}{2} \norm{(\mA - \mI) \vx + \vx_0}_2^2$ with precision $\beta := \frac{\epsilon^2}{16^2} \frac{r}{(6 B^2 R^2)^2}$  and compute a point $\hat{\vx} \in \cP^{+\beta}$ such that $f(\hat{\vx}) \leq f(\vx) + \beta$ for all $\vx \in \cP^{-\beta} \supset \cP^{-\epsilon}$. Thus, by Lemma~\ref{lem:approx-optimality},
    \[
        &\inp{\nabla f(\hat{\vx}), \vx' - \hat{\vx}} \geq - 6 B^2 R^2 \sqrt{\beta / r} \geq -\frac{\epsilon}{16}\\
        \implies &[(\mA - \mI) \hat{\vx} + \vx_0]^\top (\mA - \mI) (\vx' - \hat{\vx}) \geq -
        \frac{\epsilon}{16} \quad \forall \vx' \in \cP^{-\epsilon}\\
        \implies &\va^\top [(\mA - \mI) \vx' + \vx_0] \geq \va^\top [(\mA - \mI) \hat{\vx} + \vx_0] -\frac{\epsilon}{16} \quad \forall \vx' \in \cP^{-\epsilon}\\
        \implies &\va^\top [(\mA - \mI) \vx' + \vx_0] \geq \norm{\va}_2^2 -\frac{\epsilon}{16} \quad \forall \vx' \in \cP^{-\epsilon},
    \]
    where $\va := (\mA - \mI) \hat{\vx} + \vx_0$. Additionally, the Lemma gives us that for $\bar{\vx} := \argmin_{\vx \in \cP^{-\epsilon}} f(\vx)$ it holds $\norm{(\mA - \mI)(\bar{\vx} - \hat{\vx})}_2 \leq 5 B R \sqrt{\beta / r} \leq \frac{\sqrt{\epsilon}}{2 \sqrt{2}}$. Thus,
    \[
        \norm{\va}_2 = \norm{\phi(\hat{\vx}) - \hat{\vx}}_2 \geq \norm{\phi(\bar{\vx}) - \bar{\vx}}_2 - \frac{\sqrt{\epsilon}}{2 \sqrt{2}} \geq \frac{\sqrt{\epsilon}}{2 \sqrt{2}}, \numberthis{eq:sep_norm_lower_bound}
    \]
    where the last inequality follows from \eqref{eq:norm_lb_one}. Taking squares, we have $\norm{\va}_2^2 \geq \epsilon / 8$.

    We conclude that $\va$ defines a strong separating hyperplane of $\vzero$ from $(\mA - \mI) \cP^{-\epsilon} + \vx_0$ as
    \[
        &\va^\top [(\mA - \mI) \vx' + \vx_0] \geq \epsilon / 16 > 0 \quad \forall \vx' \in \cP^{-\epsilon}\\
        \implies &\va^\top \phi(\vx') > \va^\top \vx' \quad \forall \vx' \in \cP^{-\epsilon}. \numberthis{eq:sep_point_from_trans}
    \]
    Next, consider the convex program $\max_{\vx \in \cP} \vec{a}^\top \mA \vx$ and a precision parameter $\gamma := \epsilon^2 / (8 d)$.
    We compute a $\gamma$-weak solution $\vx_* \in \cP^{+\gamma}$ such that
    \[
        \vec{a}^\top \mA \vx_* &\geq \vec{a}^\top \mA \vx - \gamma \quad \forall \vx \in \cP^{-\gamma} \supset \cP^{-\epsilon}.
    \]
    Combining with \eqref{eq:sep_point_from_trans}, we get
    \[
            \va^\top \phi(\vx_*) = \va^\top \mA \vx_* + \va^\top \vx_0 > \vec{a}^\top \vx - \gamma \quad \forall \vx \in \cP^{-\epsilon}.
    \]

    Let $\mat C = \vec{a} \vx_*^\top \in \bbR^{d \times d}$ and consider any $\phi' := (\mA', \vx_0') \in \Phi(\cP, \epsilon)$. By definition, it must hold that $\phi'(\vx_*) = \vx'$ for some $\vx' \in \cP^{-\epsilon}$. Thus,
    \[
        \inp{\mat C, \mA} + \va^\top \vx_0 &= \vec{a}^\top \mA \vx_* + \va^\top \vx_0\\
            &> \vec{a}^\top \vx' - \gamma\\
            &= \vec{a}^\top \phi'(\vx_*) - \gamma\\
            &= \inp{\mat C, \mA'} + \va^\top \vx_0' - \gamma.
    \]
    Finally, to get a valid weak separating hyperplane, it suffices to show that $N := \max(\norm{\mat C}_{\ell_\infty}, \norm{\va}_\infty) > 0$. To this end, note that, by \eqref{eq:sep_norm_lower_bound}, it holds $\norm{\vec a}_2 \geq \epsilon / 8 \implies N > \norm{\va}_2 / \sqrt{d} \geq \epsilon / (8 \sqrt{d})$. We conclude that $(\mat{C} / N, \va / N) \in \bbR^{d \times (d+1)}$ is a valid $\epsilon$ weak separating hyperplane because $\gamma / N \leq \epsilon$ and
    \[
        \frac{1}{N} \inp{\mat{C}, \mA} + \frac{1}{N} \va^\top \vx_0 > \frac{1}{N} \inp{\mat{C}, \mA'} + \frac{1}{N} \va^\top \vx_0' - \epsilon,
    \]
    for all $(\mA, \vx_0') \in \Phi(\cP, \epsilon)$.
\end{proof}

\subsection{Shell Ellipsoid with Weak Oracles}

\begin{algorithm}[ht]
    \caption{$\shellelpsd(\cF)$ either finds $\phi \in \cF^{+\epsilon}$ with an $\epsilon$-approximate fixed point in $\cP$, or returns a weak separating frontier from set $\Phi(\cP, \epsilon)$}\label{algo:shellelpsd}
    \KwData{Weak separation oracle for compact convex strategy set $\cP \subset \bbR^d$ with $\cB_d(r) \subseteq \cP \subseteq \cB_d(\vzero, R)$}
    \KwInput{Compact convex set $\cF \subset \bbR^{d \times (d+1)}$ of affine transformations, precision parameter $\epsilon$}
    \KwOutput{Either a transformation $\phi \in \cF^{+\epsilon}$ with an $\epsilon$-approximate fixed point inside $\cP$, or a compact convex set $\cQ \supseteq \Phi(\cP, \epsilon)^{-\epsilon}$ with $\vol{ \cQ \cap \cF^{-\epsilon} } < \epsilon$}

    \DontPrintSemicolon

    Initialize $\cQ = \bbR^{d \times (d+1)}$ and ellipsoid $\mathcal{E} \supset \cF$

    Initialize precision parameter $0 < \delta < \epsilon$ as in the ellipsoid method (Theorem~\ref{thm:ellipsoid-sep-to-opt})

    \Repeat{$\vol{\mathcal{E}} < V_{d \times (d+1)}(\epsilon)$}{
        Compute a centroid $\phi$ of $\mathcal{E}$

        \uIf{$\phi \notin \cF^{+\delta}$}{
            Generate a hyperplane that $\delta$-weakly separates $\phi$ from $\cF$
        }
        \uElse{
            Execute the semi-separation oracle from Lemma~\ref{lem:sep-non-fixed-point}
            
            \uIf{$\phi$ has a $\delta$-approximate fixed point inside $\cP$}{
                \Return $\phi$
            }
            \uElse{
                Get a hyperplane that $\delta$-weakly separates $\phi$ from $\Phi(\cP, \delta)$
    
                Add the halfspace containing $\Phi(\cP, \delta)^{-\delta}$ to $\cQ$
            }
        }
    
        Update the ellipsoid $\mathcal{E}$
    }

    \Return $\cQ$
\end{algorithm}

This is the weak version of Lemma~\ref{lem:shellelpsd-strong}.

\begin{lemma}[Shell Ellipsoid with Weak Oracles]\label{lem:shellelpsd}
    Let $\cP \subset \bbR^d$ be a compact convex set to which we have access via a weak separation oracle and which is well-bounded with $\cB_d(r) \subseteq \cP \subseteq \cB_d(\vzero, R)$. For any compact convex set $\cM \subset \bbR^{d \times (d+1)}$ and precision parameter $\epsilon > 0$, Algorithm~\ref{algo:shellelpsd} runs in oracle-polynomial time and guarantees that
    \begin{itemize}
        \item if $\vol{\cF \cap \Phi(\cP, \epsilon)} \geq V_{d \times (d+1)}(\epsilon)$ then the algorithm returns an affine transformation $\phi \in \cF^{+\epsilon}$ with an $\epsilon$-approximate fixed point inside $\cP$,

        \item if it fails to find such an affine transformation, it returns a separating frontier $\cQ$ that is compact and convex and satisfies $\Phi(\cP, \epsilon)^{-\epsilon} \subseteq \cQ$ and $\vol{\cQ \cap \cF^{-\epsilon}} < V_{d \times (d+1)}(\epsilon)$.
    \end{itemize}
\end{lemma}
\begin{proof}
    First note that the ellipsoid method from Theorem~\ref{thm:ellipsoid-sep-to-opt} only needs to fix a parameter $\delta > 0$ in the beginning based on $\epsilon$ and the rest of the input parameters.
    Using Lemma~\ref{lem:sep-non-fixed-point} we can see that the separation oracle for the ellipsoid algorithm in Algorithm~\ref{algo:shellelpsd} requires oracle-polynomial time. Thus, the whole algorithm runs in oracle-polynomial time. We move on to proving the correctness of the algorithm.
    
    For the purposes of this proof, a False Positive affine transformation $\phi$ is one that has an approximate fixed point inside $\cP$, but is not contained in $\Phi(\cP, \epsilon)$. Note also that all transformations $\phi' \in \Phi(\cP, \epsilon)$ must have an $\epsilon$-approximate fixed point that can be found by Lemma~\ref{lem:eps-approx-fp}. This holds because $\vx \mapsto \phi'(\vx)$ can be seen as a continuous map of $\cP^{-\epsilon}$ to $\cP^{-\epsilon}$ and thus, by Brouwer's fixed-point theorem, it must have a fixed point $\vx \in \cP^{-\epsilon}$.

    Let us consider the case $\vol{\cF \cap \Phi(\cP, \epsilon)} \geq V_{d \times (d+1)}(\epsilon)$. First, assume that our separation oracle never gives a False Positive response during the execution of the algorithm. Then the whole procedure is equivalent to solving the Nonemptiness problem on the convex set $\cF \cap \Phi(\cP, \epsilon)$ using a separation oracle, which can be done in oracle-polynomial time \citep[Theorem 4.2.2]{Grotschel1993:Geometric}. Since $\vol{\cF \cap \Phi(\cP, \epsilon)} \geq V_{d \times (d+1)}(\epsilon)$, the algorithm succeeds in finding a valid transformation $\phi \in (\cF \cap \Phi(\cP, \epsilon))^{+\epsilon} \subset \cF^{+\epsilon}$ with an $\epsilon$-approximate fixed point inside $\cP$. Now assume that one of the centroids during the execution of the algorithm is a False Positive. In this case, the algorithm simply returns this centroid and halts. Thus, we have proven that the first guarantee of the Lemma is always satisfied.

    Now consider the case $\vol{\cF \cap \Phi(\cP, \epsilon)} < V_{d \times (d+1)}(\epsilon)$. In this case, \shellelpsd{} will either find a False Positive transformation $\phi \in \cF^{+\epsilon}$ with an $\epsilon$-approximate fixed point inside $\cP$, or it will produce a collection of halfspaces inside $\cQ$ (with $\Phi(\cP, \epsilon)^{-\epsilon} \subset \cQ$, as each halfspace contains $\Phi(\cP, \epsilon)^{-\epsilon}$) that, together with the halfspaces $\cal{W}$, produced by the separation oracle for $\cF$, satisfy $\vol{\cQ \cap \cal{W} \cap \cF} < V_{d \times (d+1)}(\epsilon) \implies \vol{\cQ \cap \cF^{-\epsilon}} < V_{d \times (d+1)}(\epsilon)$. This concludes the proof that $\cQ$ is a valid separating frontier, as is desired for the Lemma.
\end{proof}

\subsection{Weak variants of Shell Gradient Descent and Shell Projection} %

\begin{algorithm}[ht]
    \caption{Shell Gradient Descent on the set $\cX$ (with weak oracles)}\label{algo:improper_ogd}
    \KwData{Compact convex set $\cM_{\cX} \supset \cX$, %
    step sizes $\eta_t$, and precision parameters $\epsilon_t$}

    \For{$t = 1, 2, \dots T$}{
        Output $\vx_t \in \cX_t$
        
        and receive feedback $\lt \in [-1, 1]^d$

        Let $\cX_{t+1} \subseteq \cM_{\cX}$ be any convex shell set of $\cX$

        $\vx_{t+1} = \Pi_{\cX_{t+1}}(\vx_t - \eta_t \lt)$ with precision $\epsilon_t$
    }
\end{algorithm}

Below, we give the weak version of Theorem~\ref{thm:improper_ogd-strong}.
    
\begin{theorem}\label{thm:improper_ogd}
    Let $\cM_{\cX} \subset \bbR^d$ be a compact convex set with diameter $D$ that contains the convex strategy set $\cX$. Let $\vell_1, \dots, \vell_T \in [-1, 1]^d$ and consider any arbitrary sequence of shell sets $\cX_1, \dots, \cX_T$ such that $\cX \subseteq \cX_t \subseteq \cM_{\cX}$ for all $t \in [T]$. Pick any $\vx_1 \in \cX_1$ and assume $\eta_{t+1} \leq \eta_t$. Then, Algorithm~\ref{algo:improper_ogd} has regret
    \[
        \max_{\vx^* \in \cX} \sum_{t=1}^T \inp{\lt, \vx_t - \vx^*} \leq \frac{D^2}{2 \eta_T} + \sum_{t=1}^T \frac{\eta_t}{2} \norm{\lt}_2^2 + D \sum_{t=1}^T \frac{\epsilon_t}{\eta_t}
    \]
\end{theorem}
\begin{proof}
\begin{align*}
    \norm{\Pi_{\cX_{t+1}}(\vx_t-\eta_t \lt)-\vx^*}^2&\leq \norm{\vx_t-\eta_t \lt-\vx^*}^2\\
    &=\norm{\vx_t-\vx^*}^2+\eta_t^2\norm{\lt}^2-2\eta_t  \inp{\lt,\vx_t-\vx^*}
\end{align*}
where the inequality follows since $\vx^* \in \cX$ and therefore $\vx^* \in \cX_{t}$ for all $t$. Also,
\[
\norm{\Pi_{\cX_{t+1}}(\vx_t-\eta_t \lt)-\vx^*} &\geq  \norm{\vx_{t+1}-\vx^*} - \norm{\vx_{t+1}-\Pi_{\cX_{t+1}}(\vx_t-\eta_t \lt)}\\
&\geq \norm{\vx_{t+1}-\vx^*} - \eps_t
\]
and
\[
\norm{\Pi_{\cX_{t+1}}(\vx_t-\eta_t \lt)-\vx^*}^2 &\geq \norm{\vx_{t+1}-\vx^*}^2 - 2\eps_t \norm{\vx_{t+1}-\vx^*} + \eps_t^2 \geq \norm{\vx_{t+1}-\vx^*}^2 - 2D\eps_t
\]
Thus,
\begin{align*}
    \inp{\lt, \vx_t-\vx^*} &\leq \frac{1}{2\eta_t}\p{\norm{\vx_t-\vx^*}^2-\norm{\vx_{t+1}-\vx^*}^2+2D\eps_t}+\frac{\eta_t}{2} \norm{\lt}^2
\end{align*}
Summing over $t$
\begin{align*}
    \sum_{t=1}^T (\lt(\vx_t)-\lt(\vx^*))
    &\leq \sum_{t=1}^T \norm{\vx_t-\vx^*}^2 \p{\frac{1}{2\eta_t}-\frac{1}{2\eta_{t-1}}}+ \sum_{t=1}^T \frac{\eta_t}{2}\norm{\lt}^2+ D \sum_{t=1}^T \frac{\epsilon_t}{\eta_t}\\
    &\leq \frac{D^2}{2\eta_T}+ \sum_{t=1}^T \frac{\eta_t}{2}\norm{\lt}^2+ D \sum_{t=1}^T \frac{\epsilon_t}{\eta_t}
\end{align*}
under the convention $1/\eta_0 = 0$, since $\norm{\vx_t-\vx^*} \leq D$.
\end{proof}

\begin{algorithm}[ht]
    \caption{$\shellproj_\Phi(\phi_a)$ --- projects $\phi_a$ to a shell of $\Phi$}\label{algo:noisyproj}
    \KwData{Weak separation oracle for a compact convex strategy set $\cP \subset \bbR^d$, bounding balls $\cB_d(r) \subseteq \cP \subseteq \cB_d(\vzero, R)$, compact convex set $\cM \subset \bbR^{d \times (d+1)}$ that is guaranteed to contain the endomorphic transformations $\Phi(\cP) \subset \cM$}
    \KwInput{Affine transformation $\phi_a \in \bbR^{d \times (d+1)}$, precision parameters $\eta$ and $\epsilon \in (0, 1)$}
    \KwOutput{Compact convex set $\Tilde{\Phi} \supset \Phi(\cP, \eta)^{-\epsilon}$, affine transformation $\phi_b \in \Tilde{\Phi}^{+\epsilon}$ that is an $\epsilon$-approximate projection of $\phi_a$ onto $\Tilde{\Phi}$ and is guaranteed to have a $\min(\epsilon, \eta)$-approximate fixed point inside $\cP$}
    
    \DontPrintSemicolon

    Initialize $\Tilde{\Phi} := \cM$

    $H := 2 (\text{diam}(\cM) + \text{dist}(\phi_a, \cM))$

    Let $\epsilon' := \min(\eta, \poly(\epsilon, r, \frac{1}{H}, \frac{1}{R}))$

    \For{$q \gets 0$, increment by $\eps/(2H)$}{

        Run $\shellelpsd(\Tilde{\Phi} \cap \cB_{d \times (d+1)}(\phi_a, q))$ with precision $\epsilon'$
        
        \eIf{it finds $\phi_b$ with an $\epsilon'$-approximate fixed point inside $\cP$}{
            \Return $\Tilde{\Phi}$ and $\phi_b$
        }{
            
            Update $\Tilde{\Phi}$ with the separating frontier found by \shellelpsd{}

        }
    }

\end{algorithm}

Next, we give the weak version of Theorem~\ref{thm:shellyproj-strong}.

\begin{proposition}[Weak Shell Projection]\label{prop:shellyproj}
    Let $\cP \subset \bbR^d$ be a compact convex strategy set that is bounded between balls $\cB_d(r) \subseteq \cP \subseteq \cB_d(\vzero, R)$.
    Let $\cM$ be a compact convex set such that $\Phi(\cP, \eta) \subset \cM$. For any $\phi_a \in \bbR^{d \times (d+1)}$, Algorithm~\ref{algo:noisyproj} runs in oracle-polynomial time and returns a set $\Tilde{\Phi} \supset \Phi(\cP, \eta)^{-\epsilon}$ and an affine transformation $\phi_b = \shellproj_\Phi(\phi_a)$ that admits a $\min(\epsilon, \eta)$-approximate fixed point $\vx \in \cP^{+\eta}$ with $\norm{\phi_b(\vx) - \vx}_2^2 \leq \epsilon$ and
    \[
        \norm{\phi_b - \Pi_{\Tilde{\Phi}}(\phi_a)}_F^2 \leq \epsilon,
    \]
    where $\Pi_{\Tilde{\Phi}}(\phi_a)$ is the projection of $\phi_a$ onto $\Tilde{\Phi}$.
\end{proposition}
\begin{proof}
    Without loss of generality, assume that $\eta < r/4R$, so that $\cB_{d \times (d+1)}(\frac{r}{4R}) \subseteq \Phi(\cP, \eta)$ (Lemma~\ref{lem:eps-inscribed-lintrans}). Also note that for any $\epsilon' < \eta$, Lemma~\ref{lem:subset-lintrans} implies that $\Phi(\cP, \eta') \supseteq \Phi(\cP, \eta)$.
    In this algorithm, we use $\epsilon' := \min(\eta, \frac{\epsilon}{8H} \frac{r}{8R} \frac{1}{H} , \frac{\epsilon}{16 H} \frac{r^2}{12 R^2} \frac{1}{R^2 + d} )$ where $H := 2 (\text{diam}(\cM) + \text{dist}(\phi_a, \cM))$ and initialize parameter $\epsilon' \in (0, \epsilon)$.
    We begin by proving that the algorithm is guaranteed to find at least one transformation $\phi_b$ with an approximate fixed point inside $\cP$. Indeed, it finds one when $q$ is at least large enough to satisfy %
    $q = H / 2 = \text{diam}(\Tilde{\Phi}) + \text{dist}(\phi_a, \Tilde{\Phi})$.
    This means that $\cB_{d \times (d+1)}(\phi_a, q) \supset \Tilde{\Phi} \supset \Phi(\cP, \epsilon')$.
    Thus $\vol{\Tilde{\Phi} \cap \Phi(\cP, \epsilon')} = \vol{\Phi(\cP, \epsilon')} > V_{d \times (d+1)}(r / 4R) > V_{d \times (d+1)}(\epsilon')$ and the guarantee immediately follows from Lemma~\ref{lem:shellelpsd}.

    Note that every time we increment $q$, we update the shell set $\Tilde{\Phi}$ to be $\Tilde{\Phi} = \Tilde{\Phi}' \cap \cQ$, where $\Tilde{\Phi}'$ is the shell set in the previous iteration and $\cQ$ is the separating frontier found by \shellelpsd{}.
    Thus, by Lemma~\ref{lem:shellelpsd} we have that $\Phi(\cP, \eta)^{-\epsilon} \subset \Phi(\cP, \epsilon')^{-\epsilon'} \subset \Tilde{\Phi}' \cap \cQ = \Tilde{\Phi}$, because $\Phi(\cP, \epsilon')^{-\epsilon'} \subset \cQ$ and $\Phi(\cP, \epsilon')^{-\epsilon'} \subset \Tilde{\Phi}'$.
    Additionally,
    \[
        \vol{\cB_{d \times (d+1)}(\phi_a, q-\eps/(2H)) \cap \Tilde{\Phi}^{-\epsilon'}} &= \vol{\cB_{d \times (d+1)}(\phi_a, q-\eps/(2H)) \cap (\Tilde{\Phi}')^{-\epsilon'} \cap \cQ^{-\epsilon'}}\\
            &< V_{d \times (d+1)}(\epsilon').
    \]
    Next, we will prove that for any $\phi \in \cB_{d \times (d+1)}(\phi_a, q-\eps/(2H)) \cap \Tilde{\Phi}^{-\epsilon'}$, it is $\norm{\phi_a - \phi}_F > q-\frac{\eps}{2H} - \frac{\epsilon}{8 H}$. Assume otherwise; then there would exist $\phi \in \cB_{d \times (d+1)}(\phi_a, q-\eps/(2H)) \cap \Tilde{\Phi}^{-\epsilon'}$ with $s := \norm{\phi_a - \phi}_F \leq q-\frac{\eps}{2H} - \frac{\epsilon}{8 H}$.
    Applying Lemma~\ref{lem:min-ball-intersection} for $\cM := \Tilde{\Phi}^{-\epsilon'} \subset \bbR^{d \times d} \times \bbR^d$, and $\vb := \phi_a$, $\vp := \phi$, it follows that $\cM \supset \Phi(\cP, \eta)^{-2\epsilon'} \supseteq \cB_{d \times (d+1)}(\frac{r}{4R} - 2\epsilon') \supseteq \cB_{d \times (d+1)}(\frac{r}{8R})$.
    Consequently, it has to hold that $\vol{\cB_{d \times (d+1)}(\phi_a, q-\eps/(2H)) \cap \Tilde{\Phi}^{-\epsilon'}} > V_{d \times (d+1)}(t)$ for $t = \frac{r}{8R} \frac{1}{2 \text{diam}(\Tilde{\Phi})} (q-\eps/(2H) - s) > \frac{r}{8R} \frac{1}{H} \frac{\epsilon}{8 H} \geq \epsilon'$ which is a contradiction.
    
    There always exists a $\phi_b \in \Tilde{\Phi}^{+\epsilon'} \cap \cB_{d \times (d+1)}(\phi_a, q) $ with an $\epsilon'$-approximate fixed point in $\cP$. We see that the set $\Tilde{\Phi} \cap \cB_{d \times (d+1)}(\phi_a, H)$ is guaranteed to contain such a transformation. %
    
    We have proven that the following three invariant properties are true in every iteration of the search:
    \begin{enumerate}
        \item\label{prprt:binsearch-lb} (Almost-empty intersection) For any $\phi \in \cB_{d \times (d+1)}(\phi_a, q-\eps/(2H)) \cap \Tilde{\Phi}^{-\epsilon'}$ it is $\norm{\phi_a - \phi}_F > q-\frac{\eps}{2H} - \frac{\eps}{8 H}$

        \item $\Tilde{\Phi}$ is a compact convex set such that $\Phi(\cP, \eta)^{-\epsilon} \subset \Tilde{\Phi}$

        \item (Consistency) There exists $\phi_b \in \Tilde{\Phi}^{+\epsilon'} \cap \cB_{d \times (d+1)}(\phi_a, q)$ with an $\epsilon'$-approximate fixed point inside $\cP$
    \end{enumerate}
    When the algorithm terminates, $\phi_b \in \Tilde{\Phi}^{+\epsilon'} \cap \cB_{d \times (d+1)}(\phi_a, q)$ is a valid transformation with an $\epsilon'$-approximate fixed point inside $\cP$. Since $\phi_b \in \cB_{d \times (d+1)}(\phi_a, q)$, it follows that
    \[
        \norm{\phi_a - \phi_b}_F^2 \leq q \numberthis{eq:binsearch-final-ub}
    \]
    and, additionally, Property~\ref{prprt:binsearch-lb} implies that
    \[
        \norm{\phi_a - \Pi_{\Tilde{\Phi}^{-\epsilon'}}(\phi_a)}_F > q-\frac{\eps}{2H} - \frac{\epsilon}{8 H}. \numberthis{eq:binsearch-final-lb}
    \]

    Using Lemma~\ref{lem:eps-project-closeness}, we get that for $\phi_b' = \Pi_{\Tilde{\Phi}^{-\epsilon'}}(\phi_b)$ it holds $\norm{\phi_b - \phi_b'}_F \leq \frac{R_{\Tilde{\Phi}}}{r_{\Tilde{\Phi}}} \epsilon'$. And using Lemmas~\ref{lem:eps-inscribed-lintrans} and \ref{lem:lintrans-sphere-spectral}, it follows that
    \[
        \norm{\phi_b - \phi_b'}_F \leq \frac{3R}{r} \sqrt{R^2 + d} \frac{4R}{r} \epsilon' \leq \frac{\epsilon}{16H}.
    \]
    Thus, applying \eqref{eq:binsearch-final-ub},
    \[
        \norm{\phi_a - \phi_b'}_F &\leq \norm{\phi_a - \phi_b}_F + \norm{\phi_b - \phi_b'}_F\\
            &\leq q + \frac{\epsilon}{16H}\numberthis{eq:binsearch-final-ub-proj}%
    \]

    The projection of $\phi_a$ onto $\Tilde{\Phi}^{-\epsilon'}$ satisfies $\Pi_{\Tilde{\Phi}^{-\epsilon'}}(\phi_a) = \argmin_{\mat{F} \in \Tilde{\Phi}^{-\epsilon'}} \norm{\phi_a - \mat{F}}_F$. Applying the convex optimality conditions as in \citet[Proposition 2.11]{orabona2022:modern}, we get
    \[
        \inp{\phi_a - \Pi_{\Tilde{\Phi}^{-\epsilon'}}(\phi_a), \Pi_{\Tilde{\Phi}^{-\epsilon'}}(\phi_a) - \phi_b'} \geq 0
    \]
    With this tool, we now have
    \[
        \norm{\phi_a - \phi_b'}_F^2 &= \norm{\phi_a - \Pi_{\Tilde{\Phi}^{-\epsilon'}}(\phi_a) + \Pi_{\Tilde{\Phi}^{-\epsilon'}}(\phi_a) - \phi_b'}_F^2\\
            &= \norm{\phi_a - \Pi_{\Tilde{\Phi}^{-\epsilon'}}(\phi_a)}_F^2 + \norm{\Pi_{\Tilde{\Phi}^{-\epsilon'}}(\phi_a) - \phi_b'}_F^2 + 2 \inp{\phi_a - \Pi_{\Tilde{\Phi}^{-\epsilon'}}(\phi_a), \Pi_{\Tilde{\Phi}^{-\epsilon'}}(\phi_a) - \phi_b'}\\
            &\geq \norm{\phi_a - \Pi_{\Tilde{\Phi}^{-\epsilon'}}(\phi_a)}_F^2 + \norm{\Pi_{\Tilde{\Phi}^{-\epsilon'}}(\phi_a) - \phi_b'}_F^2.
    \]
    Rearranging and using inequalities \eqref{eq:binsearch-final-ub-proj} and \eqref{eq:binsearch-final-lb}, we conclude
    \[
        \norm{\Pi_{\Tilde{\Phi}^{-\epsilon'}}(\phi_a) - \phi_b'}_F^2 &\leq \norm{\phi_a - \phi_b'}_F^2 - \norm{\phi_a - \Pi_{\Tilde{\Phi}^{-\epsilon'}}(\phi_a)}_F^2\\
            &< \p{q + \frac{\epsilon}{16H}}^2 - \left( q-\frac{\eps}{2H} - \frac{\epsilon}{8 H} \right)^2\\
            & \leq \frac{2q\eps}{H}\\
            &\leq \epsilon.
    \]
\end{proof}

\subsection{Minimizing Linear Swap Regret with Weak Oracles}

\begin{algorithm}[ht]
    \caption{Linear-swap regret minimizer for convex strategy sets with weak oracles}\label{algo:main}
    \KwData{Weak separation oracle for a compact convex strategy set $\cP \subset \bbR^d$ in Isotropic Position (Definition~\ref{defn:isotropic}), precision parameter $\eta > 0$}

    Let $R_\phi = 4d^2$ and $r_\phi = 1 / 4 (d+1)$
    
    Let $\cM := \cB_{d \times (d+1)}(\vzero, R_\phi)$ which guarantees that $\Phi(\cP, \eta) \subset \cM$

    Set step size $\beta := \frac{R_\phi}{d^{3/2} \sqrt{T}}$ and parameter 
    $\epsilon := \frac{1}{d T^2} \frac{r_\phi}{R_\phi}$

    Let $\vp_1 \in \cP^{+\eta}$ and $\phi_1$ be any valid transformation in $\Phi(\cP, \eta)$ (\eg the Approximate Identity)

    \For{$t = 1, 2, \dots, T$}{
        Output $\pt \in \cP^{+\eta}$ and receive feedback $\lt \in [-1, 1]^d$

        Set $\mat{L}_t = \left( \lt \pt^\top, \lt \right) \in \bbR^{d \times (d+1)}$

        Update $\phi_{t+1} = \shellproj_\Phi(\phi_t - \beta \mat{L}_t)$ with precisions $\eta$ and $\epsilon$

        Compute $\vp_{t+1} \in \cP^{+\eta}$ such that $\norm{\vp_{t+1} - \phi_{t+1}(\vp_{t+1})}_2^2 \leq \epsilon$
    }
\end{algorithm}

This is the weak version of Theorem~\ref{algo:main-strong}.

\begin{theorem}\label{thm:linswap-regret-main-algo}
    Let $\cP \subset \bbR^d$ be a compact convex strategy set in the Isotropic Position (Definition~\ref{defn:isotropic}).
    Let the loss vectors be $\lt \in [-1, 1]^d$ and $\eta > 0$ be a precision parameter. Then
    Algorithm~\ref{algo:main} outputs strategy points $\pt \in \cP^{+\eta}$ and guarantees
    \[
        \mathrm{LinSwapReg} = \sum_{t=1}^T \inp{\pt, \lt} - \min_{\phi \in \Phi(\cP, \eta)} \sum_{t=1}^T \inp{\phi(\pt), \lt} = O\left( d^4 \sqrt{T} \right)
    \]
    and its per-iteration time complexity is oracle-polynomial in $d$ and $T$. %
\end{theorem}
\begin{proof}
    First observe that, Lemma~\ref{lem:lintrans-sphere-spectral} and the Isotropic Position of $\cP$ imply $\Phi(\cP, \eta) \subseteq \lintrans{\cB_d(1)}{\cB_d(\vzero, d+1)} \subseteq \cM$.
    Also, it holds that $\cB_{d \times (d+1)}(r_\phi) \subset \Phi(\cP, \eta) \subset \cB_{d \times (d+1)}(\vzero, R_\phi)$ where $R_\phi = 4d^2$ (Lemma~\ref{lem:lintrans-sphere-spectral}) and $r_\phi = 1 / 4 (d+1)$ (Lemma~\ref{lem:eps-inscribed-lintrans}).
    
    Let $D = \text{diam}(\cM)$ and $G$ be an upper bound to all loss gradients, $\norm{\mat{L}_t}_F \leq G$. 
    For the diameter, it is $D = 2 R_\phi \leq 8 d^2$.

    To bound $G$ notice that $\norm{\mat{L}_t}_F \leq \sqrt{d \norm{\lt}_\infty^2 \norm{\pt}_2^2 + \norm{\lt}_2^2} \leq 2 d^{3/2}$, because $\lt \in [-1, 1]^d$ and $\pt \in \cP^{+\eta} \subset \cB_d(\vzero, d + 1 + \eta)$. Thus, $G \leq 2 d^{3/2}$.

    Following \citet{Gordon08:No}, the linear-swap regret of Algorithm~\ref{algo:main} is equal to
    \[
        \mathrm{LinSwapReg} &= \sum_{t=1}^T \inp{\pt, \lt} - \min_{\phi \in \Phi(\cP, \eta)} \sum_{t=1}^T \inp{\phi(\pt), \lt}\\
            &= \sum_{t=1}^T \inp{\phi_t(\pt), \lt} - \min_{\phi \in \Phi(\cP, \eta)} \sum_{t=1}^T \inp{\phi, \mat{L}_t} + \sum_{t=1}^T \inp{\pt - \phi_t(\pt), \lt}\\
            &\leq \sum_{t=1}^T \inp{\phi_t, \mat{L}_t} - \min_{\phi \in \Phi(\cP, \eta)} \sum_{t=1}^T \inp{\phi, \mat{L}_t} + \sum_{t=1}^T \sqrt{\epsilon d}
    \]
    However, in the following, we will only be able to work on the restricted set $\Phi(\cP, \eta)^{-\epsilon}$. 
    To overcome this issue note that, by Lemma~\ref{lem:eps-project-closeness}, we get that for any $\phi \in \Phi(\cP, \eta)$ there exists $\phi' \in \Phi(\cP, \eta)^{-\epsilon}$ with $\norm{\phi - \phi'}_F \leq \frac{R_\phi}{r_\phi} \epsilon$. 
    Thus, for $\epsilon = \frac{1}{d T^2} \frac{r_\phi}{R_\phi}$,
    \[
        \left| \sum_{t=1}^T \inp{\phi - \phi', \mat{L}_t} \right| \leq \norm{\phi - \phi'}_F \sum_{t=1}^T \norm{\mat{L}_t}_F \leq 2 T d^{3/2} \frac{R_\phi}{r_\phi} \epsilon = O(\sqrt{d})
    \]
    And the regret is bounded as
    \[
        \mathrm{LinSwapReg} &\leq \sum_{t=1}^T \inp{\phi_t, \mat{L}_t} - \min_{\phi \in \Phi(\cP, \eta)^{-\epsilon}} \sum_{t=1}^T \inp{\phi, \mat{L}_t} + \sum_{t=1}^T \sqrt{\epsilon d} + O(1)\\
        &\leq \sum_{t=1}^T \inp{\phi_t, \mat{L}_t} - \min_{\phi \in \Phi(\cP, \eta)^{-\epsilon}} \sum_{t=1}^T \inp{\phi, \mat{L}_t} + O(\sqrt{T d})
    \]
    
    This is an instance of Shell Gradient Descent, since we know from Proposition~\ref{prop:shellyproj} that \shellproj{} finds a set $\Tilde{\Phi}$ with $\Phi(\cP, \eta)^{-\epsilon} \subset \Tilde{\Phi} \subset \cM$ and a transformation $\phi_{t+1} \in \bbR^{d \times (d+1)}$ such that $\norm{\phi_{t+1} - \Pi_{\Tilde{\Phi}}(\phi_t - \beta \mat{L}_t)}_F^2 \leq \epsilon$. Applying Theorem~\ref{thm:improper_ogd} for step sizes equal to $\beta$ and precision parameters equal to $\epsilon$,we get that the per-iteration complexity is oracle-polynomial and the regret is bounded by
    \[
        \sum_{t=1}^T \inp{\phi_t, \mat{L}_t} - \min_{\phi \in \Phi(\cP, \eta)^{-\epsilon}} \sum_{t=1}^T \inp{\phi, \mat{L}_t} \leq O(D G \sqrt{T}),
    \]
    Finally, combining this with the two bounds for $G$ and $D$, we conclude that $\text{LinSwapReg} \leq O(D G \sqrt{T}) = O\left( d^4 \sqrt{T} \right)$.
\end{proof}

\section{Computing a Linear Correlated Equilibrium with Weak Oracles}\label{app:eah-weak}

In this section, we will describe how to compute a high-precision linear correlated equilibrium in polynomial time using weak oracle access to the strategy set. These are the results of Section~\ref{sec:eah} but for weak oracles.

\subsection{The Ellipsoid Against Hope framework with Weak Oracles}

First, we clearly define the weak version of the good-enough-response problem, as follows.

\begin{mdframed}
    \textbf{The Weak Good-Enough-Response Problem}

    Given a vector $\vy \in \bbR^N$ and an $\epsilon > 0$,
    \vspace{-3mm}
    
    \hspace*{5mm} find $(\vec{x}, \vec{x}^\top \mat{A}) \in \mathcal{X} \times \bbR^N$ such that $\vec{x}^\top \mat{A} \vec{y} > -\epsilon$.
\end{mdframed}

\begin{algorithm}[ht]
    \SetAlgoNoLine
    \caption{Ellipsoid Against Hope for bilinear zero-sum games with general convex strategies and weak oracle access}\label{alg:eah}
    \KwData{Parameters $r_y, R_y$ such that $\cB_N(r_y) \subseteq \cY \subseteq \cB_N(\vzero, R_y)$, precision parameter $\epsilon > 0$, constant $B \geq 1$ such that $\norm{\vx^\top \mA}_2 \leq B$}
	\KwIn{An oracle that for every $\vy \in \cB_N(\vzero, R_y)$, either produces a weak separating hyperplane (\texttt{SEP}) of $\vy$ from $\cY$ or a weak good-enough-response (\texttt{GER})}
	\KwOut{A sparse solution $\vec{x}^*$ of \eqref{eq:primal_cp} represented as a mixture of \texttt{GER} oracle responses.}

    Set $\xi := \epsilon/3$ and $\delta := \frac{r_y}{4 R_y} \frac{\xi}{8B}$

    Set $\cM := \cB_N(\vzero, R_y)$

    Execute the ellipsoid method on \eqref{eq:dual_cp}, using the given oracle
    \[
        \tag{$D$}\label{eq:dual_cp}
        \text{find}\ &\vec y \in \cM\\
        \text{s.t.}\ &\max_{\vec x \in \cX} \vec{x}^\top \mat{A} \vec{y} \leq -\frac{\xi}{4}
    \]
    and stop when the volume becomes less than $V_N(\delta)$.

    Let $\widetilde{\cY}$ be the intersection of $\cM$ and the \texttt{SEP} oracle responses

    Let $\vx_1, \dots, \vx_L$ be the GER oracle response vectors and define $\mat{X} = [\vx_1 \mid \dots \mid \vx_L]$

    Set $\gamma := \xi / (3 \sqrt{L} B R_y)$
    
    Using the response vectors, create \eqref{eq:primal_compressed_cp} and compute a $(\gamma/2)$-weakly feasible solution $\vec{\lambda}_*$
    \[
        \tag{$P'$}\label{eq:primal_compressed_cp}
        \text{find}\ &\vec{\lambda} \in (\Delta^L)^{+\gamma}\\
        \text{s.t.}\ &\min_{\vec y \in \widetilde{\cY}^{-\delta}} \vec{\lambda}^\top (\mat{X}^\top \mat{A}) \vec{y} \geq -\xi
    \]

    Compute the projection $\hat{\vlam}_* = \Pi_{\Delta^L}(\vlam_*)$ of $\vlam_*$ on $\Delta^L$
    
    Compute final solution $\vec{x}_* = \mat{X} \hat{\vlam}_*$
\end{algorithm}

\begin{restatable}{theorem}{thmeah}\label{thm:ellipsoid-against-hope}
    Let $\cX \subset \bbR^M, \cY \subset \bbR^N$ be compact convex sets, to which we do not necessarily have direct oracle access. And let $\epsilon > 0$ be a precision parameter. If the following hold
    \begin{enumerate}
        \item $\mA \in \bbR^{M \times N}$ is a matrix such that $\forall \vx \in \cX,\ \norm{\vx^\top \mA}_2 \leq B$ for some $B \geq 1$.

        \item the set $\cY$ is well-bounded as $\cB_N(r_y) \subseteq \cY \subseteq \cB_N(\vzero, R_y)$,

        \item\label{assmpt:eah_sep_oracles} there exists an oracle that for every point $\vy \in \cB_N(\vzero, R_y)$, either produces a weak separating hyperplane (\texttt{SEP}) of $\vy$
        from $\cY$ or a weak good-enough-response (\texttt{GER}),

        \item the encoding lengths of both the \texttt{SEP} and the \texttt{GER} responses produced by the previous oracle are polynomially bounded,
    \end{enumerate}
    then, Algorithm~\ref{alg:eah} runs in oracle-polynomial time and computes an \emph{exact} solution $\vx_*$ to
    \[
        \tag{$P$}\label{eq:primal_cp}
        \text{find}\ &\vec{x} \in \mathcal{X}\\
        \text{s.t.}\ &\min_{\vec y \in \mathcal{Y}} \vec{x}^\top \mat{A} \vec{y} \geq -\epsilon.
    \]
    Furthermore, $\vx_*$ is a mixture of polynomially many \texttt{GER} responses.
\end{restatable}

\begin{proof}
    Set $\xi := \epsilon / 3$ and $\delta := \frac{r_y}{4 R_y} \frac{\xi}{8B}$, such that $\cY^{-\delta}$ is non-empty.
    Let $\cK_1$ be the set of all feasible solutions to \eqref{eq:dual_cp} and note that
    \[
        \cK_1^{-\delta} \supseteq \cM^{-\delta} \cap \set{\vy \in \bbR^N \mid \max_{\vx \in \cX} \vx^\top \mA \vy \leq - \frac{\xi}{2}}.
    \]
    To see this, we can pick any $\vy$ in the right-hand side set and show that $\cB_N(\vy, \delta) \subseteq \cK_1$. To this end, consider any $\vy' \in \cB_N(\vy, \delta)$ and note that $\vy' \in \cM$ and $| \vx^\top \mA (\vy - \vy') | \leq \norm{\vx^\top \mA}_2 \norm{\vy - \vy'}_2 \leq B \delta$ for all $\vx \in \cX$. Consequently,
    \[
        \max_{\vec x \in \cX} \vec{x}^\top \mat{A} \vec{y}' \leq \max_{\vec x \in \cX} \vec{x}^\top \mat{A} \vec{y} + B \delta \leq -\frac{\xi}{2} + B \delta \leq -\frac{\xi}{4}
    \]
    so $\vy' \in \cK_1$.

    \paragraph{Almost emptiness of \eqref{eq:dual_cp}.}
    By Assumption~\ref{assmpt:eah_sep_oracles}, it follows that the combination of the semi-separation and the good-enough-response oracles will always produce a valid weak separating hyperplane for \eqref{eq:dual_cp}. Thus, running the ellipsoid method on \eqref{eq:dual_cp}, we can guarantee that $\cK_1^{-\delta}$ is empty.
    Let the GER oracle responses used in this process be $\vx_1, \dots, \vx_L$ and let the intersection of $\cM$ with the semi-separation oracle responses be $\widetilde{\cY}$. Note that $\widetilde{\cY} \supseteq \cY^{-\delta}$ since the semi-separation oracle will not produce any separating hyperplanes for points inside $\cY^{-\delta}$. Furthermore, note that we have direct (weak) oracle access to $\widetilde{\cY}$ by combining the oracle for $\cM$ and the explicitly given separating hyperplanes.

    Now consider the following convex program and note that both the encoding length of the hyperplanes defining $\widetilde{\cY}$ and the encoding length of $\mat{X}^\top \mA$ are polynomially bounded, by assumption.
    \[
        \tag{$D'$}\label{eq:dual_compressed_cp}
        \text{find}\ &\vec y \in \widetilde{\cY}\\
        \text{s.t.}\ &\max_{\vlam \in \Delta^L} \vlam^\top (\mat{X}^\top \mat{A}) \vec{y} \leq -\frac{\xi}{4}
    \]
    If we run the ellipsoid method on \eqref{eq:dual_compressed_cp} and use the same oracle responses as before, it will consider the same sequence of points and thus, it will guarantee that the below set is empty
    \[
        \widetilde{\cY}^{-\delta} \cap \set{\max_{\vlam \in \Delta^L} \vlam^\top (\mat{X}^\top \mat{A}) \vec{y} \leq - \frac{\xi}{2}}.
    \]
    We have proven that
    \[
        \min_{\vy \in \widetilde{\cY}^{-\delta}} \max_{\vlam \in \Delta^L} \vlam^\top (\mat{X}^\top \mat{A}) \vec{y} > -\frac{\xi}{2}. \numberthis{eq:eah_dual_infeasibility}
    \]

    \paragraph{Weak feasibility of \eqref{eq:primal_compressed_cp}.} We denote with $\cK_2$ the set of all feasible solutions of \eqref{eq:primal_compressed_cp}. In the previous step, we established the (almost) infeasibility of \eqref{eq:dual_compressed_cp} in \eqref{eq:eah_dual_infeasibility}. Combining this with the minimax theorem, it follows that \eqref{eq:primal_compressed_cp} is feasible and, furthermore,
    \[
        \max_{\vlam \in \Delta^L} \min_{\vy \in \widetilde{\cY}^{-\delta}} \vlam^\top (\mat{X}^\top \mat{A}) \vec{y} > -\frac{\xi}{2}.
    \]
    Let $\vlam \in \Delta^L$ be any point in the simplex such that $\vlam^\top (\mat{X}^\top \mA) \vy > -\xi / 2$ and define the parameter $\gamma := \xi / (3 \sqrt{L} B R_y)$. For any point $\vlam' \in \cB_L(\vlam, \gamma)$ it holds that $\vlam' \in (\Delta^L)^{+\gamma}$ and
    \[
        | (\vlam' - \vlam)^\top (\mat{X}^\top \mA) \vy | &\leq \norm{\vlam' - \vlam}_2 \norm{(\mat{X}^\top \mA) \vy}_2\\
        &\leq \gamma \sqrt{L} \norm{(\mat{X}^\top \mA) \vy}_\infty\\
        &\leq \gamma \sqrt{L} B \norm{\vy}_2 = \gamma \sqrt{L} B R_y.
    \]
    Thus, $(\vlam')^\top (\mat{X}^\top \mA) \vy > -\xi/2 - \gamma \sqrt{L} B R_y \geq -\xi$
    and the set of solutions $\cK_2$ is guaranteed to contain a ball of radius $\gamma$.

    \paragraph{Weak separation oracle for \eqref{eq:primal_compressed_cp}.}
    Our next goal is to show that we can construct an efficient weak separation oracle for \eqref{eq:primal_compressed_cp}, which can be used to compute a weakly feasible solution $\lambda_* \in \cK_2^{+\gamma/2}$ with the ellipsoid method.
    
    Assume that we want to construct a $\beta$-weak separation oracle for $\cK_2$. Given a $\vlam \in \bbQ^L$, the construction goes as follows.

    First, using the strong separation oracle from Lemma~\ref{lem:inflated-polytope-sep}, check whether $\vlam \in (\Delta^L)^{+\gamma}$ and if it is not, return a strong separating hyperplane of $\vlam$ from $(\Delta^L)^{+\gamma}$.

    Now, define $\delta_1 := \xi \frac{\beta}{\beta + 1} \frac{1}{L (1 + \gamma)}$ and $\delta_2 := \min(\delta, \frac{r_y - \delta}{R_y + \delta} \frac{1}{4 B} \delta_1)$ and consider the following optimization sub-problem
    \[
        \min_{\vy \in \widetilde{\cY}} \vec{\lambda}^\top (\mat{X}^\top \mat{A}) \vec{y}.
    \]
    Using weak optimization, we can compute a solution $\vy_* \in \widetilde{\cY}^{+\delta_2}$ such that $\vec{\lambda}^\top (\mat{X}^\top \mA) \vy_* - \delta_2 \leq \vec{\lambda}^\top (\mat{X}^\top \mA) \vy$ for all $\vy \in \widetilde{\cY}^{-\delta_2} \supset \widetilde{\cY}^{-\delta}$. There are two cases:
    \begin{itemize}
        \item If $\vec{\lambda}^\top (\mat{X}^\top \mA) \vy_* \geq -\xi + \delta_2$ then for all $\vy \in \widetilde{\cY}^{-\delta}$, it is $\vec{\lambda}^\top (\mat{X}^\top \mA) \vy \geq \vec{\lambda}^\top (\mat{X}^\top \mA) \vy_* - \delta_2 \geq -\xi$ so $\vlam$ is a feasible point for \eqref{eq:primal_compressed_cp}.

        \item Otherwise, $\vlam$ satisfies
        \[
            \vec{\lambda}^\top (\mat{X}^\top \mA) \vy_* < -\xi + \delta_2. \numberthis{eq:sep_compr_secondcase}
        \]
        Consider any feasible $\vlam' \in \cK_2$ and note that
        \[
            (\vlam')^\top (\mat{X}^\top \mA) \vec{y}_* &\geq \min_{\vy \in \widetilde{\cY}^{-\delta_2}} (\vlam')^\top (\mat{X}^\top \mA) \vy - \frac{\delta_1}{2}\\
                &\geq -\xi - \frac{\delta_1}{2}, \numberthis{eq:sep_compr_primal}
        \]
        where the first inequality follows because, by Lemma~\ref{lem:eps-project-closeness}, there exists $\vy' \in \widetilde{\cY}^{-\delta_2}$ such that $\norm{\vy' - \vy_*}_2 \leq \frac{R_y + \delta}{r_y - \delta} \delta_2$ and, consequently,
        \[
            |(\vlam')^\top (\mat{X}^\top \mA) (\vy_* - \vy')| \leq \norm{(\vlam')^\top (\mat{X}^\top \mA)}_2 \norm{\vy_* - \vy'}_2 \leq 2 B \frac{R_y + \delta}{r_y - \delta} \delta_2 \leq \frac{\delta_1}{2}.
        \]
        Thus, combining \eqref{eq:sep_compr_secondcase} with \eqref{eq:sep_compr_primal} we get
        \[
            \vlam^\top (\mat{X}^\top \mA \vy_*) < (\vlam')^\top (\mat{X}^\top \mA \vy_*) + \delta_2 + \frac{\delta_1}{2} \leq (\vlam')^\top (\mat{X}^\top \mA \vy_*) + \delta_1,
        \]
        for all $\vlam' \in \cK_2$.
        Additionally, $\norm{\mat{X}^\top \mA \vy_*}_\infty > (\xi - \delta_2) \frac{1}{L (1 + \gamma)}$, because $\norm{\vlam}_2 \leq 1 + \gamma$ and we would otherwise have a contradiction in \eqref{eq:sep_compr_secondcase}. Thus, for the normalized vector $\vc = (\mat{X}^\top \mA \vy_*) / \norm{\mat{X}^\top \mA \vy_*}_\infty$ it holds
        \[
             \vlam^\top \vc &< (\vlam')^\top \vc + \frac{\delta_1}{\xi - \delta_2} L (1 + \gamma)\\
                &\leq (\vlam')^\top \vc + \frac{\xi \beta / (\beta + 1)}{\xi / (\beta + 1)}\\
                &\leq (\vlam')^\top \vc + \beta
        \]
        for all $\vlam' \in \cK_2$.
        In other words, $\vc = (\mat{X}^\top \mA \vy_*)$ constitutes a weak separating hyperplane of $\vlam$ from $\cK_2$.
    \end{itemize}
    
    \paragraph{Computing the final solution.}
    Using the weak separation oracle from the previous step, as well as the certificate that $\cK_2^{-\gamma/2}$ is non-empty, we can run the ellipsoid method in polynomial time and compute a weakly feasible solution $\vlam_* \in \cK_2^{+\gamma/2}$.

    The set $\Delta^L$ is a rational polytope, so using Lemma~\ref{lem:polytope-projection}, we can compute an exact projection $\hat{\vlam}_* = \Pi_{\Delta^L}(\vlam_*)$. Since $\vlam_* \in \cK_2^{+\gamma/2} \subseteq (\Delta^L)^{2\gamma}$, it follows that $\norm{\hat{\vlam}_* - \vlam_*}_2 \leq 2 \gamma$. Consequently,
    \[
        \hat{\vlam}_*^\top (\mat{X} \mA) \vy \geq \vlam_*^\top (\mat{X} \mA) \vy - 2 \gamma \sqrt{L} B R_y, \quad \forall \vy \in \widetilde{\cY} \numberthis{eq:eah_lambda_proj}
    \]
    Finally, we show that $\vx_* = \mat{X} \hat{\vlam}_*$ is indeed an $\epsilon$-approximate solution to the initial problem. Since $\vx_k \in \cX$ and $\hat{\vlam}_* \in \Delta^L$ it immediately follows that $\vx_* \in \cX$. It remains to show that $\vx_*$ satisfies the other constraints of \eqref{eq:primal_cp}. Let $\vlam \in \cK_2$ be the closest feasible solution to $\vlam_*$. Then we have
    \[
        \min_{\vy \in \cY} \vx_*^\top \mA \vy &\geq \min_{\vy \in \cY^{-\delta}} \vx_*^\top \mA \vy - B \frac{R_y}{r_y} \delta & \text{(Cauchy–Schwarz and Lemma~\ref{lem:eps-project-closeness})}\\
            &\geq \min_{\vy \in \widetilde{\cY}} \vx_*^\top \mA \vy - B \frac{R_y}{r_y} \delta & \text{($\widetilde{\cY} \supseteq \cY^{-\delta}$)}\\
            &\geq \min_{\vy \in \widetilde{\cY}^{-\delta}} \vx_*^\top \mA \vy - 2 \frac{R_y + \delta}{r_y - \delta} B \delta & \text{(Cauchy–Schwarz and Lemma~\ref{lem:eps-project-closeness})}\\
            &\geq \min_{\vy \in \widetilde{\cY}^{-\delta}} \vx_*^\top \mA \vy - \xi\\
            &= \min_{\vy \in \widetilde{\cY}^{-\delta}} \hat{\vlam}_*^\top (\mat{X} \mA) \vy - \xi\\
            &\geq \min_{\vy \in \widetilde{\cY}^{-\delta}} \vlam_*^\top (\mat{X} \mA) \vy - 2 \gamma \sqrt{L} B R_y - \xi & \text{(from \eqref{eq:eah_lambda_proj})}\\
            &\geq \min_{\vy \in \widetilde{\cY}^{-\delta}} \vlam^\top (\mat{X} \mA) \vy - 3 \gamma \sqrt{L} B R_y - \xi & \text{(Cauchy-Schwarz)}\\
            &\geq -\xi - 3 \gamma \sqrt{L} B R_y - \xi & \text{($\vlam \in \cK_2$)}\\
            &\geq -\epsilon
    \]

    This completes the proof.
\end{proof}

\subsection{Computing an Equilibrium in Convex Games with Weak Oracles}

First, we state the weak version of Lemma~\ref{lem:eahutilitymat-strong}, which is proven in the same way.

\begin{restatable}{lemma}{lemeahutilitymat}
    Let $\cP_1, \dots, \cP_n$ be the compact convex strategy sets of an $n$-player convex game and define $\cX_i := \set{1} \times \cP_i^{+\eta}$. For any $\vx \in \cX := \co\{\cX_1 \otimes \dots \otimes \cX_n\}$ and $\vy = (1, \phi_1, \dots, \phi_n) \in \cY := \set{1} \times \Phi(\cP_1, \eta) \times \dots \times \Phi(\cP_n, \eta)$, it holds that
    \[
        \vx^\top \mA \vy = \sum_{i=1}^n \E_{\vs \sim \vx}[u_i(\vs) - u_i(\phi_i(\vs_i), \vs_{-i})]
    \]
\end{restatable}

Finally, we prove the weak version of Theorem~\ref{thm:eah-equil-computation-strong}. To prove it, we need to (i) tighten the precision parameters in the application of the Ellipsoid Against Hope framework, because we can only use the Approximate Identity (Lemma~\ref{lem:approx-identity}) instead of the true identity transformation, and (ii) carefully deal with the added constant dimension in the set of deviations by using an inflation trick to make the set well-bounded.

\begin{restatable}{theorem}{thmequilcomputationweak}\label{thm:eah-equil-computation-weak}
    Let $G$ be an $n$-player convex game with compact convex strategy sets $\cP_i \subset \bbR^{d_i}$ for $i \in [n]$, given through a weak separation oracle, that are well-bounded as $\cB_{d_i}(r_i) \subseteq \cP_i \subseteq \cB_{d_i}(\vzero, R_i)$. Furthermore, assume that $G$ satisfies the polynomial utility gradient property (Assumption~\ref{asmpt:polynomial-utility-gradient}) for the set relaxations $\cX_i = \set{1} \times \cP_i^{+\eta}$, and $u_i(\vs) \in [-1, 1]$ for every strategy profile $\vs \in \cX_1 \times \dots \times \cX_n$. Then, there exists an oracle-polynomial time algorithm that computes an $(\epsilon, \eta)$-approximate linear correlated equilibrium. Furthermore, the computed equilibrium is represented as a mixture of polynomially many product distributions over strategy profiles.
\end{restatable}
\begin{proof}
    Let $\delta := \min(\epsilon, \eta) \frac{1}{2 n} \min_i \frac{r_i}{5 R_i}$.
    We define the Correlator-Deviator game with strategy sets $\cX := \co\{\cX_1 \otimes \dots \otimes \cX_n\} \subset \bbR^M$ and $\cY := \set{1} \times \Phi(\cP_1, \delta) \times \dots \times \Phi(\cP_n, \delta) \subset \bbR^N$ and utility matrix $\mA$. Now assume that we have an optimal solution $\vx_*$ to the problem
    \[
        \numberthis{eq:linear_eah_primal_cp} \text{find}\ &\vx \in \cX\\
        \text{s.t.}\ &\min_{\vy \in \cY} \vx^\top \mA \vy \geq -\delta.
    \]
    We first show that this $\vx_*$ defines a valid $(\epsilon, \eta)$-approximate linear equilibrium or, in other words, that the individual players' expectations are close to $0$. To this end, observe that since \eqref{eq:linear_eah_primal_cp} is satisfied for all $\vy \in \cY$, it must also be satisfied for $\vy = (1, \phi_1', \dots, \phi_k, \phi_{i+1}', \dots, \phi_n')$ for all $\phi_k \in \Phi(\cP_k, \delta) \supseteq \Phi(\cP_k, \eta)$, where $\phi_j'$ for $j \neq k$ are the Approximate Identities from Lemma~\ref{lem:approx-identity}.
    Consequently, for this $\vy$,
    \[
        -\frac{\epsilon}{2} &\leq -\delta \leq \vx_*^\top \mA \vy\\
        &= \E_{\vs \sim \vx_*}[u_k(\vs) - u_k(\phi_k(\vs_k), \vs_{-k})] + \sum_{j \neq k} \E_{\vs \sim \vx_*}[u_j(\vs_j - \phi_j'(\vs_j), \vs_{-j})]\\
        &\leq \E_{\vs \sim \vx_*}[u_k(\vs) - u_k(\phi_k(\vs_k), \vs_{-k})] + \sum_{j \neq k}\frac{5R_j}{r_j} \delta\\
        &\leq \E_{\vs \sim \vx_*}[u_k(\vs) - u_k(\phi_k(\vs_k), \vs_{-k})] + \frac{\epsilon}{2}.
    \]
    And the desired individual inequality follows by rearranging,
    \[
        \E_{\vs \sim \vx_*}[u_k(\vs) - u_k(\phi_k(\vs_k), \vs_{-k})] \geq -\epsilon.
    \]
    We conclude that an optimal $\vx_*$ defines an $(\epsilon, \eta)$-approximate linear correlated equilibrium.
    
    At this point we could try applying Theorem~\ref{thm:ellipsoid-against-hope}. However, $\cY$ is not full-dimensional as is required by the Theorem conditions. To alleviate this problem we can instead work with the set $\cY' = [1 - \gamma, 1 + \gamma] \times \Phi(\cP_1, \delta) \times \dots \times \Phi(\cP_n, \delta)$ for $\gamma = \delta / (2 \sqrt{N})$. This set is well-bounded as $\cB_N(r_y') \subseteq \cY' \subseteq \cB_N(\vzero, R_y')$ with $r_y' = \min(\gamma, \min_{i \in [n]} \frac{r_i}{4 R_i})$ and $R_y' = \gamma + \sum_{i \in [n]} \frac{3R_i}{r_i} \sqrt{R_i^2 + d}$. 
    Additionally, for any distribution $\vx \in \cX$, it holds that $\norm{\vx^\top \mA}_2 \leq \sqrt{N} =: B$, because $U_i[k] \in [-1, 1]$ for all $k$.
    
    If we can also construct an oracle that for all $\vy \in \cB_N(\vzero, R_y')$ either produces a weak separating hyperplane of $\vy$ from $\cY'$ or a weak good-enough-response, then we are done. We can compute any feasible solution $\vx_* \in \cX$ such that $\min_{\vy \in \cY'} \vx^\top \mA \vy \geq -\delta$. Then, the exact same solution $\vx_*$ will also be feasible for our problem \eqref{eq:linear_eah_primal_cp}, because $\cY \subset \cY'$. It remains to see how to construct the required weak oracle.

    Let $\vy \in \bbR^N$ be an input point and $\beta \in (0, r_y')$ a precision parameter.
    If $y_1 \notin [1-\gamma, 1+\gamma]$, then we return the hyperplane defined by $\va := (-\text{sign}(y_1), \vzero)$, which satisfies $\va^\top \vy < \va^\top \vy', \quad \forall \vy' \in \cY'$.

    For the rest of the proof, it is $y_1 \in [1-\gamma, 1+\gamma]$. Let $\phi_1, \dots, \phi_n$ be the linear transformations corresponding to $\vy$. We apply the weak semi-separation oracle of Lemma~\ref{lem:sep-non-fixed-point} to each one of the $\phi_i$ transformations on the full-dimensional sets $\cP_i$ with precision $\beta' := \min(\delta, \beta / n, \delta^2 / (2 \sum_{i=1}^n \sqrt{d_i+1})^2)$. If any semi-separation oracle $j$ responds with a weak separating hyperplane $\vc \in \bbR^{d_j \times (d_j+1)}$ of $\phi_j$ from $\Phi(\cP_j, \delta)$, then we return $\va := (0, \vzero, \dots, \vc, \vzero, \dots, \vzero)$ which satisfies $\va^\top \vy < \va^\top \vy' + \beta' \quad \forall \vy' \in (\cY')^{-\beta'}$ which is a valid weak separating hyperplane since $\beta' < \beta$.

    Otherwise, we are in the case where we have computed $\beta'$-approximate fixed points $\vx_i^* \in \cP_i^{+\beta'} \subset \cP_i^{+\delta}$ such that $\norm{\phi_i(\vx_i^*) - \vx_i^*}_2^2 \leq \beta'$ for all players $i \in [n]$. In this case, similarly to \citet{Farina2024:eah}, we can construct a good-enough-response by considering the product distribution $\vx = (1, \vx_i^*) \otimes \dots \otimes (1, \vx_i^*)$. Now let $\hat{\vy} = (1, \phi_1, \dots, \phi_n)$ such that $\norm{\hat{\vy} - \vy}_2 \leq \gamma$.
    By following the exact same steps as in \citet[Lemma 4.4]{Farina2024:eah}, we conclude that
    \[
        \vx^\top \mA \hat{\vy} = \sum_{i=1}^n \vec{g}_i(\vx_{-i}) \cdot [\vx_i^* - \phi_i(\vx_i^*)],
    \]
    where $\vec{g}_i(\vx_{-i}) = \E_{\vs_{-i} \sim \vx_{-i}}[\nabla u_i(\vx_{-i})]$ is given by Assumption~\ref{asmpt:polynomial-utility-gradient}. Then, by the Cauchy–Schwarz inequality,
    \[
        |\vx^\top \mA \hat{\vy}| &\leq \sum_{i=1}^n \norm{\vec{g}_i(\vx_{-i})}_2 \norm{\vx_i^* - \phi_i(\vx_i^*)}_2\\
            &= \sum_{i=1}^n \sqrt{d_i} \sqrt{\beta'}\\
            &\leq \delta / 2.
    \]
    Additionally, $|\vx^\top \mA (\hat{\vy} - \vy)| \leq \norm{\vx^\top \mA}_2 \norm{\hat{\vy} - \vy}_2 \leq B \gamma \leq \delta / 2$. Thus,
    \[
        \vx^\top \mA \vy \geq \vx^\top \mA \hat{\vy} - \delta/2 \geq -\delta,
    \]
    which constitutes a good-enough-response and can be computed in oracle-polynomial time, alongside $\vx^\top \mA$. Thus, all conditions of Theorem~\ref{thm:ellipsoid-against-hope} are satisfied and we can use the Ellipsoid Against Hope framework to compute an $(\epsilon, \eta)$-approximate linear correlated equilibrium that is a mixture of polynomially many product distributions over strategy profiles.
\end{proof}

\end{document}